%% file: main.tex
\documentclass[12pt,a4paper,onecolumn]{IEEEtran}
\IEEEoverridecommandlockouts
\input{packages}

\input{defis}

\usepackage{tikz}
\usetikzlibrary{arrows.meta, positioning}
\usepackage[utf8]{inputenc} 
\usepackage[T1]{fontenc}
\usepackage{url}              % provides \url{...}
\usepackage{cite}             % improves presentation of citations
\usepackage{multirow}
\usepackage{wrapfig}

\usepackage{mleftright}       % fix to wrong spacing of \left-,
\mleftright                   % \middle- \right-commands 

\usepackage{graphicx}         % provides \includegraphics{...} to
\usepackage{booktabs}         % fixes poor spacing in tables and
\begin{document}

\title{Error-Minimizing Measurements in Postselected One-Shot Symmetric Quantum State Discrimination and Acceptance as a Performance Metric
\author{%
	\IEEEauthorblockN{Saurabh Kumar Gupta, Abhishek K. Gupta}
%	\IEEEauthorblockA{\\\textit{Department of Electrical Engineering}\\
%		{Indian Institute of Technology Kanpur},	Kanpur, India\\
%		\{saurabhg20, gkrabhi\}@iitk.ac.in}
\thanks{Saurabh Kumar Gupta and Abhishek K. Gupta are with IIT Kanpur, India, 208016. Email:saurabhg20@iitk.ac.in,gkrabhi@iitk.ac.in. A preliminary work was presented at ISIT 2023\cite{gupta2023characterizing}.}
}}

\maketitle

\begin{abstract}
	Quantum state discrimination refers to the determination of the unknown state of a quantum object that can be in one of two possible states, by performing a certain measurement that declares its state to be one of two states. Hence, it can be expressed as a hypothesis testing problem with null and alternative hypotheses representing the two states. Postselected state discrimination (PSD) considers measurements with an additional outcome, corresponding to not making any determination, thus rejecting the test. Thus, allowing the test not to make a decision in unreliable cases can help PSD reduce the effective error in state determination among cases when a decision is made. However, this comes at the expense of reduced probability of making a decision which we term {\em acceptance}. In this work, we focus on deriving optimal measurements that maximize acceptance while minimizing the postselected error. We first give the set of all measurements that can achieve the minimum error, along with a method to parameterize these error-minimizing measurements. We further provide an example to demonstrate that these measurements vary in acceptance and thus establish the need to choose a measurement that minimizes error and maximizes acceptance simultaneously. We derive the expression of acceptance for an arbitrary error-minimizing measurement in terms of the parameters used for parameterization proposed earlier. We find the maximum value of acceptance achievable over the set of all error-minimizing measurements. We also provide an example measurement that achieves the maximum acceptance.
\end{abstract}

%\listoffigures

\textit{Index Terms-} quantum state discrimination, hypothesis testing, measurement operator, type-1 and type-2 error.

\section{Introduction}
Quantum hypothesis testing involves discriminating between hypotheses corresponding to the quantum properties of nature, with numerous applications in quantum information science \cite{watrous2018theory,hayashi2016quantum,wilde2013quantum}. A specific instance of quantum hypothesis testing is quantum state discrimination, which aims to determine the state of a given quantum object that can exist in one of two possible states ($\rho$ and $\sigma$), based on a measurement on the object. 
These states can be regarded as the null and the alternative hypotheses. If the decision has to be made based on the measurement of a single copy of the quantum object, it is referred to as one-shot hypothesis testing. If infinitely many copies are available, it is called asymptotic hypothesis testing. When the test falsely concludes the alternative hypothesis, it is called type-1 error; when the test falsely decides in favor of the null hypothesis, it is called type-2 error. Ideally, the goal is to select a measurement and a decision rule with both type-1 and type-2 errors arbitrarily small. Since simultaneous minimization of the two errors is difficult, a trade-off is usually adopted. For example, in asymmetric hypothesis testing, the objective is to minimize the type-2 error under a constraint on the type-1 error. On the other hand, symmetric hypothesis testing aims to minimize the average error probability.
%\newpage

In one-shot discrimination of quantum states, the minimum error achievable in an asymmetric hypothesis testing problem can be obtained using semi-definite programming. In symmetric hypothesis testing, the error can be expressed in a closed-form expression utilizing the Helstrom-Holevo theorem \cite{helstrom1969quantum,holevo1972analogue}. In the asymptotic case, the minimum error in quantum hypothesis testing vanishes exponentially with the number of copies of states, and the error exponent was characterized in \cite{nussbaum2009chernoff,audenaert2007discriminating} for the symmetric case and in  \cite{hiai1991proper,ogawa2005strong} for  the asymmetric case. In particular, quantum relative entropy\cite{umegaki1962conditional} was shown to be the error exponent of asymmetric hypothesis testing. Further, \cite{hayashi2002optimal} presents the optimal measurements that can achieve this error. In conventional quantum hypothesis testing, as described above, one of the two hypotheses is selected. The key problem remains that non-orthogonal states cannot be distinguished with zero error via conventional quantum hypothesis testing.

{\Blue

	To distinguish the non-orthogonal states with zero-error, prior work \cite{fiuravsek2003optimal,chefles1998unambiguous,feng2004unambiguous,stojnic2007unambiguous,bandyopadhyay2014unambiguous,gupta2024unambiguous,zhang2023unambiguous,zhou2012unambiguous,chefles1998quantum,jafarizadeh2008optimal,kleinmann2010unambiguous,hayashi2008state,hashimoto2010unitary,ivanovic1987differentiate,dieks1988overlap,peres1988differentiate,chefles1998strategies,rudolph2003unambiguous,croke2006maximum,herzog2009discrimination,bagan2012optimal,zhuang2020ultimate} have proposed to include an extra outcome in the quantum hypothesis testing where no state is declared, effectively, rejecting both hypotheses, or equivalently declaring a failure of the test. Such discrimination strategies are termed unambiguous state discrimination (USD) (see \cite{chefles1998unambiguous,feng2004unambiguous,stojnic2007unambiguous,bandyopadhyay2014unambiguous,gupta2024unambiguous,zhou2012unambiguous,jafarizadeh2008optimal}). For an example, consider a pair of states $\{\rho,\sigma\}$. USD aims to design a measurement such that only possible outcomes are (i) declare $\rho$, and (ii) reject the test, when the state is $\rho$ and; the possible outcomes are (i) declare $\sigma$ and (ii) reject the test when the state is $\sigma$. This is illustrated in Fig. \ref{fig: diff} (b). In other words, in a USD, the measurement either detects the correct state or declare a decision failure, resulting in a zero error among accepted (selected) results. However, it is important to note that there is a certain probability that a decision is not made, which can be termed failure or test-rejection probability. Therefore, it is desirable to have a low failure probability while ensuring zero error in USD. While desirable, it is noteworthy that USD is not always possible, as its feasibility depends on the state pair $\{\rho, \sigma\}$ we aim to distinguish. The prior work in USD has focused on determining the condition on the state pair for USD to be feasible, determining the USD measurement under the above condition, and determining USD measurements with the minimum probability of rejection, along with derivation of specific bounds on the rejection probability \cite{chefles1998unambiguous,feng2004unambiguous,stojnic2007unambiguous,bandyopadhyay2014unambiguous,gupta2024unambiguous,jafarizadeh2008optimal}. In Table \ref{table: sumContri}, we summarize related papers from the past literature along with their system model and  relevant contributions.  Readers are suggested to refer to  \cite{barnett2009quantum}, and \cite{bae2015quantum} for a comprehensive survey of various quantum state
	discrimination strategies. 

\input{contriTable.tex}

As mentioned above, USD is not always feasible. The work in \cite[Theorem 2]{feng2004unambiguous} showed that USD is feasible only when the two states in the pair of states have unequal support. Note that the support $\Pi_\nu$  of a state $\nu$ refers to the projection operator onto the eigenspace spanned by the set of eigenvectors corresponding to all nonzero eigenvalues of $\nu$. It is possible to generalize USD further to include the cases when the supports of the pair of states are equal. One such generalization is postselected discrimination (PSD), which utilizes measurements that can give one of the two states or {\em reject} as its outcome. Here, the reject outcome corresponds to not selecting any of the hypotheses (see Fig. \ref{fig: diff} (c)) and declaring a failure of the test \cite{regula2022postselected}. In other words, the decision is only made when the third outcome does not occur, which can be collectively termed {\em selection of the test}. Hence, such discrimination is termed postselected discrimination. Note that in contrast to USD, in PSD, a decision can be either incorrect or correct, resulting in a certain error post selection. 

\begin{figure}[t]
	\begin{center}
		\includegraphics[scale=0.6]{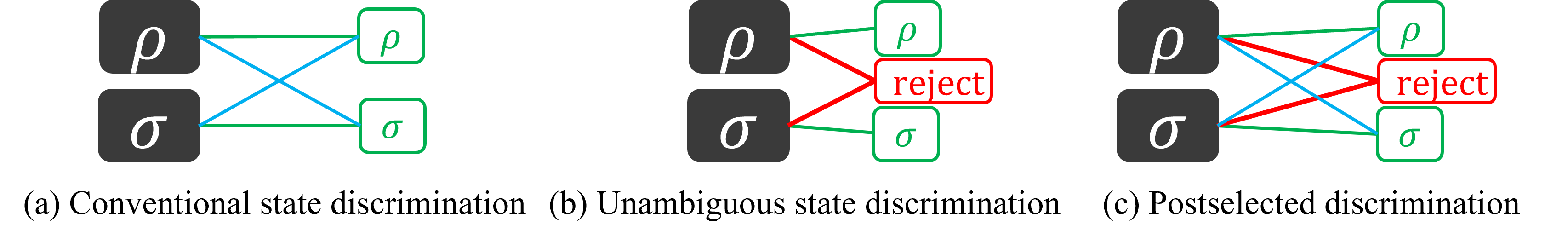}
	\end{center}	
	\vspace*{-5mm}
	\caption{\Blue In all three subfigures, the black boxes on the left represent an unknown object in the state $\rho$ or $\sigma$. The colored boxes on the right represent the possible outcomes, where the green box corresponds to the accepted outcome, and the red box corresponds to the rejected outcome. The green line represents correct detection, the blue line shows incorrect detection, and the red line shows rejecting the test. Note that, in conventional hypothesis testing, only accepted outcomes were considered. Furthermore, from (b), it is clear that outcomes corresponding to blue lines are not allowed in USD.}
	\label{fig: diff}
\end{figure}

		The postselected probability of any event is defined as the probability of the event conditioned on the observation of one of the selected outcomes. Hence, in the symmetric PSD, the postselected symmetric error (See \eqref{eq: defSymError}) is defined as the probability of getting an incorrect outcome (corresponding to blue lines in Fig. \ref{fig: diff} (c)) given that the test is selected (corresponding to both blue and green lines, which merge to green boxes, in Fig. \ref{fig: diff} (c)). Recall that the event that the test is selected corresponds to the case when a decision is made and includes getting a correct or incorrect outcome. The minimum achievable value of postselected symmetric error is given in \cite{regula2022postselected}, along with an example measurement that achieves this value. However, a general characterization of all measurements that can achieve this error is not given in the prior work.
		
		It is reasonable to wonder if, in postselected hypothesis testing,  postselected symmetric error is the only performance metric and what constitutes an optimal measurement. To understand this, let us consider an example where two distinct measurements have the same postselected symmetric error, but different rejection probabilities. Here, rejection probability refers to the probability of obtaining the reject outcome, i.e., the probability of declaring the test failure and not making a decision. Essentially, even though the error is the same, one should prefer the one with a lower rejection probability. Further, a measurement that achieves the minimum value of error may have a high rejection probability. Such measurements are practically unusable. This possibility raises the question of how we can find measurements that minimize the probability of rejection and the postselected error. To address this concern, in our preliminary work \cite{gupta2023characterizing}, we define another performance metric, {\em acceptance}, as the probability that one of the first two outcomes occurs and a decision is made to select one of the two hypotheses. Hence, the acceptance for any state is given by the probability of observing an accepted outcome, given that the quantum object is in that state. In particular, acceptance for state $\rho$ (or $\sigma$) is the probability of events corresponding to green and blue lines originating from states $\rho$ (or $\sigma$, depending on the state under consideration), respectively, in Fig. \ref{fig: diff} (c). It is intuitive to find a measurement that maximizes acceptance while minimizing the postselected error, which is the primary focus of the work.  
		
}

{\Blue

\textit{Contributions:} In this work, we consider symmetric postselected discrimination (PSD) for a pair of mixed states, $\rho$ and $\sigma$. We first consider the case where the two states have equal support, making USD infeasible. Hence, the PSD will result in a non-zero postselected error. We first determine the set of all measurements that minimize the postselected symmetric error. In particular, we represent error-minimizing measurements in a parametric form, allowing for a complete characterization of these measurements. Such parameterization represents a novel contribution compared to past literature, as it enables us to further refine the set of error-minimizing measurements to optimize any other performance metrics of interest. In particular, compute the maximum possible value of acceptance and provide an example for a measurement that achieves the maximum acceptance. This part of the work is entirely different than USD and is unrelated to the prior literature. 

Next, we study USD as a special case of PSD where the supports of the two states are not equal. As discussed before, in this case, the postselected error vanishes, and therefore, the probability of a correct decision (also known as the success probability) is simply given as the expected acceptance. We first present the set of all error-minimizing (zero-error here) measurements along with their parameterized form, which satisfy the conditions given for USD measurements in \cite{chefles1998unambiguous,feng2004unambiguous,stojnic2007unambiguous,bandyopadhyay2014unambiguous,jafarizadeh2008optimal,kleinmann2010unambiguous}. While most of these prior works have focused on maximizing success probability or obtaining general bounds  on it (see \cite{chefles1998unambiguous,feng2004unambiguous,stojnic2007unambiguous,bandyopadhyay2014unambiguous,gupta2024unambiguous,jafarizadeh2008optimal}), we provide closed-form expressions of the maximum acceptance for each state separately and show achievability by presenting one example of a measurement that achieves this maximum value of acceptance. The key contributions of our paper can be summarized as follows.

	\begin{enumerate}
		\item \textbf{Characterization of the set of error-minimizing measurements -} 
		The main contribution of our work lies in determining the set of all error-minimizing measurements (i.e., measurements that can achieve the minimum value of the postselected symmetric error) in PSD and providing a parameterization to express these measurements for their complete characterization, which was not available in previous works. While the prior work \cite{regula2022postselected} has presented the minimum value of postselected symmetric error in PSD, along with an example of such error-minimizing measurement,  it does not provide a list of all error-minimizing measurements. It is not trivial to obtain this list, as  \cite{regula2022postselected} utilized the Lagrange duality approach to derive the minimum error, which does not provide any pointers towards finding the set of all optimal points. Finding the exact constraint on such measurement operators was a challenge, which our work resolved. Characterizing the complete set is essential to present any discussion on the optimality of measurements. To characterize error-minimizing measurements, we begin with noticing that any measurement that achieves the supremum given in \eqref{eq: cricial2ndLastSymBase} minimizes the error. We focus on finding the necessary and sufficient conditions for a measurement to achieve this supremum and thus obtain the complete set of measurements achieving the minimum error, as shown in the proof of Theorem \ref{lem: conditionSym}, which is given in Appendix \ref{proof: conditionSym}. While deriving, we also obtain an alternative and novel derivation of the minimum error given in \cite{regula2022postselected}, which we have included  in the same theorem which further validates our results.  
		
		\item \textbf{Construction: unique and appropriate parameterization -} We provide a construction of error-minimizing measurements based on parameters that can be varied appropriately to obtain all such measurements. A measurement constructed in this way is always an error-minimizing measurement. Furthermore, the parameterization is unique for each error-minimizing measurement in the sense that, changing the value of any of the parameters gives a different measurement. One key benefit of having such characterization is that we can refine these measurements to optimize other performance metrics, such as acceptance. Acceptance is a key metric proposed and maximized in this work for error-minimizing measurements. The parameterization provided us with a closed-form expression of acceptance, thus offering a direct approach to finding the maximum value of acceptance. The presented parameterization enables us to design an error-minimizing measurement to achieve any desired value of acceptance for each state. This observation is added as Remark \ref{rem: desiredAccEq} and Remark \ref{rem: desiredAccNeq}. The parameterization also helped us in finding the maximum acceptance easily, which is evident from the closed-form expression of the maximum acceptance in many cases.
		
		\item \textbf{A novel technique: subspace transformation for operators -} In Theorem \ref{lem: conditionSym}, we obtained the constraint on an arbitrary measurement given by a positive operator valued measure (POVM) $\Lambda = \{\Lambda_\rho,\Lambda_\sigma,\mathrm{I}-\Lambda_\rho-\Lambda_\sigma\}$ for it to be error-minimizing. We get these constraints on $\sigma^{1/2}\Lambda_\rho\sigma^{1/2}$, and $\sigma^{1/2}\Lambda_\sigma\sigma^{1/2}$ as
		$\sigma^{1/2}\Lambda_\rho\sigma^{1/2}\in\mathcal{P}(\Pi_1) \text{ and } \sigma^{1/2}\Lambda_\sigma\sigma^{1/2}\in\mathcal{P}(\Pi_2),$
		where $\Pi_1$ and $\Pi_2$, obtained in Theorem \ref{lem: conditionSym}, correspond to the projection operator on a pair of orthogonal subspaces. $\mathcal{P}(\Pi)$ denotes the set of all positive semidefinite operators having eigenvectors (corresponding to non-zero eigenvalues) lying in the subspace corresponding to the projection operator $\Pi$. At this stage, if we proceed in a conventional sense, we would have written an optimization problem as maximizing acceptance over the set of measurements satisfying these constraints. With that approach, unique parameterization would have been quite complicated. To simplify it further, we came up with a novel subspace-transformation that simplifies the constraints. We first observe that with these constraints on the operator, with some transformation, we can obtain another projection operator $\Pi_3$ and $\Pi_4$ such that the condition $\Lambda_\rho\in\mathcal{P}(\Pi_3)$, and $\Lambda_\sigma\in\mathcal{P}(\Pi_4)$ gives equivalent constraints. The equivalence is shown in Appendix \ref{sec: genProj}. To the best of our knowledge, this transformation, proved in Theorem \ref{lem: genMainProj} and Theorem \ref{thm: projSubSet} of Appendix \ref{sec: genProj}, has not been shown and used previously in this context. Using the proposed transformation and the constraints given in Theorem \ref{lem: conditionSym}, we directly obtain the set of all error-minimizing measurements, which is given in Theorem \ref{lem: setSym}. Therefore, we have obtained a parameterization of $\Lambda_\rho$, and $\Lambda_\sigma$ from Theorem \ref{lem: setSym}.
		
		\item \textbf{Observations: properties of error minimizing measurements -} After finding constraints on error minimizing measurements, we present the properties of error minimizing measurements in Corollary \ref{cor: alwaysOnlyOne} and Corollary \ref{cor: equalityCondition}. One key observation provided is that, for the case with $\rho$ and $\sigma$ having equal support (i.e., $\Pi_\rho=\Pi_\sigma$), each any error minimizing measurement has one of the following three properties depending on the value of prior probabilities $p_\rho$ and $p_\sigma$.
		\begin{enumerate}[label = (\roman*)]
			\item If $p_\rho>p_\rho^*$ (i.e., $p_\sigma<p_\sigma^*$): Error-minimizing measurement never detects $\sigma$. It either detects $\rho$ or rejects the test.
			\item If $p_\rho<p_\rho^*$ (i.e., $p_\sigma>p_\sigma^*$): Error-minimizing measurement never detects $\rho$. It either detects $\sigma$ or rejects the test.
			\item If $p_\rho=p_\rho^*$ (i.e., $p_\sigma=p_\sigma^*$): There are error-minimizing measurements detecting either/both the hypotheses.
		\end{enumerate}
		Here, $p_\rho^*$ and $p_\sigma^*$ are certain constants given by (16).  An exposition is provided in Fig \ref{fig:mostlyOne}.
		
		\item \textbf{PSD as a generalization of USD -}
		While USD is feasible only if the two subspaces where each state lies are not equal (see Theorem 2 in \cite{feng2004unambiguous}), PSD is a more general framework allowing discrimination with non-zero error. We also study USD as a special case of PSD for the case when the supports of the two states are not equal. We first present the set of all error-minimizing (which becomes zero-error in USD) measurements along with their parameterized form, which satisfy the conditions given for USD measurements in \cite{chefles1998unambiguous,feng2004unambiguous,stojnic2007unambiguous,bandyopadhyay2014unambiguous,jafarizadeh2008optimal}. 
		
	\end{enumerate}

}

The notation needed across the paper is given below.

\textit{Notation}: 
Let $\hilbert$ denote a Hilbert's space. Let $\mathcal{L}(\hilbert)$, $\mathcal{R}(\hilbert)$, $\mathcal{P}(\hilbert)$ and $\mathcal{D}(\hilbert)$ respectively denote the set of all linear, Hermitian, positive semidefinite and density operator over Hilbert's space $\hilbert$. $\I\in\mathcal{P}(\hilbert)\subset\mathcal{R}(\hilbert)\subset\mathcal{L}(\hilbert)$ is identity transformation. $\nu_1-\nu_2\in\mathcal{P}(\hilbert)$ is also written as $\nu_1\geq\nu_2$ or $\nu_2\leq\nu_1$. Any operator $\Pi\in\mathcal{P}(\hilbert)$ is a projection operator if $\Pi\Pi=\Pi$. For any $\nu\in\mathcal{R}(\hilbert)$, $\Pi_{\nu}$ denotes projection operator onto the eigenspace spanned by the set of eigenvectors corresponding to all non-zero eigenvalues of $\nu$. Note that by the properties of projector $\Pi_\nu\nu\Pi_\nu=\nu$. $\Pi_{\nu}^{\max}$ and $\Pi_{\nu}^{\min}$ denote the projection operator onto eigen-space corresponding to the maximum and smallest non-zero eigenvalues of $\nu$, denoted as $\|\nu\|_{\infty}$ and $\|\nu\|_{\infty,0}$ respectively. For some projection operator $\Pi\in\cp(\ch)$, $\mathcal{P}(\Pi)\subset\mathcal{P}(\hilbert)$ and $\mathcal{S}(\Pi)\subset\mathcal{D}(\hilbert)$ denote set of all positive semidefinite operators and density operators respectively that are invariant w.r.t. $\Pi$ i.e. $\mathcal{P}(\Pi)\isdefinedas\{\nu:\Pi\nu\Pi=\nu ,\nu\in\mathcal{P}(\hilbert)\}$ and  $\mathcal{S}(\Pi)\isdefinedas\{\nu:\Pi\nu\Pi=\nu ,\nu\in\mathcal{D}(\hilbert)\}$. For any $\nu\in\mathcal{P}(\hilbert)$, $\nu^{-1}$ is an operator obtained by substituting all non-zero eigenvalues of $\nu$ by their inverse. Note that, with this definition of inverse, we get $\nu^{-1}\nu=\Pi_{\nu}$. Given a pair of operators $\nu_1,\nu_2\in \mathcal{P}(\hilbert)$, we denote $R_{\max}(\nu_1,\nu_2)=\|\nu_2^{-1/2}\nu_1\nu_2^{-1/2}\|_{\infty}$ and  $R_{\min}(\nu_1,\nu_2)=\|\nu_2^{-1/2}\nu_1\nu_2^{-1/2}\|_{\infty,0}$. $R_{\max}$ is related to the max relative entropy of the states, defined as the log of the maximum eigenvalue of $\nu_2^{-1/2}\nu_1\nu_2^{-1/2}$ \cite{datta2009min}.

\textit{Organization: }The paper is organized as follows. \textbf{Section \ref{sec: System}} describes the system model, the problem statement describing PSD between the two states $\rho$ and $\sigma$ as a hypothesis problem, and the minimum possible postselected error from the literature, followed by definitions of acceptance, the set of error-minimizing measurements and the maximum acceptance achieved over this set. {\Blue \input{flow.tex}

}

{\Blue
	
		\textbf{Section \ref{sec: rhoEqsigma}} studies the scenario when the supports of the two states $\rho$ and $\sigma$ are equal, i.e., $\Pi_\rho = \Pi_\sigma$ and presents Theorem \ref{lem: conditionSym}, \ref{lem: setSym}, \ref{lem: genAccSym}, and \ref{thm: symmetric}. A flow-chart illustrating the contributions of these theorems is included in Fig. \ref{fig: thmFlow}.

	 \begin{enumerate}[label={\textit{Theorem {\arabic*}}}]
		\setlength{\itemindent}{2cm}
		
		\item states the necessary and sufficient condition that any measurement needs to satisfy to achieve the minimum postselected error. We obtain three different conditions in three different cases depending on the prior probability of the unknown state being $\rho$ (or $\sigma$). We state some important observations regarding the error-minimizing measurements in Remark \ref{rem: notesOnT}, \ref{rem: notesOnP}, and \ref{rem: betterMetric} and Corollary \ref{cor: alwaysOnlyOne}-\ref{cor: equalityCondition}. Example \ref{exmp: first} illustrates how to obtain error-minimizing measurements utilizing Theorem \ref{lem: conditionSym} for a pair of states to be discriminated and further shows that acceptance varies over error-minimizing measurements. Example \ref{exmp: q1} shows that one can obtain the most general error-minimizing measurements from the conditions obtained in Theorem \ref{lem: conditionSym}.
		
		\item presents the set $\mathcal{E}_s$ of all error-minimizing measurements in a form that can be parameterized easily for all three cases separately. We derive Theorem \ref{lem: setSym} by finding the general subspace where the measurements must lie, using the conditions obtained in Theorem \ref{lem: conditionSym}. This result is a generalization from the literature in the sense that achievability was shown in the literature \cite{regula2022postselected}, but an exhausting set of all such measurements was not known. Then, we give a parametric method to construct an arbitrary error-minimizing measurement by varying certain parameters. The method for all the three cases is summarized in Table \ref{table: paraEqual}. Example \ref{exmp: q2} illustrates illustrates how to obtain error-minimizing measurements for a pair of quantum states and develops an intuition towards the presented approach to derive the maximum value of acceptance by maximizing one of the parameters.
		
		\item obtains acceptance for an arbitrary error minimizing measurement from the set $\mathcal{E}_s$ which was earlier derived in Theorem \ref{lem: setSym}.
		
		\item derives the maximum value of acceptance over the set of all error-minimizing measurements $\mathcal{E}_s$. It provides a closed-form expression of the maximum acceptance. In one case, we have given it as an optimization problem. Remark \ref{rem: obtainsMaxAccEq} provides an example of a measurement  in $\mathcal{E}_s$ that achieves the maximum acceptance. This part of the work as presented in Section III is entirely different than USD and is unrelated to the prior literature. 
		
		\hspace*{-4mm}\textbf{Section \ref{sec: rhoNeqsigma}} studies USD as a special case of PSD when the supports of the two states $\rho$ and $\sigma$ are not equal, i.e., $\Pi_\rho \neq \Pi_\sigma$. 
		
		\item obtains the set of all error minimizing measurements $\mathcal{E}_s$, in a form that is easy to parameterize. We observe that the error may vanish in three different ways, resulting in three different conditions. We express $\mathcal{E}_s$ as a union of three sets, each satisfying one of the three conditions. We discuss the properties of the measurements in each set. We then show that two sets remain empty when $\Pi_\rho<\Pi_\sigma$ or $\Pi_\rho>\Pi_\sigma$. Further, these measurements are parameterized in Lemma \ref{lem: para1}, \ref{lem: para2}, and \ref{lem: para3}.
		
		\item obtains the acceptance for an arbitrary error minimizing measurement from the set $\mathcal{E}_s$.
		
		\item obtains the maximum acceptance over the set of all error-minimizing measurements $\mathcal{E}_s$ by maximizing the acceptance obtained on the three sets, and the result is summarized in Table \ref{table: maxAccNeq} in the theorem.
\end{enumerate}

\textbf{Section \ref{sec:conclusion}} concludes the paper with a summary and possible generalizations of this work.
}

To maintain the flow, we have just given the key results in the main text while supporting mathematical results are stated and derived in the appendix. Now, we proceed to describe the system model.

\section{System Model} \label{sec: System}
We consider a quantum object that can be in one of the two possible states $\rho$, or $\sigma\in \mathcal{D}(\hilbert)$, with prior probabilities $p_\rho$ and $p_\sigma$, respectively. The objective of PSD is to determine its state with a high accuracy via an appropriately designed measurement. PSD can be seen as a post-selection quantum hypothesis testing problem where the quantum object being in the state $\rho$ corresponds to the null hypothesis, while the quantum object being in the state $\sigma$ corresponds to the alternative hypothesis. A measurement is performed with positive operator-valued measure (POVM) $\Lambda = \{\Lambda_\rho,\Lambda_\sigma, \I-(\Lambda_\rho+\Lambda_\sigma)\}$. Note that $\Lr,\Ls$ must satisfy the condition $\Lambda_\rho,\Lambda_\sigma\in \mathcal{P}(\hilbert)$ and $\Lambda_\rho+\Lambda_\sigma\leq \I$ for $\Lambda$ to be a valid measurement. Hence, the set of all such measurements is given by 
\begin{align}
	\mathcal{M} \isdefinedas \left\{\Lambda:\Lambda=\{\Lambda_{\rho},\Lambda_{\sigma},\I-\Lambda_{\rho}-\Lambda_{\sigma}\},\Lambda_{\rho}\geq0,\Lambda_{\sigma}\geq0, \Lambda_{\rho}+\Lambda_{\sigma}\leq \I\right\}.
\end{align} 
The performed measurement results in one of the three outcomes corresponding to the operators $\Lr$, $\Ls$, or $\I-\Lr-\Ls$. The first outcome (corresponding to the operator $\Lr$) corresponds to accepting the null hypothesis equivalent to declaring the unknown state as $\rho$, the second outcome (corresponding to the operator $\Ls$) corresponds to accepting the alternative hypothesis and declaring it as $\sigma$. If the third outcome (corresponding to the operator $\I-\Lr-\Ls$) occurs, no decision is made, rejecting the test as inconclusive. The first two outcomes collectively represent the selection or acceptance of the test, where a decision is taken to select one of the hypothesis. Various outcomes and their probabilities are summarized in Table \ref{table: probs}.

\begin{table}[H]
	\centering
	\begin{tabular}{|c||c|c|c||c|}
		\hline
		\multirow{2}{*}{Actual state of} &\multicolumn{4}{c|}{Probability of\ldots}\\
		\cline{2-5}
		\multirow{2}{*}{the quantum object}& \multirow{2}{*}{declaring $\rho$}  & \multirow{2}{*}{declaring $\sigma$}  & \multirow{2}{*}{rejection}  &  {accepting null or alternative}\\ %\cline{2-4} 
		& & & & hypothesis i.e. $\rho$ or $\sigma$\\
		\hline\hline
		\multirow{2}{*}{$\rho$} & \multirow{2}{*}{$\tr\lr{\Lr\rho}$} & \multirow{2}{*}{$\tr\lr{\Ls\rho}$} & \multirow{2}{*}{$\tr\lr{(\I-\Lr-\Ls)\rho}$}& \multirow{2}{*}{$\tr\lr{\Lr\rho}+\tr\lr{\Ls\rho}$} \\
		&                   &                   &               &    \\
		Description ---& Correct outcome                  &   Type-1 error                &       Rejection        &  Acceptance  \\ \hline
		\multirow{2}{*}{$\sigma$} & \multirow{2}{*}{$\tr\lr{\Lr\sigma}$} & \multirow{2}{*}{$\tr\lr{\Ls\sigma}$} &\multirow{2}{*}{$\tr\lr{(\I-\Lr-\Ls)\sigma}$} & \multirow{2}{*}{$\tr\lr{\Lr\sigma}+\tr\lr{\Ls\sigma}$} \\
		&                   &                   &               &    \\		Description ---& Type-2 error               &   Correct outcome               &       Rejection        &  Acceptance  \\ \hline
	\end{tabular}
	\caption{Various outcomes and their probabilities for measurement $\Lambda$ given object is in the state $\rho$ or $\sigma$.}
	\label{table: probs}
\end{table}
\noindent Type-1 error refers to the error when the state of the quantum object is declared to be $\sigma$ when the state is actually $\rho$. Similarly, Type-2 error refers to the error when the state of the quantum object is declared to be $\rho$ when the state is actually $\sigma$. Given prior probabilities $p_\rho$ and $p_\sigma$ (collectively denoted as $p$) such that $p_\rho>0,p_\sigma>0$, and $p_\rho+p_\sigma=1$, the total error probability is $p_\rho\tr(\Lambda_\sigma\rho)+p_\sigma\tr(\Lambda_\rho\sigma)$.

We define \textit{acceptance} as the probability of accepting the test. Note that this corresponds to selecting one of the first two outcomes and excludes the third outcome. Hence, the acceptance for the state $\rho$ is  
\begin{gather}
	A_\rho(\Lambda)\isdefined\tr((\Lambda_\rho+\Lambda_\sigma)\rho) \label{eq: defArho}
	\intertext{and for the state $\sigma$, it is}
	A_\sigma(\Lambda)\isdefined\tr((\Lambda_\rho+\Lambda_\sigma)\sigma).\label{eq: defAsigma}
\end{gather}
It may seem from the definition of acceptance given in \eqref{eq: defArho} and \eqref{eq: defAsigma} that the higher acceptance leads to better measurement, as it corresponds to the probability of accepting at least one hypothesis. However, we can observe that the expression of acceptance consists of two terms: the probability of correct decision and the probability of error. Intuitively, we desire the probability of correct decision to be as large as possible and the error as low as possible, hence a trade-off. Hence, our objective is to maximize acceptance only after minimizing the error, i.e., to find the measurement with the highest acceptance over the set of all the measurements with the minimum error.

The postselected probability of an event is defined as the probability of the event conditioned on the event that the alternative or the null hypothesis is accepted. Given prior probabilities $p_\rho$ and $p_\sigma$, the error probability is $p_\rho\tr(\Lambda_\sigma\rho)+p_\sigma\tr(\Lambda_\rho\sigma)$ and the probability that the null or alternative hypothesis is accepted is $p_\rho A_\rho(\Lambda)+p_\sigma A_\sigma(\Lambda)$. So, the postselected symmetric error $e(\Lambda)$ is defined as \cite{regula2022postselected}
\begin{align}
e(\Lambda)\isdefinedas\frac{p_\rho\tr(\Lambda_\sigma\rho)+p_\sigma\tr(\Lambda_\rho\sigma)}{p_\rho A_\rho(\Lambda)+p_\sigma A_\sigma(\Lambda)}=\frac{\tr(p_\sigma\Lambda_\rho\sigma+p_\rho\Lambda_\sigma\rho)}{\tr((\Lambda_\rho+\Lambda_\sigma)(p_\rho\rho+p_\sigma\sigma))}. \label{eq: defSymError}
\end{align} 

With a little abuse of notation, we have extended definition of $e(\Lambda)$ for any $\Lambda=\{\Lambda_\rho,\Lambda_\sigma\}$ such that $\Lambda_\rho,\Lambda_\sigma\in\mathcal{P}(\hilbert)$. We denote the set of all such operators as 
\begin{align}
	\mathcal{O}\isdefinedas \left\{\Lambda:\Lambda=\{\Lambda_{\rho},\Lambda_{\sigma}\},\Lambda_{\rho},\Lambda_{\sigma}\in\mathcal{P}(\hilbert)\right\}.
\end{align}
Further $\Lambda\in\mathcal{M}$ is a measurement or POVM and taken as $\Lambda=\{\Lambda_\rho,\Lambda_\sigma,\I-\Lambda_\rho-\Lambda_\sigma\}$.  $\Lambda\in\mathcal{O}$ is an operator (pair) and taken as $\Lambda=\{\Lambda_\rho,\Lambda_\sigma\}$ without additional constraints of being a measurement. We now define the minimum postselected symmetric error.

\begin{defn}[\textbf{The minimum postselected symmetric error}]
	Given prior probability $p_\rho$ and $p_\sigma$, the minimum postselected symmetric error is defined as the minimum achievable postselected symmetric error probability over all measurements, i.e., $$\ds e_s(\rho,\sigma,p)\isdefinedas\inf_{\Lambda\in\mathcal{M}} e(\Lambda).$$
\end{defn}
\noindent Here, $\ds \mathcal{M}$ is the set of all measurements. Utilizing \eqref{eq: defSymError}, $e_s(\rho,\sigma,p)$ can be simplified as\cite[Theorem 6]{regula2022postselected}
\begin{align}
	e_s(\rho,\sigma,p)&=\inf_{\Lambda\in\mathcal{M}} e(\Lambda)=\inf_{\Lambda\in\mathcal{M}} \frac{\tr(p_\sigma\Lambda_\rho\sigma+p_\rho\Lambda_\sigma\rho)}{\tr((\Lambda_\rho+\Lambda_\sigma)(p_\rho\rho+p_\sigma\sigma))}\nonumber\\
%	&=\inf_{\Lambda\in\mathcal{M}} \lr{1+\frac{{\tr(p_\rho\Lambda_\rho\rho+p_\sigma\Lambda_\sigma\sigma)}}{\tr(p_\sigma\Lambda_\rho\sigma+p_\rho\Lambda_\sigma\rho)}}^{-1}\nonumber\\
	&= \lr{1+\sup_{\Lambda\in\mathcal{M}}\frac{{\tr(p_\rho\Lambda_\rho\rho+p_\sigma\Lambda_\sigma\sigma)}}{\tr(p_\sigma\Lambda_\rho\sigma+p_\rho\Lambda_\sigma\rho)}}^{-1}\nonumber\\
	&= \lr{1+\sup_{\Lambda\in\mathcal{O}}\frac{{\tr(p_\rho\Lambda_\rho\rho+p_\sigma\Lambda_\sigma\sigma)}}{\tr(p_\sigma\Lambda_\rho\sigma+p_\rho\Lambda_\sigma\rho)}}^{-1}.\label{eq: 2ndLastSymBase}
\end{align}
{\Blue The last step follows from the fact that, for every element $\Lambda_o\in \mathcal{O}$, we have a $\Lambda_m\in\mathcal{M}$ such that the expression $\displaystyle \frac{{\mathrm{Tr}(p_\rho\Lambda_\rho\rho+p_\sigma\Lambda_\sigma\sigma)}}{\mathrm{Tr}(p_\sigma\Lambda_\rho\sigma+p_\rho\Lambda_\sigma\rho)}$ takes the same value for both $\Lambda_o$ and $\Lambda_m$. Also, for every element $\Lambda_m\in\mathcal{M}$, we have a $\Lambda_o\in \mathcal{O}$ such that $\displaystyle \frac{{\mathrm{Tr}(p_\rho\Lambda_\rho\rho+p_\sigma\Lambda_\sigma\sigma)}}{\mathrm{Tr}(p_\sigma\Lambda_\rho\sigma+p_\rho\Lambda_\sigma\rho)}$ takes same value for both $\Lambda_m$ and $\Lambda_0$. We have written for set $\mathcal{O}$ because we get a simple subspace constraint for obtaining the supremum in \eqref{eq: 2ndLastSymBase} and \eqref{eq: cricial2ndLastSymBase}, as will be clear from Theorem \ref{lem: conditionSym}.}

Now, from \cite[Equation (66)]{regula2022postselected}, we know that
\begin{align}
	\sup_{\Lambda\in\mathcal{O}} {\frac{{\tr(p_\rho\Lambda_\rho\rho+p_\sigma\Lambda_\sigma\sigma)}}{\tr(p_\sigma\Lambda_\rho\sigma+p_\rho\Lambda_\sigma\rho)}}=\Xi\left(p_\rho\rho,p_\sigma\sigma\right), \label{eq: cricial2ndLastSymBase}
\end{align}
where the function $\Xi\left(\nu_1,\nu_2\right)$ is Thompson metric \cite{thompson1963certain}, is given as 
\begin{align}
	\Xi\left(\nu_1,\nu_2\right)=\begin{cases}
		\max\{R_{\max}\lr{\nu_1,\nu_2},R_{\max}\lr{\nu_2,\nu_1}\}, &\text{ if $\Pi_{\nu_1}=\Pi_{\nu_2}$}, \\ \infty, & \text{ otherwise.}
	\end{cases} \label{eq: defXi}
\end{align}
Substituting from \eqref{eq: cricial2ndLastSymBase} in \eqref{eq: 2ndLastSymBase}, the minimum postselected symmetric error is given as \cite{regula2022postselected}
\begin{align}
	e_s(\rho,\sigma,p)&=\lr{ \Xi\left(p_\rho\rho,p_\sigma\sigma\right)+1}^{-1}.\label{eq: symBase}
\end{align}

Although the minimum error is given in \cite{regula2022postselected}, along with an example measurement achieving this error, it is not known how to design an arbitrary error-minimizing measurement. One of the key goals in this work is to characterize the complete set of measurements with postselected error to have the minimum value as given in \eqref{eq: symBase}, which we define next.\hfill\\\vspace{-10mm}
\begin{wrapfigure}{r}{0.4\textwidth}
	\vspace{5mm}
	\centering
	\includegraphics[width=0.3\textwidth]{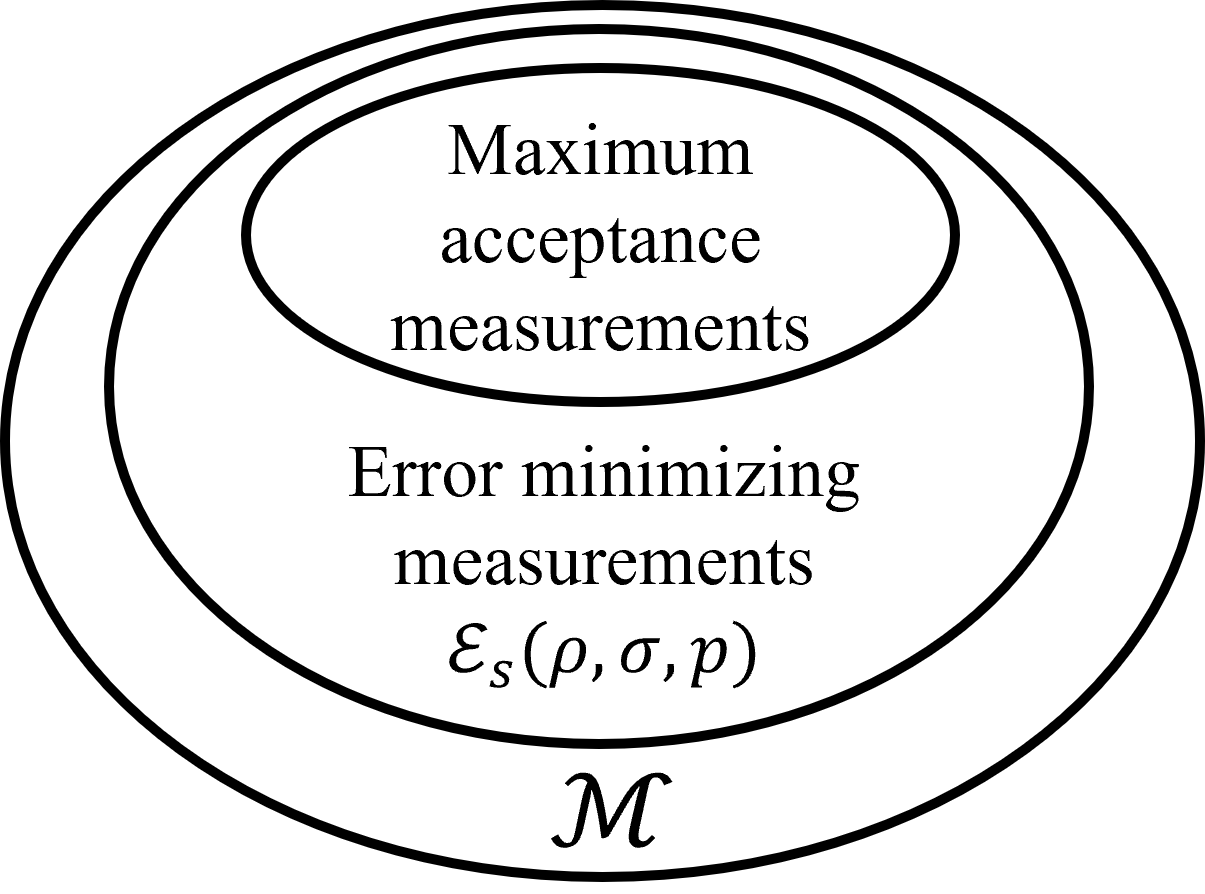}
	\caption{An illustration showing various sets of measurements. $\mathcal{M}$ is set of all measurements. $\mathcal{E}_s(\rho,\sigma,p)$ is set of all error-minimizing measurements. Maximum acceptance measurements comprise the set that achieves the maximum acceptance over $\mathcal{E}_s(\rho,\sigma,p)$.}
	\label{fig: setDiagramSymmetric}
\end{wrapfigure}
\begin{defn}[Set of error-minimizing measurements] The set of all error-minimizing measurements is defined as $$\mathcal{E}_s(\rho,\sigma,p)\isdefinedas\{\Lambda: e(\Lambda)=e_s(\rho,\sigma,p),\Lambda\in\mathcal{M}\}.$$
\end{defn}
 Within the set of error-minimizing measurements, there are measurements for which the acceptance for the state $\rho$ and $\sigma$ is maximized (see Fig. \ref{fig: setDiagramSymmetric}). The maximum acceptance is defined below.

\begin{defn}[The maximum acceptance] The maximum acceptance possible over the set $\mathcal{E}_s(\rho,\sigma,p)$ of all error-minimizing measurements is defined as
\begin{gather*}
	A_\rho^s=\max_{\Lambda\in \mathcal{E}_s(\rho,\sigma,p)} A_\rho(\Lambda)\ \text{ and } \ A_\sigma^s=\max_{\Lambda\in \mathcal{E}_s(\rho,\sigma,p)} A_\sigma(\Lambda)
\end{gather*}
for the states $\rho$ and $\sigma$ respectively. 
\end{defn}

{\Blue For the measurements that achieve the minimum error, we are interested in understanding how frequently a decision is made for each possible input state. Hence, we define acceptance of each state individually. The benefit of defining in this fashion is evident from Theorem \ref{thm: notSymmetricAcceptance}, where acceptance for one of the states turns out to be zero for error-minimizing measurements. This observation would not be evident otherwise.} 

{\Blue It can be observed from the expression of the minimum error in \eqref{eq: symBase} and the definition of $\Xi(\cdot,\cdot)$ in \eqref{eq: defXi} that a finite non-zero minimum error exists in the case when the supports of the two states are equal, i.e., $\Pi_\rho=\Pi_\sigma$. We already know from Theorem 2 in \cite{feng2004unambiguous} that such cases cannot be discriminated unambiguously. Error minimizing measurements for such cases ($\Pi_\rho=\Pi_\sigma$) are discussed in Section \ref{sec: rhoEqsigma}. In the remaining cases, the error vanishes, and the analysis reduces to USD. Such cases are discussed in Section \ref{sec: rhoNeqsigma}. The analysis is split this way because the two cases have different subspace constraints on the error-minimizing measurements, thus requiring separate treatments.}

\section{Case: $\rho$ and $\sigma$ have the same support i.e. $\Pi_\rho=\Pi_\sigma$} \label{sec: rhoEqsigma}

We now study PSD for the case where the supports of two states are not equal, and hence, USD is not feasible. We begin with the first subsection by finding the condition on the measurement operators that must be satisfied for the minimum error to be achieved. We start with a key step \eqref{eq: cricial2ndLastSymBase} in the derivation of \eqref{eq: symBase} from \cite{regula2022postselected}, building towards a direction for obtaining the error-minimizing condition on any measurement. We note that any measurement that achieves the supremum in \eqref{eq: cricial2ndLastSymBase} is error-minimizing. We focus on obtaining the condition for achieving this supremum in the proof of Theorem 1 given in Appendix A. Although the objective of Theorem \ref{lem: conditionSym} is to get the condition for an arbitrary measurement to be an error-minimizing measurement, we also obtained a novel derivation of \eqref{eq: symBase}, which is included as a result of the same theorem. Next, we present some key novel properties of the error-minimizing measurements based on the derived condition. Further, we write the set of error-minimization measurements in a parameterized form, providing a method to construct an arbitrary error-minimizing measurement. Then, we provide an example to show that the acceptance for an arbitrary error-minimizing measurement varies with these parameters, even though each one achieves the minimum error. This demonstrates the need to maximize acceptance for obtaining an optimal measurement. In the last subsection, we derive the expression for maximum achievable acceptance over the set of all error-minimizing measurements. 

\subsection{The set of all error minimizing measurements and construction of an arbitrary error-minimizing measurement}
{\Blue Recall that the minimum postselected error in \eqref{eq: symBase} is obtained in terms of $\Xi(\cdot,\cdot)$. From \eqref{eq: defXi}, note that, if $\Pi_\rho=\Pi_\sigma$, we get $\Xi(p_\rho\rho,p_\sigma\sigma) = \max(R_{\max}(p_\sigma\sigma,p_\rho\rho), R_{\max}(p_\sigma\sigma,p_\rho\rho))$. As we see later that the analysis depends on the relative value of the two terms $\max(R_{\max}(p_\sigma\sigma,p_\rho\rho)$, and $R_{\max}(p_\sigma\sigma,p_\rho\rho))$, we will study the following three cases separately.
\begin{enumerate}[label={$\bullet$ Case $\underline{\mathcal{C}\mathrm{\arabic*}}$: }]
	\setlength{\itemindent}{1.5cm}
	\item $R_{\max}(p_\rho\rho,p_\sigma\sigma)>R_{\max}(p_\sigma\sigma,p_\rho\rho)$,
	\item $R_{\max}(p_\rho\rho,p_\sigma\sigma)<R_{\max}(p_\sigma\sigma,p_\rho\rho)$, and 
	\item $R_{\max}(p_\rho\rho,p_\sigma\sigma)=R_{\max}(p_\sigma\sigma,p_\rho\rho)$.
\end{enumerate} 
In particular, the proof of Theorem \ref{lem: conditionSym} shows that we obtain distinct conditions for an error-minimizing measurement for each of these three cases.  Note that the value of two possible states and their prior probabilities determine which case the analysis would correspond to. }
The following theorem derives the condition on the measurement operators to achieve the minimum postselected symmetric error, along with providing a novel proof of deriving its minimum possible value.

\begin{theorem} \label{lem: conditionSym}
For any $\rho,\sigma\in\mathcal{D}(\hilbert)$, postselected symmetric error is lower bounded as
\begin{align}
	e(\Lambda)\geq e_s(\rho,\sigma,p)\ \forall \  \Lambda\in\mathcal{O}
\end{align}
 and the equality is obtained iff the measurement operators $\{\Lambda_\rho,\Lambda_\sigma\}$ satisfy the condition:
\begin{align}
	\begin{cases}
		\sigma^{1/2}\Lr\sigma^{1/2}\in\mathcal{P}(\Tau^{\max}),\quad\sigma^{1/2}\Ls\sigma^{1/2}=0,& \qquad\text{for Case $\underline{\mathcal{C}1}$,}\\ 
		\sigma^{1/2}\Lr\sigma^{1/2}=0,\qquad\qquad\sigma^{1/2}\Ls\sigma^{1/2}\in\mathcal{P}(\Tau^{\min}),&\qquad\text{for Case $\underline{\mathcal{C}2}$,}\\ 
		\sigma^{1/2}\Lr\sigma^{1/2}\in\mathcal{P}(\Tau^{\max}),\quad\sigma^{1/2}\Ls\sigma^{1/2}\in\mathcal{P}(\Tau^{\min}),&\qquad\text{for Case $\underline{\mathcal{C}3}$,}
	\end{cases} \label{eq: conditionSym}
\end{align}  
where $\Tau^{\max}\isdefinedas\Pi^{\max}_{\sigma^{-1/2}\rho\sigma^{-1/2}}$, and $\Tau^{\min}\isdefinedas\Pi^{\min}_{\sigma^{-1/2}\rho\sigma^{-1/2}}$.
\begin{proof}
	The proof is given in Appendix \ref{proof: conditionSym}.
\end{proof} 
\end{theorem}

Recall from \eqref{eq: 2ndLastSymBase} that the minimum value of $e(\Lambda)$ requires finding  $\ds\sup_{\Lambda\in\mathcal{O}}\frac{{\tr(p_\rho\Lambda_\rho\rho+p_\sigma\Lambda_\sigma\sigma)}}{\tr(p_\sigma\Lambda_\rho\sigma+p_\rho\Lambda_\sigma\rho)}$. We substitute $\max$ in the place of $\sup$, as the maximum is shown to be achievable in \cite{regula2022postselected}. We can express the maximum as
\begin{align}
	\max_{\Lambda\in\mathcal{O}}\frac{{\tr(p_\rho\Lambda_\rho\rho+p_\sigma\Lambda_\sigma\sigma)}}{\tr(p_\sigma\Lambda_\rho\sigma+p_\rho\Lambda_\sigma\rho)}&=\max\lr{\max_{\Lr\in\mathcal{P}(\hilbert)}\frac{{\tr(p_\rho\Lambda_\rho\rho)}}{\tr(p_\sigma\Lambda_\rho\sigma)},\max_{\Ls\in\mathcal{P}(\hilbert)}\frac{{\tr(p_\sigma\Lambda_\sigma\sigma)}}{\tr(p_\rho\Lambda_\sigma\rho)}} \label{eq: convRatio}\\
	&=\max(R_{\max}(p_\rho\rho,p_\sigma\sigma),R_{\max}(p_\sigma\sigma,p_\rho\rho)).\nonumber
\end{align}
The first case in \eqref{eq: conditionSym} corresponds to ${R_{\max}(p_\rho\rho,p_\sigma\sigma)}>{ R_{\max}(p_\sigma\sigma,p_\rho\rho)}$, and it signifies the condition that $\max_{\Lr\in\mathcal{P}(\hilbert)}\frac{p_\rho\tr\lr{\Lr\rho}}{p_\sigma\tr\lr{\Lr\sigma}}>\max_{\Ls\in\mathcal{P}(\hilbert)}\frac{p_\sigma\tr\lr{\Ls\sigma}}{p_\rho\tr\lr{\Ls\rho}}$. In this case, $\Lr$ is taken such that $\tr\lr{\Ls\rho}=\tr\lr{\Ls\sigma}=0$ and $\ds \max_{\Lambda\in\mathcal{O}}\frac{{\tr(p_\rho\Lambda_\rho\rho+p_\sigma\Lambda_\sigma\sigma)}}{\tr(p_\sigma\Lambda_\rho\sigma+p_\rho\Lambda_\sigma\rho)}=\max_{\Lr\in\mathcal{P}(\hilbert)}\frac{{\tr(p_\rho\Lambda_\rho\rho)}}{\tr(p_\sigma\Lambda_\rho\sigma)}=R_{\max}(p_\rho\rho,p_\sigma\sigma)$.
Similarly, second case is corresponds to ${R_{\max}(p_\rho\rho,p_\sigma\sigma)}<{ R_{\max}(p_\sigma\sigma,p_\rho\rho)}$, which signifies the condition that\\ $ \max_{\Lr\in\mathcal{P}(\hilbert)}\frac{p_\rho\tr\lr{\Lr\rho}}{p_\sigma\tr\lr{\Lr\sigma}}<\max_{\Ls\in\mathcal{P}(\hilbert)}\frac{p_\sigma\tr\lr{\Ls\sigma}}{p_\rho\tr\lr{\Ls\rho}}$. The third case corresponds to ${R_{\max}(p_\rho\rho,p_\sigma\sigma)}={ R_{\max}(p_\sigma\sigma,p_\rho\rho)}$, and it signifies when $\ds \max_{\Lr\in\mathcal{P}(\hilbert)}\frac{p_\rho\tr\lr{\Lr\rho}}{p_\sigma\tr\lr{\Lr\sigma}}=\max_{\Ls\in\mathcal{P}(\hilbert)}\frac{p_\sigma\tr\lr{\Ls\sigma}}{p_\rho\tr\lr{\Ls\rho}}$, either can be chosen. We will have to deal with the cases separately. %and hence need to represent them compactly. We represent them as $\mathcal{C}1, \mathcal{C}2,$ and $\mathcal{C}3$ respectively for the remaining of this section, as mentioned in \eqref{eq: conditionSym}.}

\begin{remark}[\textbf{Notes on $\Tau^{\max}$ and $\Tau^{\min}$}] \label{rem: notesOnT}
	$\Tau^{\max}$ denotes the projection operator onto the subspace where $\rho$ is largest in comparison to $\sigma$, in the sense that any vector $|\psi\rangle$ in this subspace, we get $\frac{\langle\psi|\rho|\psi\rangle}{\langle\psi|\sigma|\psi\rangle}=R_{\max}(\rho,\sigma)$, which is the highest possible value it can have. Similarly, for the subspace corresponding to the projection operator $\Tau^{\min}$, this ratio $\frac{\langle\psi|\rho|\psi\rangle}{\langle\psi|\sigma|\psi\rangle}=R_{\min}(\rho,\sigma)$, which is also the minimum possible value it can achieve. In the classical scenario, i.e., while discriminating a pair of probability distributions, corresponding to some random variable, on taking density operators corresponding to probabilities to be discriminated, $\Tau^{\max}$ and $\Tau^{\min}$ would respectively correspond to choosing the realization of the random variable for which the ratio of the two probabilities are maximum and minimum.
\end{remark}

\begin{remark}[\textbf{Notes on $\mathcal{P}(\Tau^{\max})$ and $\mathcal{P}(\Tau^{\min})$}] \label{rem: notesOnP}
	Now, $\sigma^{1/2}\Lr\sigma^{1/2}\in\mathcal{P}(\Tau^{\max})$ means that all the eigenvectors of $\sigma^{1/2}\Lr\sigma^{1/2}$ corresponding to positive eigenvalues lie in the subspace corresponding to the operator $\Tau^{\max}$. Similarly, $\sigma^{1/2}\Ls\sigma^{1/2}\in\mathcal{P}(\Tau^{\min})$ means that all the eigenvectors of $\sigma^{1/2}\Ls\sigma^{1/2}$ corresponding to positive eigenvalues lie in the subspace corresponding to the operator $\Tau^{\min}$. So, Theorem \ref{lem: conditionSym} states that, by restricting the subspace where eigenvector of the operators $\sigma^{1/2}\Lr\sigma^{1/2}$ and $\sigma^{1/2}\Ls\sigma^{1/2}$ lie, one can obtain measurements that achieve the minimum postselected symmetric error. 
\end{remark}

Now note that the condition $R_{\max}(p_\rho\rho,p_\sigma\sigma)=R_{\max}(p_\sigma\sigma,p_\rho\rho)$ in $\mathcal{C}3$ is equivalent to
\begin{align}
	R_{\max}(p_\rho\rho,p_\sigma\sigma)&=R_{\max}(p_\sigma\sigma,p_\rho\rho)\\&\Leftrightarrow\frac{p_\rho}{p_\sigma}R_{\max}(\rho,\sigma) = \frac{p_\sigma}{p_\rho}R_{\max}(\sigma,\rho)\\&\Leftrightarrow\frac{p_\rho}{p_\sigma} = \frac{\sqrt{R_{\max}(\sigma,\rho)}}{\sqrt{R_{\max}(\rho,\sigma)}}. \label{eq: ratioFixed}
\end{align}
\eqref{eq: ratioFixed} along with $p_\rho+p_\sigma=1$ gives $p_\rho=p_\rho^*$ and $p_\sigma=p_\sigma^*$ where
\begin{align}
	p_\rho^*= \frac{\sqrt{R_{\max}(\sigma,\rho)}}{\sqrt{R_{\max}(\rho,\sigma)}+\sqrt{R_{\max}(\sigma,\rho)}} \text{ and } p_\sigma^*=\frac{\sqrt{R_{\max}(\rho,\sigma)}}{\sqrt{R_{\max}(\rho,\sigma)}+\sqrt{R_{\max}(\sigma,\rho)}}. \label{eq: star_p}
\end{align} 

Further, we can show that the condition $R_{\max}(p_\rho\rho,p_\sigma\sigma)>R_{\max}(p_\sigma\sigma,p_\rho\rho)$ is equivalent to $p_\rho>p_\rho^*$ (i.e., $p_\sigma<p_\sigma^*$). Hence, we can write the three cases in Theorem 1 as 
\begin{itemize}
	\item {$\underline{\mathcal{C}1}: p_\rho>p_\rho^*$ (i.e., $p_\sigma<p_\sigma^*$)},
	\item {$\underline{\mathcal{C}2}: p_\rho<p_\rho^*$ (i.e., $p_\sigma>p_\sigma^*$)},
	\item {$\underline{\mathcal{C}3}: p_\rho=p_\rho^*$ (i.e., $p_\sigma=p_\sigma^*$)}.
\end{itemize}

\begin{corollary} \label{cor: alwaysOnlyOne}
	If $p,\rho,\sigma$ satisfy $R_{\max}(p_\rho\rho,p_\sigma\sigma)\neq R_{\max}(p_\sigma\sigma,p_\rho\rho)$, or equivalently $p_\rho\ne p^*_\rho$, then for any $\Lambda$ such that $e(\Lambda)=e_s(\rho,\sigma,p)$, one of the following two holds:
	\begin{enumerate}
		\item The measurement never detects $\sigma$ that is  {$\tr\lr{\Lambda_\sigma\rho}=\tr\lr{\Lambda_\sigma\sigma}=0$}. The measurement outcome is either $\rho$ or the third outcome, i.e., rejecting the test(see Fig. \ref{fig:mostlyOne} (a)). 
		\item The measurement never detects $\rho$ that is  {$\tr\lr{\Lambda_\rho\rho}=\tr\lr{\Lambda_\rho\sigma}=0$}. The measurement outcome is either $\sigma$ or the third outcome, i.e., rejecting the test (see Fig. \ref{fig:mostlyOne} (b)).
	\end{enumerate}
So, for $p,\rho,\sigma$ such that $R_{\max}(p_\rho\rho,p_\sigma\sigma)\neq R_{\max}(p_\sigma\sigma,p_\rho\rho)$, any measurement that minimizes postselected symmetric error, either never decides in favor of $\rho$ or never makes a decision in favor of $\sigma$. Fig. \ref{fig:mostlyOne} (a)  and Fig. \ref{fig:mostlyOne} (b) depict the two possible measurements in this case.
	\begin{proof}
	From the Theorem \ref{lem: conditionSym}, when $R_{\max}(p_\rho\rho,p_\sigma\sigma)\neq R_{\max}(p_\sigma\sigma,p_\rho\rho)$ we get 
	\begin{align} 
		e(\Lambda)=e_s(\rho,\sigma,p)\Rightarrow
		\begin{cases}
			\sigma^{1/2}\Lambda_{\sigma}\sigma^{1/2}=0\Rightarrow\tr\lr{\Lambda_\sigma\sigma}=0,\text{ and }\tr\lr{\Lambda_\sigma\rho}=\tr\lr{\Lambda_\sigma\Pi_\rho\rho\Pi_\rho}\\
			\qquad=\tr\lr{\Lambda_\sigma\Pi_\sigma\rho\Pi_\sigma}=\tr\lr{\sigma^{1/2}\Lambda_{\sigma}\sigma^{1/2}\sigma^{-1/2}\rho\sigma^{-1/2}}=0,
			 \\\text{ or }
			\\\sigma^{1/2}\Lambda_{\rho}\sigma^{1/2}=0\Rightarrow\tr\lr{\Lambda_\rho\sigma}=0, \text{ and }\tr\lr{\Lambda_\rho\rho}=\tr\lr{\Lambda_\rho\Pi_\rho\rho\Pi_\rho}\\
			\qquad=\tr\lr{\Lambda_\rho\Pi_\sigma\rho\Pi_\sigma}=\tr\lr{\sigma^{1/2}\Lambda_{\rho}\sigma^{1/2}\sigma^{-1/2}\rho\sigma^{-1/2}}=0.
		\end{cases}
	\end{align}
So, one of the two cases must hold for any error-minimizing measurement.
	\end{proof}
\end{corollary}
%
%	
%\begin{figure}[h]
%	\centering
%	\begin{subfigure}[b]{0.25\textwidth}
%		\centering
%		\includegraphics[scale=0.7]{../fig/alwaysOnlyOneRho.png}
%		\caption{$\ $}
%		\label{fig: alwaysOnlyOneRho}
%	\end{subfigure}
%	\hspace{15mm}
%	\begin{subfigure}[b]{0.25\textwidth}
%		\centering
%		\includegraphics[scale=0.7]{../fig/alwaysOnlyOneSigma.png}
%		\caption{$\ $}
%		\label{fig: alwaysOnlyOneSigma}
%	\end{subfigure}
%	\caption{The figure shows the two possible outcomes when $R_{\max}(p_\rho\rho,p_\sigma\sigma)\neq R_{\max}(p_\sigma\sigma,p_\rho\rho)$ for any error minimizing measurement where (a) $\sigma$ never declared, (b) $\rho$ never declared.}
%	\label{fig: alwaysOnlyOne}
%\end{figure}

\begin{remark} \label{rem: betterMetric}
	The statement in the Corollary \ref{cor: alwaysOnlyOne} says that a measurement that achieves the minimum postselected error can only detect one of the state. The observation puts a very fundamental question, ``Does a lower postselected error indicate better hypothesis testing when prior probability is known?" The observation suggests a negative answer. So, the next question is, ``What is a better metric to assess the quality of hypothesis testing when prior probability is known?" We will not address it in this work and leave it as an open question.
\end{remark}
	\begin{figure}[h]
	\centering
	\begin{subfigure}[b]{0.25\textwidth}
		\centering
		\includegraphics[scale=0.6]{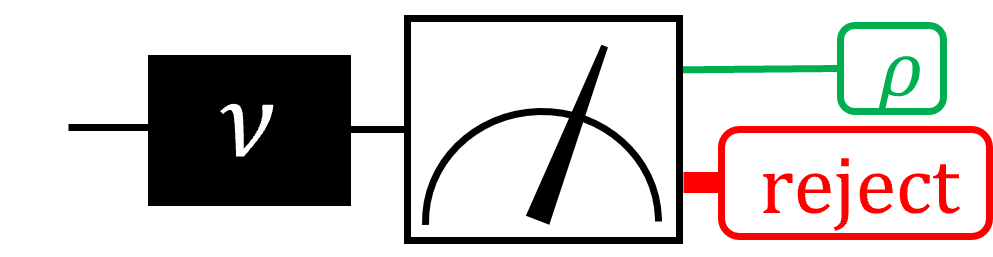}
		\caption{$p_\rho>p_\rho^*$ and $p_\sigma<p_\sigma^*$}
		\label{fig:mostlyOne1}
	\end{subfigure}
	\hspace{2mm}
	\begin{subfigure}[b]{0.25\textwidth}
		\centering
		\includegraphics[scale=0.6]{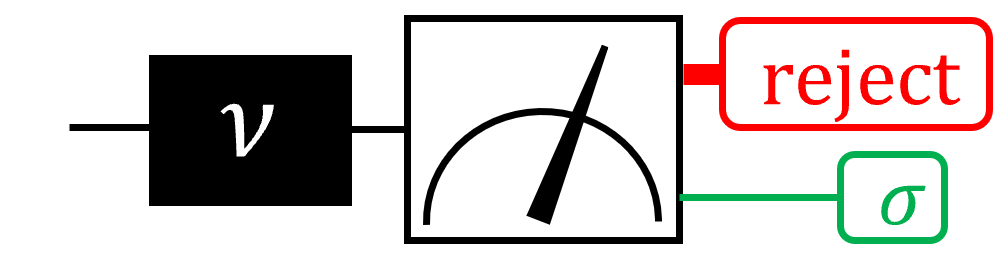}
		\caption{$p_\rho<p_\rho^*$ and $p_\sigma>p_\sigma^*$}
		\label{fig:mostlyOne2}
	\end{subfigure}
	\hspace{2mm}
	\begin{subfigure}[b]{0.25\textwidth}
		\centering
		\includegraphics[scale=0.6]{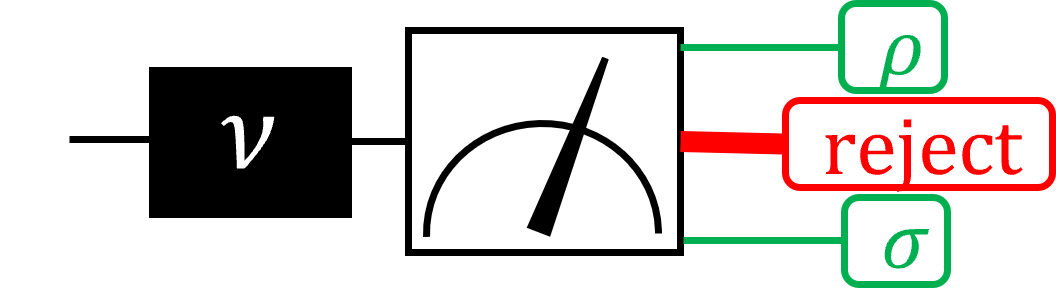}
		\caption{$p_\rho=p_\rho^*$ and $p_\sigma=p_\sigma^*$}
		\label{fig:mostlyOne3}
	\end{subfigure}
	\caption{An illustration showing possible outcomes for different prior probability values of $p_\rho$ and $p_\sigma$. The figure shows the possible outcomes when $p_\rho\ne p_\rho^*$ (equivalently $p_\sigma\ne p_\sigma^*$) for any error minimizing measurement where (a) $\sigma$ never declared, (b) $\rho$ never declared. Further, (c) shows that if $p_\rho = p_\rho^*$ (equivalently $p_\sigma = p_\sigma^*$), all the outcomes are possible.}
	\label{fig:mostlyOne}
\end{figure}
\begin{corollary} \label{cor: equalityCondition}
	A measurement with non-zero probabilities of detecting both $\rho$ and $\sigma$ is possible only if the prior probabilities $p_\rho=p_\rho^*$ and $p_\sigma=p_\sigma^*$. Further
	\begin{itemize}
		\item {$\underline{\mathcal{C}1}: p_\rho>p_\rho^*$ (i.e., $p_\sigma<p_\sigma^*$)}: Error-minimizing measurement never detects $\sigma$. It either detects $\rho$ or rejects the test. See Fig. \ref{fig:mostlyOne}(a).
		\item {$\underline{\mathcal{C}2}: p_\rho<p_\rho^*$ (i.e., $p_\sigma>p_\sigma^*$)}: Error-minimizing measurement never detects $\rho$. It either detects $\sigma$ or rejects the test. See Fig. \ref{fig:mostlyOne}(b).
		\item {$\underline{\mathcal{C}3}: p_\rho=p_\rho^*$ (i.e., $p_\sigma=p_\sigma^*$)}: There are error-minimizing measurements detecting either/both the hypotheses. See Fig. \ref{fig:mostlyOne}(c).
	\end{itemize}
\end{corollary}
\noindent
The observation says that if the prior probability of a state being $\rho$ is high, $\sigma$ is never detected for any error-minimizing measurement. Similarly, if the prior probability of the state being $\sigma$ is high, $\rho$ is never detected. So, an error-minimizing measurement will detect the high probability state or declare nothing. Also, a particular value of the pair $(p_\rho,p_\sigma)$ exists such that an error-minimizing measurement can detect both states.

For better exposition, we now give two examples. The first example shows how Theorem 1 can be used to find the error-minimizing measurements. We also observe that there are measurements, despite having the minimum postselected symmetric error, rejecting both hypotheses with a probability close to $1$. The observation shows that all error-minimizing measurements are not useful and there is a need to find those that also maximize acceptance.

%	\subsection{An illustrative example: Not all the error minimizing measurements have same acceptance}

Here, we take a Hilbert's space with it's basis $\{|0\rangle,|1\rangle,|2\rangle\}$. We consider both $\rho$ and $\sigma$ to be a mixture of 3 orthogonal pure states $\{|0\rangle,|1\rangle,|2\rangle\}$. $\rho$ has a higher probability to be in $|0\rangle$ than the probability of $\sigma$ to be in $|0\rangle$. Further, $\rho$ has a smaller probability to be in $|1\rangle$ than $\sigma$. $|2\rangle$ constitutes $\rho$ and $\sigma$ with equal probability. In relative terms, observe that $|0\rangle$ is relatively prominent in the state $\rho$, $|1\rangle$ being prominent in the state $\sigma$, while $|2\rangle$ has an equal probability of both. %We obtain that the measurements that minimize error declare the unknown state as $\rho$ by measuring $|0\rangle\langle0|$ and $\sigma$ is never detected. We will observe that, for the error-minimizing measurements, the value of acceptances for states $\rho$ and $\sigma$ vary depending on the choice of measurement. For an error-minimizing measurement, the acceptance can be very small, thus not declaring any state with a probability close to 1. We suggest that maximizing acceptance should also be considered while designing an error-minimizing measurement.

\begin{exmp} \label{exmp: first} Let $\{|0\rangle,|1\rangle,|2\rangle\}$ be the basis of a Hilbert’s space. Let us consider that two states% Take $p_\rho=1/2,p_\sigma=1/2$,
	 $$\rho=\frac{\mu}{2}|0\rangle\langle0|+\frac{\mu}{2}|1\rangle\langle1|+(1-\mu)|2\rangle\langle2|\text{ and } \sigma=\frac{\mu}{4}|0\rangle\langle0|+\frac{3\mu}{4}|1\rangle\langle1|+(1-\mu)|2\rangle\langle2|$$ with prior probabilities $p_\rho=1/2,p_\sigma=1/2$. Note that, here $\Pi_\rho=\Pi_\sigma=|0\rangle\langle0|+|1\rangle\langle1|+|2\rangle\langle2|=\I$, so $\I-\Pi_\sigma=0$. Now,
	\begin{align}
		\sigma^{-1/2}\rho\sigma^{-1/2}=2|0\rangle\langle0|+\frac{2}{3}|1\rangle\langle1|+|2\rangle\langle2| \text{ and }\rho^{-1/2}\sigma\rho^{-1/2}=\frac{1}{2}|0\rangle\langle0|+\frac{3}{2}|1\rangle\langle1|+|2\rangle\langle2|.
	\end{align}
	Observe that $R_{\max}(p_\rho\rho,p_\sigma\sigma)=R_{\max}(\rho,\sigma)=2$, $R_{\max}(p_\sigma\sigma,p_\rho\rho)=R_{\max}(\sigma,\rho)=3/2$, $$ \Tau^{\max}=\Pi^{\max}_{\sigma^{-1/2}\rho\sigma^{-1/2}}=|0\rangle\langle0|.$$ Thus, the minimum postselected symmetric error $e_s(\rho,\sigma,p)=1/3$. Note that $R_{\max}(p_\rho\rho,p_\sigma\sigma)>R_{\max}(p_\sigma\sigma,p_\rho\rho)$, so $\mathcal{C}1$ is applicable here. Hence, from Theorem \ref{lem: conditionSym}, we have
	\begin{align*}
		\sigma^{1/2}\Lr\sigma^{1/2}=c|0\rangle\langle0|, \text{ and } \sigma^{1/2}\Ls\sigma^{1/2}=0.
	\end{align*}
	So $\ds \Lr=c\frac{|0\rangle\langle0|}{\mu/4}=\frac{4c}{\mu}|0\rangle\langle0|$, $\Ls=0$. So, the set of POVM characterizing the measurements that achieve the minimum postselected symmetric error is given by
	\begin{align}
		\mathcal{E}_s(\rho,\sigma,p)&=\left\{\{\Lambda_{\rho},\Lambda_{\sigma},\I-\Lambda_{\rho}-\Lambda_{\sigma}\}:\Lr=\frac{4c}{\mu}|0\rangle\langle0|,\Ls=0,c\leq\frac{\mu}{4}  \right\}\nonumber\\
		&=\left\{\{\Lambda_{\rho},\Lambda_{\sigma},\I-\Lambda_{\rho}-\Lambda_{\sigma}\}:0<c \leq 1, \Lambda_{\rho}= c|0\rangle\langle0|, \Lambda_{\sigma}= 0\right\}.
	\end{align}
	So, an arbitrary error-minimizing measurement can be parameterized as $\{c|0\rangle\langle0|,0,\I-c|0\rangle\langle0|\}$ for some $0<c\leq 1$. For these measurements, $\tr\lr{\Lr\rho}=c\frac{\mu}{2}, \tr\lr{\Lr\sigma}=c \frac{\mu}{4}, \tr\lr{\Ls\rho}=0, \tr\lr{\Ls\sigma}=0$. 

Note that, with this POVM, a conclusive decision is being made for only $c\mu/2$ and $c\mu/4$ fraction of cases when the unknown state is $\rho$ and $\sigma$ respectively. If $c$ is small, the performance of the postselected test would be the same, i.e., $e_s(\rho,\sigma,p)=1/3$; however, we will be declaring inconclusive results most of the time, so the measurement would not be useful. We need some metrics to characterize the usefulness of postselected measurement. We use the metric \textit{acceptance} to characterize the usefulness.
\end{exmp}
In the above example, we saw that  all error-minimizing measurements declare the unknown state as $\rho$ if the outcome is $|0\rangle\langle0|$, otherwise they reject the test. Therefore, $\sigma$ is never detected. We observe that, for the error-minimizing measurements, the value of acceptances for states $\rho$ and $\sigma$ vary depending on the choice of measurement, in particular, on the choice of the parameter $c$. For an error-minimizing measurement, the acceptance can be very small, which means that the measurement makes decision with a probability close to $0$. The example suggests that maximizing acceptance should also be considered while designing an error-minimizing measurement. The second example is designed to explain the derived constraint and illustrates various observations made in this subsection about error-minimizing measurements.

\begin{example2} \label{exmp: q1}
	Consider the Hilbert's space $\hilbert_2$ with the basis $\{|0\rangle,|1\rangle\}$. Take the states as $$\rho=\frac{3}{4}|+\rangle\langle +|+\frac{1}{4}|-\rangle\langle -|=\frac{1}{2}|0\rangle\langle 0|+\frac{1}{4}|1\rangle\langle 0|+\frac{1}{4}|0\rangle\langle 1|+\frac{1}{2}|1\rangle\langle 1|,\quad \text{and } \sigma=\frac{3}{4}|0\rangle\langle 0|+\frac{1}{4}|1\rangle\langle 1|.$$
	Note that, in this example, $\Pi_\rho=\Pi_\sigma=\I$. We obtain 
	$$\sigma^{-1/2}\rho\sigma^{-1/2}=\frac{2}{3}|0\rangle\langle 0|+\frac{1}{\sqrt{3}}|1\rangle\langle 0|+\frac{1}{\sqrt{3}}|0\rangle\langle 1|+2|1\rangle\langle 1|.$$
	We get eigenvalues of $\sigma^{-1/2}\rho\sigma^{-1/2}$ as $\ds \frac{4\pm\sqrt{7}}{3}$ and corresponding eigenvectors as $\ds \frac{{\sqrt{3}|0\rangle+(2\pm\sqrt{7})|1\rangle}}{\sqrt{14\pm4\sqrt{7}}}.$
	Thus $R_{\max}(\rho,\sigma)=\frac{4+\sqrt{7}}{3}, R_{\min}(\rho,\sigma)=\frac{4-\sqrt{7}}{3}$. Similarly, $R_{\max}(\sigma,\rho)=\frac{4+\sqrt{7}}{3}.$  So, we get $$p_\rho^*=\frac{\sqrt{R_{\max}(\sigma,\rho)}}{\sqrt{R_{\max}(\rho,\sigma)}+\sqrt{R_{\max}(\sigma,\rho)}}=\frac{1}{2},\text{ and similarly }p_\sigma^*=\frac{1}{2}.$$ The projection operators onto the eigenspace corresponding to the maximum and minimum eigenvalue are respectively given by
	\begin{align*}
		\Tau^{\max}&=\frac{1}{14+4\sqrt{7}}\lr{\sqrt{3}|0\rangle+(2+\sqrt{7})|1\rangle}\lr{\sqrt{3}\langle0|+(2+\sqrt{7})\langle1|} \text{ and }\\
		\Tau^{\min}&=\frac{1}{14-4\sqrt{7}}\lr{\sqrt{3}|0\rangle+(2-\sqrt{7})|1\rangle}\lr{\sqrt{3}\langle0|+(2-\sqrt{7})\langle1|}.
	\end{align*}

	As $p_\rho^*=p_\sigma^*=1/2$, using Theorem \ref{lem: conditionSym}, we obtain the minimum error and error minimizing condition as the following.
	\begin{enumerate} 
		\item If $p_\rho>1/2: e_s(\rho,\sigma,p)=\lr{1+\frac{p_\rho}{p_\sigma}\frac{4+\sqrt{7}}{3}}^{-1}$ and error minimizing measurement should satisfy the constraint that $\sigma^{1/2}\Lr\sigma^{1/2}\in\mathcal{P}(\Tau^{\max})$, $\sigma^{1/2}\Ls\sigma^{1/2}=0$.
		\item If $p_\rho<1/2: e_s(\rho,\sigma,p)=\lr{1+\frac{p_\sigma}{p_\rho}\frac{4+\sqrt{7}}{3}}^{-1}$ and error minimizing measurement should satisfy the constraint that $\sigma^{1/2}\Lr\sigma^{1/2}=0$, $\sigma^{1/2}\Ls\sigma^{1/2}\in\mathcal{P}(\Tau^{\min})$.
		\item If $p_\rho=1/2: e_s(\rho,\sigma,p)=\lr{1+\frac{4+\sqrt{7}}{3}}^{-1}$ and error minimizing measurement should satisfy the constraint that $\sigma^{1/2}\Lr\sigma^{1/2}\in\mathcal{P}(\Tau^{\max})$, $\sigma^{1/2}\Ls\sigma^{1/2}\in\mathcal{P}(\Tau^{\min})$.
	\end{enumerate}
	Note that $\Tau^{\max}$ and $\Tau^{\min}$ are both rank 1 in this example and so, any element in the set $\mathcal{P}(\Tau^{\max})$ and $\mathcal{P}(\Tau^{\max})$ are of the form $c\Tau^{\max}$ and $c\Tau^{\min}$ for some $c\geq0$. Following the fact that $\sigma^{1/2}\Lr\sigma^{1/2}$ and $\sigma^{1/2}\Ls\sigma^{1/2}$ have simple form, and that $\Pi_\sigma=\I$  (thus $\sigma^{-1/2}(\sigma^{1/2}\Ls\sigma^{1/2})\sigma^{-1/2} = \Pi_\sigma\Ls\Pi_\sigma=\Ls$, and same for $\Lr$) in this example, we can obtain constraint on $\Lr,\Ls$ by simply multiplying $\sigma^{-1/2}$ on both sides. A comprehensive way to find the most general form of $\Lr,\Ls$ for error-minimization is given next.
\end{example2}

Theorem \ref{lem: conditionSym} gives the subspace where $\sigma^{1/2}\Lr\sigma^{1/2}$ and $\sigma^{1/2}\Ls\sigma^{1/2}$ should lie for $\{\Lambda_\rho,\Lambda_\sigma,\I-\Lr-\Ls\}$ to be an error minimizing measurement. Building on Theorem \ref{lem: conditionSym}, we now focus on deriving the set of all error-minimizing measurements. We start with the constraint on $\sigma^{1/2}\Lr\sigma^{1/2}$ and  $\sigma^{1/2}\Ls\sigma^{1/2}$  given in \eqref{eq: conditionSym} and use Lemma \ref{lem: genMainProj} from Appendix \ref{sec: genProj}  to obtain the constraint on $\Lr$ and $\Ls$, from constraint on $\sigma^{1/2}\Lr\sigma^{1/2}$ and  $\sigma^{1/2}\Ls\sigma^{1/2}$. Then, we parameterize $\Lr$ and $\Ls$ to satisfy the obtained constraints and then put condition $\Lr+\Ls\leq\I$ to ensure that it remains a valid measurement, thus obtaining the set of all measurements. The set of all error-minimizing measurements is obtained in the following theorem.

\begin{comment}
\noindent \textbf{Note:} The condition for equality is written compactly as 
\begin{align}
	e(\Lambda)=e_s(\rho,\sigma,p)
	\Leftrightarrow \sigma^{1/2}\Lambda_{\rho}\sigma^{1/2} = \tilde c_\rho c_\rho\psi_{\max} \text{ and }  	\sigma^{1/2}\Lambda_{\sigma}\sigma^{1/2}=\tilde c_\sigma c_\sigma \psi_{\min},
\end{align}
with $c_\rho>0,c_\sigma>0$ and 
$\tilde c_\rho,\tilde c_\sigma$ given by
\begin{align} \label{eq: tildeC}
	(\tilde c_\rho,\tilde c_\sigma)=
	\begin{cases}
		(1,0) & \text{if } R_{\max}(p_\rho\rho,p_\sigma\sigma)> R_{\max}(p_\sigma\sigma,p_\rho\rho),\\
		(0,1) &  \text{if } R_{\max}(p_\rho\rho,p_\sigma\sigma)< R_{\max}(p_\sigma\sigma,p_\rho\rho)	,\\
		(1,1), &  \text{if } R_{\max}(p_\rho\rho,p_\sigma\sigma)= R_{\max}(p_\sigma\sigma,p_\rho\rho).
	\end{cases}
\end{align}
We will use this notation $(\tilde c_\rho,\tilde c_\sigma)$ with same definition frequently in this section to avoid the need to write the three cases separately.
\end{comment}
 
\begin{theorem} \label{lem: setSym}
	The set $\mathcal{E}_s(\rho,\sigma,p)$ of all measurements achieving the minimum postselected symmetric error  for the three cases as mentioned in \eqref{eq: conditionSym}, is given respectively by 
	\begin{enumerate}[label={$\underline{\mathcal{C}\mathrm{\arabic*}}$:}]
		\item $\ds \mathcal{E}_s(\rho,\sigma,p)=\Bigg\{\{\Lambda_\rho,\Lambda_\sigma, \I-\Lambda_\rho-\Lambda_\sigma\}:
		\psi_{\max}\in\mathcal{S}(\I-\Pi_\sigma+\proj^{\max}),\Lr=c\frac{\psi_{\max}}{\tr\lr{\psi_{\max}\sigma}},\\ \text{ } \hspace{7cm} \Ls\in\mathcal{P}(\I-\Pi_\sigma),\Ls\leq \I-\Lr, c\leq\left\|\frac{\psi_{\max}}{\tr\lr{\psi_{\max}\sigma}}\right\|_{\infty}^{-1}\Bigg\}.$
		
		\item $\ds \mathcal{E}_s(\rho,\sigma,p)=\Bigg\{\{\Lambda_\rho,\Lambda_\sigma, \I-\Lambda_\rho-\Lambda_\sigma\}: \psi_{\min}\in\mathcal{S}(\I-\Pi_\sigma+\proj^{\min}),\Ls=c\frac{\psi_{\min}}{\tr\lr{\psi_{\min}\sigma}},\\ \text{ } \hspace{7cm}\Lr\in\mathcal{P}(\I-\Pi_\sigma),\Lr\leq \I-\Ls,
		c\leq\left\|\frac{\psi_{\min}}{\tr\lr{\psi_{\min}\sigma}}\right\|_{\infty}^{-1}\Bigg\}.$
		
		\item $\ds \mathcal{E}_s(\rho,\sigma,p)=\Bigg\{\{\Lambda_\rho,\Lambda_\sigma, \I-\Lambda_\rho-\Lambda_\sigma\}:\psi_{\max}\in\mathcal{S}(\I-\Pi_\sigma+\proj^{\max}),\psi_{\min}\in\mathcal{S}(\I-\Pi_\sigma+\proj^{\min}),c_r\in[0,1], \Lr=cc_r\frac{\psi_{\max}}{\tr\lr{\psi_{\max}\sigma}},\Ls=c(1-c_r)\frac{\psi_{\min}}{\tr\lr{\psi_{\min}\sigma}}, c\leq\left\|c_r\frac{\psi_{\max}}{\tr\lr{\psi_{\max}\sigma}}+(1-c_r)\frac{\psi_{\min}}{\tr\lr{\psi_{\min}\sigma}}\right\|_{\infty}^{-1}\Bigg\}.$
	\end{enumerate}
Here $\proj^{\max}=\Pi_{\sigma^{-1/2}\Tau^{\max}\sigma^{-1/2}}$ and $\proj^{\min}=\Pi_{\sigma^{-1/2}\Tau^{\min}\sigma^{-1/2}}$ with $\Tau^{\max}\isdefinedas\Pi^{\max}_{\sigma^{-1/2}\rho\sigma^{-1/2}}$ and $\Tau^{\min}\isdefinedas\Pi^{\min}_{\sigma^{-1/2}\rho\sigma^{-1/2}}$ as defined in Theorem \ref{lem: conditionSym}.
\begin{proof} Using Theorem \ref{lem: genMainProj} and Lemma \ref{lem: null} from Appendix \ref{sec: genProj}, we get
	\begin{gather}
		\begin{gathered}
		\sigma^{1/2}\Gamma\sigma^{1/2}\in\mathcal{P}(\Pi)\Leftrightarrow\Gamma\in\mathcal{P}(\I-\Pi_\sigma+\Pi_{\sigma^{-1/2}\Pi\sigma^{-1/2}})\text{ and }\\
		\sigma^{1/2}\Gamma\sigma^{1/2}=0\Leftrightarrow\Pi_\sigma\Gamma\Pi_\sigma=0\Leftrightarrow\Gamma\in\mathcal{P}(\I-\Pi_\sigma). 
		\end{gathered}
		\label{eq: generalCondition}
	\end{gather}
	We will use these two results for $\Gamma=\Lr$ and $\Gamma=\Ls$ to get the expression for $\Lr$ and $\Ls$. We begin with the condition on $\Lr$ and $\Ls$ derived in Theorem \ref{lem: conditionSym} for the three cases separately.
	 \begin{enumerate}[label={$\underline{\mathcal{C}\mathrm{\arabic*}}$}]
		\item In this case, using Theorem \ref{lem: conditionSym}, for error minimizing POVM, we get the condition
		$\sigma^{1/2}\Lr\sigma^{1/2}\in\mathcal{P}(\Tau^{\max})$ and $\sigma^{1/2}\Ls\sigma^{1/2}=0$. Now, using equivalence shown in \eqref{eq: generalCondition}, we get 
		$$\Lr\in\mathcal{P}(\I-\Pi_\sigma+\proj^{\max}), \ \Ls\in\mathcal{P}(\I-\Pi_\sigma).$$ A general $\Lr$ can be chosen as $\displaystyle\Lr={c}\frac{\psi_{\max}}{\tr\lr{\psi_{\max}\sigma}}$ for some $\psi_{\max}\in\mathcal{S}(\I-\Pi_\sigma+\proj^{\max}),c\geq0$. On putting the condition that $\Lr+\Ls\leq \I$, we obtain $\Ls\leq \I-\Lr$, and for any such $\Ls$ to exist, it has to be ensured that $\Lr\leq \I$ and thus $\ds c\frac{\psi_{\max}}{\tr\lr{\psi_{\max}\sigma}}\leq \I\Rightarrow c\leq\left\|\frac{\psi_{\max}}{\tr\lr{\psi_{\max}\sigma}}\right\|_{\infty}^{-1}$. Writing these conditions together, we get an error-minimizing measurement as
		$$\Lr=c\frac{\psi_{\max}}{\tr\lr{\psi_{\max}\sigma}},\Ls\in\mathcal{P}(\I-\Pi_\sigma),\Ls\leq \I-\Lr,c\leq\left\|\frac{\psi_{\max}}{\tr\lr{\psi_{\max}\sigma}}\right\|_{\infty}^{-1}.$$

		\item In this case, using Theorem \ref{lem: conditionSym}, for error minimizing POVM, we get the condition
		$\sigma^{1/2}\Lr\sigma^{1/2}=0,\sigma^{1/2}\Ls\sigma^{1/2}\in\mathcal{P}(\Tau^{\min})$. Now, using equivalence shown in \eqref{eq: generalCondition}, we get
		$$\Lr\in\mathcal{P}(\I-\Pi_\sigma),\Ls\in\mathcal{P}(\I-\Pi_\sigma+\proj^{\min}).$$
		A general $\Ls$ can be chosen as $\displaystyle\Ls={c}\frac{\psi_{\min}}{\tr\lr{\psi_{\min}\sigma}}$ for some $\psi_{\min}\in\mathcal{S}(\I-\Pi_\sigma+\proj^{\min}),c\geq0$. On putting the condition that $\Lr+\Ls\leq \I$, we obtain $\Lr\leq \I-\Ls$, and for any such $\Lr$ to exist, it has to be ensured that $\Ls\leq \I$ and thus $\ds c\frac{\psi_{\min}}{\tr\lr{\psi_{\min}\sigma}}\leq \I\Rightarrow c\leq\left\|\frac{\psi_{\min}}{\tr\lr{\psi_{\min}\sigma}}\right\|_{\infty}^{-1}$. Writing these conditions together, we get an error-minimizing measurement as
		$$\Ls=c\frac{\psi_{\min}}{\tr\lr{\psi_{\min}\sigma}},\Lr\in\mathcal{P}(\I-\Pi_\sigma),\Lr\leq \I-\Ls,c\leq\left\|\frac{\psi_{\min}}{\tr\lr{\psi_{\min}\sigma}}\right\|_{\infty}^{-1}.$$
		
		\item In this case, using Theorem \ref{lem: conditionSym}, for error minimizing POVM, we get the condition $\sigma^{1/2}\Lr\sigma^{1/2}\in\mathcal{P}(\Tau^{\max})$ and $\sigma^{1/2}\Ls\sigma^{1/2}\in\mathcal{P}(\Tau^{\min})$. Now, using equivalence shown in \eqref{eq: generalCondition}, we get
		$$\Lr\in\mathcal{P}(\I-\Pi_\sigma+\proj^{\max}),\Ls\in\mathcal{P}(\I-\Pi_\sigma+\proj^{\min}).$$	
		Here, a general $\Lr$ and $\Ls$ can be chosen as $\displaystyle\Lr={cc_r}\frac{\psi_{\max}}{\tr\lr{\psi_{\max}\sigma}}$ and $\displaystyle\Ls={c(1-c_r)}\frac{\psi_{\min}}{\tr\lr{\psi_{\min}\sigma}}$ for some $\psi_{\max}\in\mathcal{S}(\I-\Pi_\sigma+\proj^{\max}),\psi_{\min}\in\mathcal{S}(\I-\Pi_\sigma+\proj^{\min}),c\geq0,c_r\in[0,1]$ and putting the condition that $\Lr+\Ls\leq \I$, we get
		$$\Lr=cc_r\frac{\psi_{\max}}{\tr\lr{\psi_{\max}\sigma}},\Ls=c(1-c_r)\frac{\psi_{\min}}{\tr\lr{\psi_{\min}\sigma}},\displaystyle c\leq\left\|c_r\frac{\psi_{\max}}{\tr\lr{\psi_{\max}\sigma}}+(1-c_r)\frac{\psi_{\min}}{\tr\lr{\psi_{\min}\sigma}}\right\|_{\infty}^{-1}.$$
	\end{enumerate}
	
	Writing this as a set, we obtain the result stated in the theorem.
\end{proof}
\end{theorem}
	
{\Blue	\begin{remark}	
		While the work in \cite{regula2022postselected} gives an example that achieves the minimum error, we have given the most general structure of error-minimizing measurements. The example of error-minimizing measurement given in \cite{regula2022postselected} satisfies the condition for error-minimizing measurements given in Theorem \ref{lem: conditionSym} and belongs to the set of error-minimizing measurements given in Theorem \ref{lem: setSym}. So, while the error obtained by example in \cite{regula2022postselected} and our obtained set of measurements $\mathcal{E}_s$ is the same, \cite{regula2022postselected} gives a single example, whereas we provide the complete set of all error-minimizing measurements.
	\end{remark}
}	
\noindent \textbf{Constructing an arbitrary error minimizing measurement: }{Theorem \ref{lem: setSym} gives an intuitive way to construct any error minimizing measurement, which is described in Table \ref{table: paraEqual}. The first column shows the parameters and the following three show how to choose them for different case $\mathcal{C}1, \mathcal{C}2$, and $\mathcal{C}3$. To construct an arbitrary error minimizing measurement, we first select $\psi_{\max}$ and $\psi_{\min}$ according to the applicable case, then choose $c_r$ and $c$ as given in the second row and then construct $\Lr$ and $\Ls$ as given in the last row. The measurement $\{\Lr,\Ls,\I-\Lr-\Ls\}$ thus constructed will be an error minimizing measurement.

\begin{table}[H]
	\centering
	\begin{tabular}{||c||c|c|c||}
		\hline\hline
		{Cases} &$\mathcal{C}1 $&$\mathcal{C}2 $&$\mathcal{C}3 $\\
		\hline\hline
		\multirow{2}{*}{$\psi_{\max}$,} &\multirow{2}{*}{$\psi_{\max}\in\mathcal{S}(\I-\Pi_\sigma+\proj^{\max})$,}& \multirow{2}{*}{$\psi_{\max}$ not needed,}&	\multirow{2}{*}{$
			\psi_{\max}\in\mathcal{S}(\I-\Pi_\sigma+\proj^{\max}),$}
		\\
		\multirow{2}{*}{$\psi_{\min}$} &\multirow{2}{*}{$\psi_{\min}$ not needed}&\multirow{2}{*}{$\psi_{\min}\in\mathcal{S}(\I-\Pi_\sigma+\proj^{\min})$}&	\multirow{2}{*}{	$
		\psi_{\min}\in\mathcal{S}(\I-\Pi_\sigma+\proj^{\min})
		$}\\&&&\\\hline\multirow{2}{*}{$c_r$,}&\multirow{2}{*}{$c_r$ not needed,}&\multirow{2}{*}{$c_r$ not needed,}&\multirow{2}{*}{$c_r\in[0,1]$,}\\&&&\\
		$c$& $\displaystyle c\leq\left\|\frac{\psi_{\max}}{\tr\lr{\psi_{\max}\sigma}}\right\|_{\infty}^{-1}$& $\displaystyle c\leq\left\|\frac{\psi_{\min}}{\tr\lr{\psi_{\min}\sigma}}\right\|_{\infty}^{-1}$&$\displaystyle c\leq\left\|\frac{c_r\psi_{\max}}{\tr\lr{\psi_{\max}\sigma}}+\frac{(1-c_r)\psi_{\min}}{\tr\lr{\psi_{\min}\sigma}}\right\|_{\infty}^{-1}$\\&&&\vspace{-2mm}\\
		\hline
		\multirow{3}{*}{$\Lr,\Ls$} & $\displaystyle
			\Lr=c\frac{\psi_{\max}}{\tr\lr{\psi_{\max}\sigma}},$& $\displaystyle
		\Ls=c\frac{\psi_{\min}}{\tr\lr{\psi_{\min}\sigma}},$&$\displaystyle
	 \Lr=cc_r\frac{\psi_{\max}}{\tr\lr{\psi_{\max}\sigma}},$\\

	& $\ds\Ls\in\mathcal{P}(\I-\Pi_\sigma),\Ls\leq \I-\Lr$& $\displaystyle\Lr\in\mathcal{P}(\I-\Pi_\sigma),\Lr\leq \I-\Ls$&$\displaystyle
	\Ls=c(1-c_r)\frac{\psi_{\min}}{\tr\lr{\psi_{\min}\sigma}}$\\ 
	\hline\hline
	\end{tabular}
\caption{Parameterization of error minimizing measurements and method to construct an arbitrary error minimizing measurement.}
\label{table: paraEqual}
\end{table}
\noindent Note that this method giving the general and parametric form of measurements, provides us with freedom in choosing the measurement by varying $\psi_{\max},\psi_{\min},c$ and $c_r$.

We now give another example that illustrates how to find the set of error-minimizing measurements for the system given in the Example \ref{exmp: q1}.

\begin{example2} \label{exmp: q2} Applying definition of $\proj^{\max}$ and $\proj^{\min}$ given in Theorem \ref{lem: setSym} to the system described in Example \ref{exmp: q1}, we get
	\begin{align*}
		\proj^{\max}&=\frac{1}{12+4\sqrt{7}}\lr{|0\rangle+(2+\sqrt{7})|1\rangle}\lr{\langle0|+(2+\sqrt{7})\langle1|} \text{ and }\\
		\proj^{\min}&=\frac{1}{12-4\sqrt{7}}\lr{|0\rangle+(2-\sqrt{7})|1\rangle}\lr{\langle0|+(2-\sqrt{7})\langle1|}.
	\end{align*}
	Further, from theorem \ref{lem: setSym}, we have $\psi_{\max}\in\mathcal{S}(\I-\Pi_\sigma+\proj^{\max})$ and $\psi_{\min}\in\mathcal{S}(\I-\Pi_\sigma+\proj^{\min})$. Note that $\Pi_\sigma=\I$ and $\proj^{\max}$ and $\proj^{\min}$ are rank 1, so the only element in $\mathcal{S}(\I-\Pi_\sigma+\proj^{\max}) = \mathcal{S}(\proj^{\max})$ is $\proj^{\max}$ and the only element in $\mathcal{S}(\proj^{\min})$ is $\proj^{\min}$. Hence, we get $\psi_{\max}=\proj^{\max}$ and $\psi_{\min}=\proj^{\min}$. Further $\tr\lr{\psi_{\max}\sigma}=\tr\lr{\proj^{\max}\sigma}=\frac{1}{4}\frac{14+4\sqrt{7}}{12+4\sqrt{7}}$ and  $\tr\lr{\psi_{\min}\sigma}=\tr\lr{\proj_{\min}\sigma}=\frac{1}{4}\frac{14-4\sqrt{7}}{12-4\sqrt{7}}$.  So, the set of error minimizing measurements can be given as below.
	\begin{enumerate} 
		\item If $\ds p_\rho>1/2: \left\{\{c\frac{4(12+4\sqrt{7})}{14+4\sqrt{7}}\proj^{\max},0,\I-c\frac{4(12+4\sqrt{7})}{14+4\sqrt{7}}\proj^{\max}\}: 0<c\leq \frac{1}{4}\frac{14+4\sqrt{7}}{12+4\sqrt{7}} \right\}$.
		\item If $\ds p_\rho<1/2: \left\{\{0,c\frac{4(12-4\sqrt{7})}{14-4\sqrt{7}}\proj^{\min},\I-c\frac{4(12-4\sqrt{7})}{14-4\sqrt{7}}\proj^{\min}\}: 0<c\leq \frac{1}{4}\frac{14-4\sqrt{7}}{12-4\sqrt{7}} \right\}$.
		\item If $ p_\rho=1/2:\\ \left\{\{cc_rc\frac{4(12+4\sqrt{7})}{14+4\sqrt{7}}\proj^{\max},c(1-c_r)c\frac{4(12-4\sqrt{7})}{14-4\sqrt{7}}\proj^{\min},  \I-cc_r\frac{4(12+4\sqrt{7})}{14+4\sqrt{7}}\proj^{\max}-c(1-c_r)\frac{4(12-4\sqrt{7})}{14-4\sqrt{7}}\proj^{\min}\}: \right.\\
		\left. \hspace{5cm} 0\leq c_r\leq 1, c\leq \left\|c_r\frac{4(12+4\sqrt{7})}{14+4\sqrt{7}}\proj^{\max}+(1-c_r)\frac{4(12-4\sqrt{7})}{14-4\sqrt{7}}\proj^{\min}\right\|_{\infty}^{-1} \right\}$.
	\end{enumerate}
	In the above description of the set, it is clear that:
	\begin{enumerate}	
		\item If $p_\rho>1/2:$ the maximum value of c is $ \frac{1}{4}\frac{14+4\sqrt{7}}{12+4\sqrt{7}} = 0.272$.
		\item If $p_\rho<1/2:$ the maximum value of c is $ \frac{1}{4}\frac{14-4\sqrt{7}}{12-4\sqrt{7}} = 0.603$.
		\item If $p_\rho=1/2:$ the maximum value of c is $ \max_{0\leq c_r\leq 1} \left\|c_r\frac{4(12+4\sqrt{7})}{14+4\sqrt{7}}\proj^{\max}-(1-c_r)\frac{4(12-4\sqrt{7})}{14-4\sqrt{7}}\proj^{\min}\right\|_{\infty}^{-1} $.
	\end{enumerate}	
	We are able to obtain a numerical value of bound on $c$ in the first two cases because $\proj^{\max}$ and $\proj^{\min}$ are rank 1 operators. Usually, the depends on $\psi_{\max}$ and $\psi_{\min}$. We will see in the next subsection that acceptance depends only on $c$ except for the third case. So, easily obtaining the maximum value of $c$ helps find the maximum value of acceptance. 
\end{example2}

\subsection{The maximum acceptance achieved by error-minimizing measurements}
Note that Theorem \ref{lem: setSym} gives the set of all measurements that minimize the postselected symmetric error as $\mathcal{E}_s(\rho,\sigma,p)$ and gives their parameterization in terms of $(c,\psi_{\max},\Ls)$, $(c,\psi_{\min},\Lr)$ or $(c,c_r,\psi_{\max},\psi_{\min})$.} The next Lemma evaluates acceptance for a given error-minimizing measurement as a function of these parameters $c$ and $c_r$.

\begin{theorem} \label{lem: genAccSym}
	For a test utilizing a measurement $\Lambda\in \mathcal{E}_s(\rho,\sigma,p)$, as given in Theorem \ref{lem: setSym},  its acceptance under three cases is given as below.
\begin{enumerate}[label={For case $\underline{\mathcal{C}\mathrm{\arabic*}}$,}]
	\setlength{\itemindent}{2cm}
	\item $A_\rho(\Lambda)=c R_{\max}(\rho,\sigma)\text{ and  }A_\sigma(\Lambda)=c$.
	\item $A_\rho(\Lambda)=c R_{\min}(\rho,\sigma)\text{ and  }A_\sigma(\Lambda)=c$.
	\item $A_\rho(\Lambda)=cc_rR_{\max}(\rho,\sigma)+c(1-c_r)R_{\min}(\rho,\sigma)\text{ and  }A_\sigma(\Lambda)=c$.
\end{enumerate}
	Here, $c$ and $c_r$ depend on $\Lambda$ and are as given in the Table \ref{table: paraEqual}.
\begin{proof} Using Theorem \ref{lem: genMainProj} from Appendix \ref{sec: genProj}, we have $$\psi_{\max}\in\mathcal{S}(\I-\Pi_\sigma+\proj^{\max})\subseteq\mathcal{P}(\I-\Pi_\sigma+\proj^{\max})\Rightarrow \sigma^{1/2} \psi_{\max} \sigma^{1/2}\in \mathcal{P}(\Tau^{\max}).$$
	Thus  $\tr\lr{\psi_{\max}\rho}=\tr\lr{\psi_{\max}\Pi_\rho\rho\Pi_\rho}\stackrel{(a)}{=}\tr\lr{\psi_{\max}\Pi_\sigma\rho\Pi_\sigma}$ 
	\begin{align*}
		&\stackrel{(b)}{=}\tr\lr{\sigma^{1/2}\psi_{\max}\sigma^{1/2}\sigma^{-1/2}\rho\sigma^{-1/2}}\\
		&\stackrel{(c)}{=}\tr\lr{\sigma^{1/2}\psi_{\max}\sigma^{1/2}\Tau^{\max}\sigma^{-1/2}\rho\sigma^{-1/2}}\\
		&\stackrel{(d)}{=}R_{\max}(\rho,\sigma)\tr\lr{\sigma^{1/2}\psi_{\max}\sigma^{1/2}\Tau^{\max}}\\
		&\stackrel{(e)}{=}R_{\max}(\rho,\sigma)\tr\lr{\sigma^{1/2}\psi_{\max}\sigma^{1/2}}.
	\end{align*}
	Here, $(a)$ is due to $\Pi_\sigma=\Pi_\rho$. $(b)$ is obtained from $\Pi_\sigma=\sigma^{-1/2}\sigma^{1/2}$ and using cyclic property of trace. $(c)$ and $(e)$ uses that $\sigma^{1/2} \psi_{\max} \sigma^{1/2}\in \mathcal{P}(\Tau^{\max})$. $(d)$ follows from the fact that $\Tau^{\max}$ is the projection onto the subspace corresponding to the maximum eigenvalue of $\sigma^{-1/2}\rho\sigma^{-1/2}$ and so, we have $\Tau^{\max}\sigma^{-1/2}\rho\sigma^{-1/2}=R_{\max}(\rho,\sigma)\Tau^{\max}$. Thus, we obtain 
	\begin{align}
		\tr\lr{\psi_{\max}\rho}=R_{\max}(\rho,\sigma)\tr\lr{\psi_{\max}\sigma}. \label{eq: psiMaxRho}
	\end{align}
	Similarly, we can show that
	\begin{align}
		\tr\lr{\psi_{\min}\rho}=R_{\min}(\rho,\sigma)\tr\lr{\psi_{\min}\sigma}. \label{eq: psiMinRho}
	\end{align}
	So, for any $\Gamma$ such that
	\begin{gather}
		\Gamma=c\frac{\psi_{\max}}{\tr\lr{\psi_{\max}\sigma}} \Rightarrow \tr\lr{\Gamma\rho}=c\frac{\tr\lr{\psi_{\max}\rho}}{\tr\lr{\psi_{\max}\sigma}}=cR_{\max}(\rho,\sigma) \text{ and }\tr\lr{\Gamma\sigma}=c\frac{\tr\lr{\psi_{\max}\sigma}}{\tr\lr{\psi_{\max}\sigma}}=c. \label{eq: gammaMax}
		\intertext{Similarly, for any $\ds \Gamma=c\frac{\psi_{\min}}{\tr\lr{\psi_{\min}\sigma}}$, we get }
		\tr\lr{\Gamma\rho}=cR_{\min}(\rho,\sigma) \text{ and }\tr\lr{\Gamma\sigma}=c. \label{eq: gammaMin}
		\intertext{On taking }
		\Gamma\in\mathcal{P}(\I-\Pi_{\sigma}) \Rightarrow \tr\lr{\Gamma\rho}=0 \text{ and }\tr\lr{\Gamma\sigma}=0. \label{eq: gammaVoid}
	\end{gather}
	
	We will use these results from \eqref{eq: gammaMax}, \eqref{eq: gammaMin}, and \eqref{eq: gammaVoid} directly for $\Gamma=\Lr$ and $\Gamma=\Ls$. Now, for three cases as given in Theorem \ref{lem: conditionSym}, acceptance is given as follows.
	
	\begin{enumerate}[label={$\underline{\mathcal{C}\mathrm{\arabic*}}$}]
	
		\item Recall that $\Lr$ and $\Ls$ are parameterized as
		$\ds\Lr=c\frac{\psi_{\max}}{\tr\lr{\psi_{\max}\sigma}},\Ls\in\mathcal{P}(\I-\Pi_\sigma),\Ls\leq \I-\Lr$. Hence, corresponding acceptance is 
		\begin{align*}
		A_\rho(\Lambda)&=\tr\lr{\Lr\rho}+\tr\lr{\Ls\rho}=c R_{\max}(\rho,\sigma)+0=c R_{\max}(\rho,\sigma),\\
		A_\sigma(\Lambda)&=\tr\lr{\Lr\sigma}+\tr\lr{\Ls\sigma}=c +0=c.
		\end{align*}
	
		\item Recall that $\Lr$ and $\Ls$ are parameterized as
		$\ds\Ls=c\frac{\psi_{\min}}{\tr\lr{\psi_{\min}\sigma}},\Lr\in\mathcal{P}(\I-\Pi_\sigma),\Lr\leq \I-\Ls$. Hence, corresponding acceptance is 
		\begin{align*}	
		A_\rho(\Lambda)&=\tr\lr{\Lr\rho}+\tr\lr{\Ls\rho}=c R_{\min}(\rho,\sigma)+0=c R_{\min}(\rho,\sigma),\\
		A_\sigma(\Lambda)&=\tr\lr{\Lr\sigma}+\tr\lr{\Ls\sigma}=c+0=c.
		\end{align*}	
	
		\item Recall that $\Lr$ and $\Ls$ are parameterized as
		$\ds \Lr=cc_r\frac{\psi_{\max}}{\tr\lr{\psi_{\max}\sigma}},\Ls=c(1-c_r)\frac{\psi_{\min}}{\tr\lr{\psi_{\min}\sigma}},\Lr+\Ls\leq \I$. Hence, corresponding acceptance is given as
		\begin{align*}	
		A_\rho(\Lambda)&=\tr\lr{\Lr\rho}+\tr\lr{\Ls\rho}=cc_rR_{\max}(\rho,\sigma)+c(1-c_r)R_{\min}(\rho,\sigma),\\
		A_\sigma(\Lambda)&=\tr\lr{\Lr\sigma}+\tr\lr{\Ls\sigma}=c.
		\end{align*}		
		
	\end{enumerate} \vspace{-7mm}
\end{proof}
\end{theorem}
\begin{remark}
Barring the case $R_{\max}(p_\rho\rho,p_\sigma\sigma)= R_{\max}(p_\sigma\sigma,p_\rho\rho)$, the expression for acceptance depends only on $c$ linearly. In turn, $c$ is upper bounded by a function of $\psi_{\max}, \psi_{\min}$ and $c_r$. So, to maximize acceptance, we only need to focus on maximizing $c$ for all error-minimizing measurements. 
\end{remark}

We now give the maximum value of acceptance that an error-minimizing measurement can have,  in the following theorem.

\begin{theorem} \label{thm: symmetric}
	For an error-minimizing measurement $\Lambda$, (i.e., $e(\Lambda)=e_s(\rho,\sigma,p), \Lambda\in\mathcal{M}$), the maximum value of acceptance for the states $\rho$ and $\sigma$ is given below.
	\begin{enumerate}[label={$\underline{\mathcal{C}\mathrm{\arabic*}}$}]
		\item $\ds A^s_\rho= R_{\max}(\rho,\sigma) \tr\lr{\proj^{\max}\sigma}\quad\text{ and}\quad  A^s_\sigma=\tr\lr{\proj^{\max}\sigma}.$
		
		\item $\ds A^s_\rho= R_{\min}(\rho,\sigma) \tr\lr{\proj^{\min}\sigma}\quad\text{ and}\quad A^s_\sigma=\tr\lr{\proj^{\min}\sigma}.$
		
		\item $A^s_\sigma=\max_{c_r\in[0,1]}\Upsilon_{\Tau^{\max},\Tau^{\min}}(\sigma,c_r)$ and 
		$$A^s_\rho=\max_{c_r\in[0,1]}(c_rR_{\max}(\rho,\sigma)+(1-c_r)R_{\min}(\rho,\sigma))\Upsilon_{\Tau^{\max},\Tau^{\min}}(\sigma^{1/2}\Pi_{\proj^{\max}+\proj^{\min}}\sigma^{1/2},c_r).$$
	\end{enumerate}
	Here, $\Upsilon_{\Pi_1,\Pi_2}(\sigma,r)$ is defined as %for any $\tr\lr{\Pi_1\Pi_2}=0$, we have
	$$\Upsilon_{\Pi_1,\Pi_2}(\sigma,r)\isdefinedas\{\max c: c(r\psi_1+(1-r)\psi_2)\leq\sigma \text{ for some }\psi_1\in\mathcal{S}(\Pi_1),\psi_2\in\mathcal{S}(\Pi_2), \text{ with }\tr\lr{\Pi_1\Pi_2}=0\}.$$
	More detailed properties of $\Upsilon_{\Pi_1,\Pi_2}\lr{\sigma,r}$
	are given in Appendix \ref{sec: upsilon}, which can help find closed-form expressions of maximum acceptance in specific cases.

\begin{proof} Let us take the three cases separately.	
\begin{enumerate}[label={$\underline{\mathcal{C}\mathrm{\arabic*}}$}]
	\item From Lemma \ref{lem: genAccSym},
	$A^s_\rho= R_{\max}(\rho,\sigma)\max_{\mathcal{E}_s(\rho,\sigma,p)} c\qquad \text{ and }\qquad A^s_\sigma=\max_{\mathcal{E}_s(\rho,\sigma,p)} c.$ Now, from Lemma \ref{lem: singleMin} in Appendix \ref{sec: maxNormMinimization}, we have
	$$\max_{\mathcal{E}_s(\rho,\sigma,p)} c=\max_{\psi_{\max}\in\mathcal {P}(\I-\Pi_\sigma+\proj^{\max})}\left\|\frac{\psi_{\max}}{\tr\lr{\psi_{\max}\sigma}}\right\|_{\infty}^{-1}=\tr\lr{\proj^{\max}\sigma}.$$
	Substituting it in the expression of $A_\rho^s$ and $A_\sigma^s$, as given in Theorem \ref{lem: genAccSym}, we get the result stated in the theorem.
	\item From Lemma \ref{lem: genAccSym},
	$A^s_\rho= R_{\min}(\rho,\sigma)\max_{\mathcal{E}_s(\rho,\sigma,p)} c\qquad \text{ and }\qquad A^s_\sigma=\max_{\mathcal{E}_s(\rho,\sigma,p)} c.$ Now, from Lemma \ref{lem: singleMin} in Appendix \ref{sec: maxNormMinimization}, we have
	$$\max_{\mathcal{E}_s(\rho,\sigma,p)} c=\max_{\psi_{\min}\in\mathcal {P}(\I-\Pi_\sigma+\proj^{\min})}\left\|\frac{\psi_{\min}}{\tr\lr{\psi_{\min}\sigma}}\right\|_{\infty}^{-1}=\tr\lr{\proj^{\min}\sigma}.$$
	Substituting it in the expression of $A_\rho^s$ and $A_\sigma^s$, as given in Theorem \ref{lem: genAccSym}, we get the result stated in the theorem.
	\item From Theorem \ref{lem: genAccSym},
	$A_\rho(\Lambda)= cc_rR_{\max}(\rho,\sigma)+c(1-c_r)R_{\min}(\rho,\sigma)\text{ and }A_\sigma(\Lambda)=c.$
	Note that these functions have to be maximized with respect to $c$ and $c_r$. Taking $c$ common and maximizing it first, we get
	\begin{align*}
		A^s_\rho&=\max_{\mathcal{E}_s(\rho,\sigma,p)} cc_rR_{\max}(\rho,\sigma)+c(1-c_r)R_{\min}(\rho,\sigma)\\
		&=\max_{c_r\in[0,1]}c_rR_{\max}(\rho,\sigma)+(1-c_r)R_{\min}(\rho,\sigma) \max_{\substack{\psi_{\max}\in\mathcal{S}(\I-\Pi_\sigma+\proj^{\max}),\\\psi_{\min}\in\mathcal{S}(\I-\Pi_\sigma+\proj^{\min})}}\left\|\frac{{c_r}\psi_{\max}}{\tr\lr{\psi_{\max}\sigma}}+\frac{\lr{1-c_r}\psi_{\min}}{\tr\lr{\psi_{\min}\sigma}}\right\|_{\infty}^{-1} \\
		&=\max_{c_r\in[0,1]}(c_rR_{\max}(\rho,\sigma)+(1-c_r)R_{\min}(\rho,\sigma))\Upsilon_{\Tau^{\max},\Tau^{\min}}(\sigma,c_r).
	\end{align*}
Here, the last step is obtained from Lemma \ref{lem: doubleMin} in Appendix \ref{sec: maxNormMinimization}. Now maximizing $A_\sigma(\Lambda)$, we obtain
	\begin{align*}
		A_\sigma^s&=\max_{\mathcal{E}_s(\rho,\sigma,p)} c=\max_{c_r\in[0,1]}\max_{\substack{\psi_{\max}\in\mathcal{S}(\I-\Pi_\sigma+\proj^{\max}),\\\psi_{\min}\in\mathcal{S}(\I-\Pi_\sigma+\proj^{\min})}}\left\|c_r\frac{\psi_{\max}}{\tr\lr{\psi_{\max}\sigma}}+(1-c_r)\frac{\psi_{\min}}{\tr\lr{\psi_{\min}\sigma}}\right\|_{\infty}^{-1}\\
		&=\max_{c_r\in[0,1]}\Upsilon_{\Tau^{\max},\Tau^{\min}}(\sigma^{1/2}\Pi_{\proj^{\max}+\proj^{\min}}\sigma^{1/2},c_r).
	\end{align*}
	Here, the last step is due to Lemma \ref{lem: doubleMin} in Appendix \ref{sec: maxNormMinimization}.
\end{enumerate} \vspace{-6mm}
\end{proof}
\end{theorem}

{\Blue 
	\begin{remark} \label{rem: desiredAccEq}
		Note that for any desired value of acceptance, error-minimizing measurements can be designed by choosing the appropriate value of constant $c$ and $c_r$, provided that the desired acceptance remains less than the maximum obtained in Theorem \ref{thm: symmetric}.
	\end{remark}
}

\begin{remark}
Note that for the case $\mathcal{C}1$, the expressions of acceptance for $\rho$ and $\sigma$, both are linear in the parameter $c$ in Theorem \ref{lem: genAccSym}, so the parameter-values that maximize acceptance for both states $\rho$ and $\sigma$ is the same. Thus, the same measurements maximize both acceptances. This observation also holds for the case $\mathcal{C}2$. In contrast, in the case $\mathcal{C}3$, we obtained acceptances as different functions of $c$ and $c_r$. So, the measurements that maximize acceptance for the state $\rho, i.e., A_\rho(\Lambda)$, will not be maximizing acceptance for the state $\sigma, i.e., A_\sigma(\Lambda)$. Similarly, the measurement that maximizes the acceptance for the state $\sigma, i.e., A_\sigma(\Lambda)$ will not maximize acceptance for the state $\rho, i.e., A_\rho(\Lambda)$.
\end{remark}

\begin{remark} \label{rem: obtainsMaxAccEq}
	The maximum acceptance for the two state can simultaneously be obtained by taking the measurement described by the POVM $\Lambda=\{\proj^{\max},0,\I-\proj^{\max}\}$ and $\Lambda=\{0,\proj^{\min},\I-\proj^{\min}\}$ in $\mathcal{C}1$ and $\mathcal{C}2$ respectively. In $\mathcal{C}3$, finding the measurement that achieves the maximum acceptance for the state $\rho$ (or $\sigma$) is involved. First step is finding $c_r$ such that $A_\rho^s$ (or $A_\sigma^s$) is obtained and placing corresponding optimal $\psi_{\max}$ and $\psi_{\min}$ and $c=\Upsilon_{\Tau^{\max},\Tau^{\min}}(\sigma,c_r)$, the maximum acceptance achieving measurement can be obtained by solving the optimization step. In $\mathcal{C}3$, the measurements that maximize acceptance for state $\rho$ and $\sigma$ are different.
\end{remark}

Now, having completed the analysis for the pair of states with the same support, in the next section, we maximize the acceptance for the pair of states $\rho$ and $\sigma$ such that $\Pi_\rho\neq\Pi_\sigma$.

\section{Case: $\rho$ and $\sigma$ do not have the same support i.e. $\Pi_\rho\neq\Pi_\sigma$} \label{sec: rhoNeqsigma}
Recall from \eqref{eq: symBase} that the minimum postselected symmetric error $e_s(\rho,\sigma,p)=\lr{ \Xi\left(p_\rho\rho,p_\sigma\sigma\right)+1}^{-1}$. 
If $\Pi_\rho\neq\Pi_\sigma$, from the definition given in \eqref{eq: defXi}, it turns out that $\Xi\left(p_\rho\rho,p_\sigma\sigma\right)=\infty$, hence $e_s(\rho,\sigma,p)=0$. In this section, we begin with finding the condition on $\Lambda$ to obtain $e(\Lambda)=e_s(\rho,\sigma,p)=0$. The following theorem states the set of all possible error-minimizing measurements as a union of three sets, each corresponding to one of the three conditions that ensures that the error vanishes. An arbitrary measurement must belong to one of three sets $\mathcal{E}_s^1(\rho,\sigma)$, $\mathcal{E}_s^2(\rho,\sigma)$ or $\mathcal{E}_s^3(\rho,\sigma)$ defined in the Theorem below to ensure that $e(\Lambda)=0$. A parameterization of measurements in these sets is described later in Lemma \ref{lem: para1}, \ref{lem: para2} and \ref{lem: para3} to allow construction of an arbitrary measurement from each of these sets.
\begin{theorem} \label{thm: beyondSymmetricSplit}
	For any measurement $\Lambda=\{\Lambda_\rho,\Lambda_\sigma,\I-\Lambda_\rho-\Lambda_\sigma\}\in\mathcal{M}$, $e(\Lambda)=0$ if and only if 
	\begin{gather}
		\Lambda\in \mathcal{E}_s(\rho,\sigma,p)=\mathcal{E}^1_{s}(\rho,\sigma)\cup \mathcal{E}^2_{s}(\rho,\sigma)\cup \mathcal{E}^3_{s}(\rho,\sigma),\text{ with }
	\end{gather}
	\vspace*{-8mm}
	\begin{align*}
		\mathcal{E}^1_{s}(\rho,\sigma)&\isdefinedas\{\Lambda:\Lr\in\mathcal{P}(\I-\Pi_\sigma),\Ls\in\mathcal{P}(\I-\Pi_{\rho+\sigma}),\tr\lr{\Lr\rho}\neq0,\Lr+\Ls\leq\I\},\\
		\mathcal{E}^2_{s}(\rho,\sigma)&\isdefinedas\{\Lambda:\Lr\in\mathcal{P}(\I-\Pi_{\rho+\sigma}),\Ls\in\mathcal{P}(\I-\Pi_\rho),\tr\lr{\Ls\sigma}\neq0,\Lr+\Ls\leq\I\},\\
		\mathcal{E}^3_{s}(\rho,\sigma)&\isdefinedas\{\Lambda:\Lr\in\mathcal{P}(\I-\Pi_\sigma),\Ls\in\mathcal{P}(\I-\Pi_\rho),\tr\lr{\Lr\rho}\neq0,\tr\lr{\Ls\sigma}\neq0,\Lr+\Ls\leq\I\}.
	\end{align*}
\begin{proof} The proof is given in Appendix \ref{proof: beyondSymmetricSplit}.
\end{proof}
\end{theorem}

%{\Blue 
%	\begin{remark}
%		On taking the number of states to be discriminated as $2$, USD measurements obtained in \cite{chefles1998unambiguous,feng2004unambiguous,stojnic2007unambiguous,bandyopadhyay2014unambiguous,jafarizadeh2008optimal} belong to the set given in Theorem \ref{thm: beyondSymmetricSplit}.
%	\end{remark}
%}

\begin{remark} \label{rm: errorMinMeasurementProperties} The three sets described in the Theorem \ref{thm: beyondSymmetricSplit} have a specific structure of outcomes, which is described below (also see Fig. \ref{fig:threeSets}).
	\begin{enumerate}[]
		\item The first set $\mathcal{E}_s^1(\rho,\sigma)$ consists of all the measurements for which only $\tr\lr{\Lr\rho}$ i.e. the probability of detecting $\rho$ when the given state is $\rho$ is non-zero and the probabilities $\tr\lr{\Lr\sigma}$, $\tr\lr{\Ls\rho}$ and $\tr\lr{\Ls\sigma}$ vanish. So, for any measurement that belongs to this set, if the given input state is $\rho$, the measurement either declares in favor of $\rho$ or doesn't make any decision (i.e. reject the test). However, if the given state is $\sigma$, it always rejects both the hypotheses. Fig. \ref{fig:threeSets}\text{(a)} illustrates this point.
		
		\item Similarly, the second set i.e. $\mathcal{E}_s^2(\rho,\sigma)$ consists of all the measurements for which only $\tr\lr{\Ls\sigma}$ i.e. the probability of detecting $\sigma$ when the given state is $\sigma$ is non-zero and the probabilities $\tr\lr{\Lr\rho}$, $\tr\lr{\Lr\sigma}$, and $\tr\lr{\Ls\rho}$ vanish. So, for any measurement that belongs to this set, if the given input state is $\sigma$, the measurement either declares in favor of $\sigma$ or reject the test. However, if the given state is $\rho$, it always rejects both the hypotheses. Fig. \ref{fig:threeSets}\text{(b)} illustrates this point.
		
		\item The third set i.e. $\mathcal{E}_s^3(\rho,\sigma)$ consists of all the measurements for which both $\tr\lr{\Lr\rho}$ and $\tr\lr{\Ls\sigma}$ are non-zero. So, for the given state $\rho$ (or $\sigma$), it either declares outcome as $\rho$ (or $\sigma$), or rejects both the hypotheses. Fig. \ref{fig:threeSets}\text{(c)} illustrates this point.
	\end{enumerate}
\end{remark}

\begin{figure}[h]
	\centering
	\begin{subfigure}[b]{0.25\textwidth}
		\centering
		\includegraphics[scale=0.6]{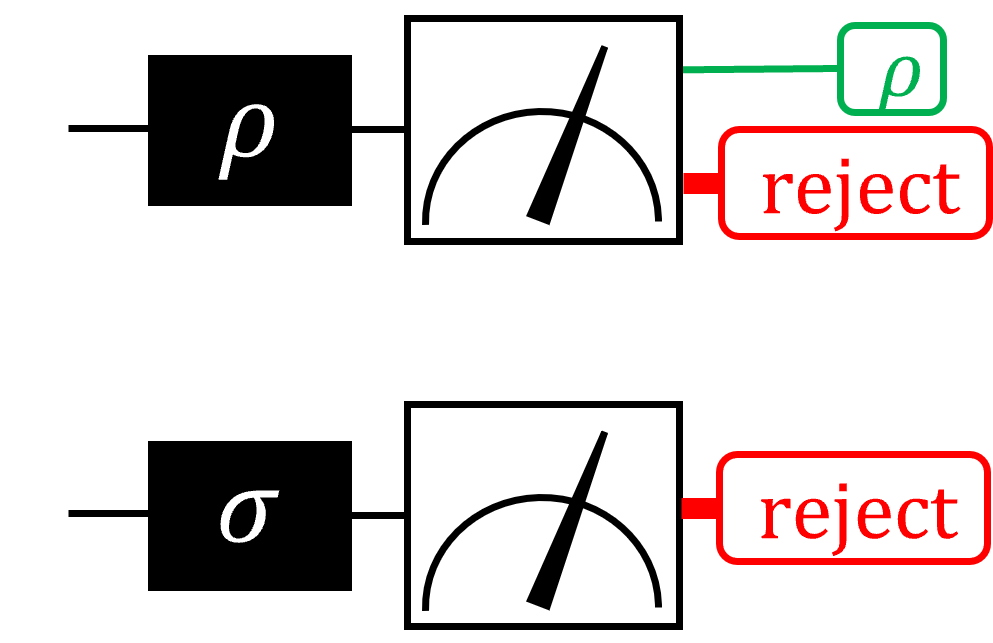}
		\caption{$\mathcal{E}^1_{s}(\rho,\sigma)$}
		\label{fig:threeSets1}
	\end{subfigure}
	\hspace{2mm}
	\begin{subfigure}[b]{0.25\textwidth}
		\centering
		\includegraphics[scale=0.6]{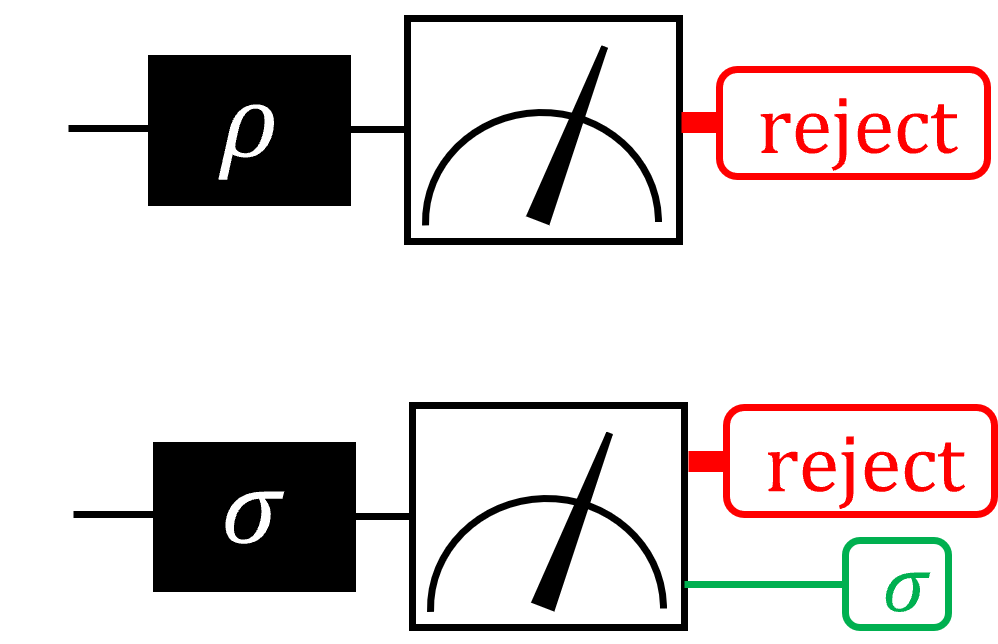}
		\caption{$\mathcal{E}^2_{s}(\rho,\sigma)$}
		\label{fig:threeSets2}
	\end{subfigure}
	\hspace{2mm}
	\begin{subfigure}[b]{0.25\textwidth}
		\centering
		\includegraphics[scale=0.6]{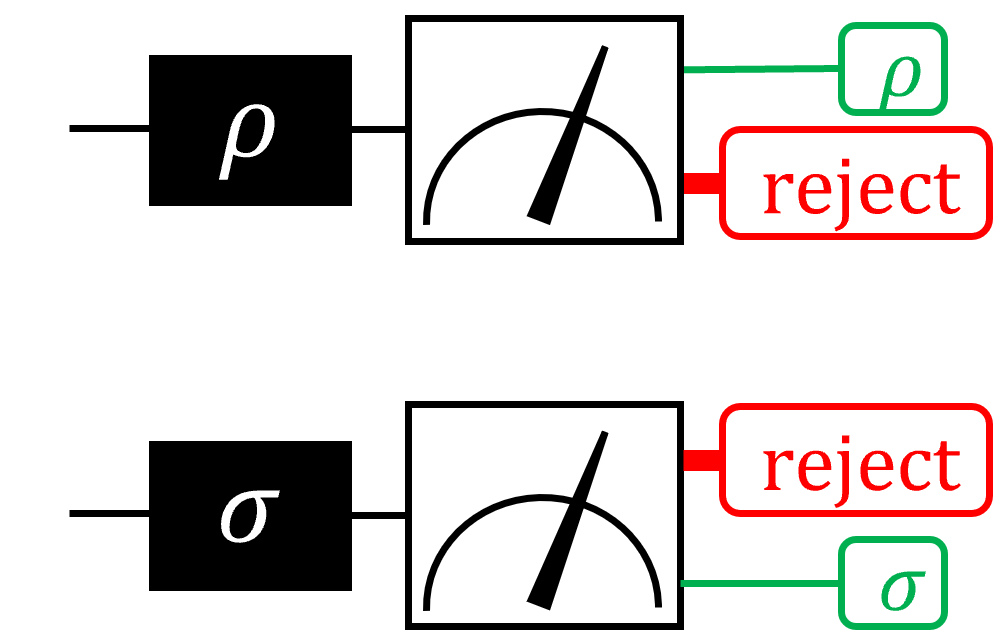}
		\caption{$\mathcal{E}^3_{s}(\rho,\sigma)$}
		\label{fig:threeSets3}
	\end{subfigure}
	\caption{The figure shows possible outcomes for an arbitrary measurement from the set $\mathcal{E}^1_{s}(\rho,\sigma)$, $\mathcal{E}^2_{s}(\rho,\sigma)$ and $\mathcal{E}^3_{s}(\rho,\sigma)$.}
	\label{fig:threeSets}
\end{figure}

\begin{remark}
	Note that for any measurement for which error vanishes, $\tr(p_\sigma\Lambda_\rho\sigma+p_\rho\Lambda_\sigma\rho)=0$, or equivalently $\tr(\Lambda_\rho\sigma)=0$ and $\tr(\Lambda_\sigma\rho)=0$. Thus, both type-1 and type-2 error must vanish. It follows from the fact that the numerator of $e(\Lambda)$ has to vanish to ensure that the error vanishes. So, such a measurement either declares the correct outcome or rejects both hypotheses but never declares the wrong outcome. {\Blue Quantum state discrimination with such requirements is called \textbf{unambiguous state discrimination (USD)}. Measurements that lead to zero probability of wrong outcome are called USD measurements. Such discrimination has been studied quite a lot in the past. On taking the number of states to be discriminated as $2$, USD measurements obtained in \cite{chefles1998unambiguous,feng2004unambiguous,stojnic2007unambiguous,bandyopadhyay2014unambiguous,jafarizadeh2008optimal} belong to the set given in Theorem \ref{thm: beyondSymmetricSplit}.} Further, note that one of $\tr\lr{\Lr\rho}$ or $\tr\lr{\Ls\sigma}$, i.e., at least one of the probability of correct decision has to be non-zero to ensure that the denominator in $e(\Lambda)$ remains non-zero.
\end{remark}

\begin{remark}
Note that the set and respective conditions depend only on the subspace spanned by eigenvectors of $\rho$ and $\sigma$ and not on the prior probabilities $p_\rho$ and $p_\sigma$. So, error-minimizing measurements are independent of the prior probabilities when $\Pi_\rho\neq \Pi_\sigma$.
\end{remark}

\begin{corollary} \label{col: case1}
	If $\Pi_\sigma<\Pi_\rho, \mathcal{E}_{s}(\rho,\sigma,p)=\mathcal{E}^1_{s}(\rho,\sigma)$. So, the set of all error-minimizing measurements is given only by the first set, i.e. $\mathcal{E}_s^1(\rho,\sigma)$. Further, for the postselected symmetric error minimizing measurements, acceptance for the state $\sigma$, $A_\sigma(\Lambda)=0$, so if the input state is $\sigma$, an error-minimizing measurement always rejects the test. Also, any error-minimizing measurement never declares the unknown state as $\sigma$. See Fig. \ref{fig:threeSets}(a) for clarity. 
	\begin{proof}
		If $\Pi_\sigma\leq\Pi_\rho\Rightarrow \Pi_{\rho+\sigma}=\Pi_\rho$. So, if $\Ls\in\mathcal{P}(\I-\Pi_\rho),$ then $\Ls\in\mathcal{P}(\I-\Pi_{\rho+\sigma})$ and so $\tr\lr{\Ls\sigma}=0$, so $\mathcal{E}_s^2(\rho,\sigma)$ and $\mathcal{E}_s^3(\rho,\sigma)$ are empty sets and $$\mathcal{E}_s(\rho,\sigma,p)=\mathcal{E}_s^1(\rho,\sigma).$$
		Now, note that for any $\Lambda\in \mathcal{E}_s^1(\rho,\sigma),\tr\lr{\Ls\sigma}=0$ and $\tr\lr{\Ls\rho}=0$, so $A_\sigma(\Lambda)=0$ for any error minimizing measurement.
	\end{proof}
\end{corollary}

\begin{corollary} \label{col: case2}
	If $\Pi_\rho<\Pi_\sigma, \mathcal{E}_{s}(\rho,\sigma)=\mathcal{E}^2_{s}(\rho,\sigma)$ and $A_\rho(\Lambda)=0$.
	So, the set of all error-minimizing measurements is given only by the second set, i.e., $\mathcal{E}_s^2(\rho,\sigma)$. Further, for the postselected symmetric error minimizing measurements, acceptance for the state $\rho$, $A_\rho(\Lambda)=0$, so if the input state is $\rho$, the error-minimizing measurement always rejects both the hypotheses. Also, any error-minimizing measurement never declares the unknown state as $\rho$. See Fig. \ref{fig:threeSets}(b) for clarity. 
	\begin{proof}Follows similar to the proof of Corollary \ref{col: case1}.
	\end{proof}
\end{corollary}
We now give a parameterization of the three sets in the following three lemmas.
\begin{lemma} \label{lem: para1} The set $\mathcal{E}^1_{s}(\rho,\sigma)$ is completely characterized as $\mathcal{E}^1_{s}(\rho,\sigma)=$
	\begin{align*}
		\Bigg\{\Lambda:&\Lr=\frac{c_1\psi_\rho}{\tr\lr{\psi_\rho\rho}},\psi_\rho\in\mathcal{S}(\I-\Pi_\sigma),\Ls\leq\I-\Lr,\Ls\in\mathcal{P}(\I-\Pi_{\rho+\sigma}),0<c_1\leq\left\|\frac{\psi_\rho}{\tr\lr{\psi_\rho\rho}}\right\|_{\infty}^{-1}\Bigg\}.
	\end{align*}
\begin{proof}From the definition of $\mathcal{E}_s^1(\rho,\sigma)$ as given in Theorem \ref{thm: beyondSymmetricSplit}, we get
	\begin{gather*}
		\Lr\in\mathcal{P}(\I-\Pi_\sigma),\Ls\in\mathcal{P}(\I-\Pi_{\rho+\sigma}),\tr\lr{\Lr\rho}\neq0,\Lr+\Ls\leq\I.
		\intertext{Taking $\psi_\rho\in\mathcal{S}(\I-\Pi_\sigma),c_1>0$, the above mentioned conditions can be parameterized as}
		\ds \Lr=\frac{c_1\psi_\rho}{\tr\lr{\psi_\rho\rho}},c_1>0,\psi_\rho\in\mathcal{S}(\I-\Pi_\sigma),\Lr\leq\I,\Ls\leq\I-\Lr,\Ls\in\mathcal{P}(\I-\Pi_{\rho+\sigma}).
		\intertext{Using $\Lr\leq \I$, we get an upper bound on $c_1$ and thus equivalent parameterization becomes}
		\ds \Lr=\frac{c_1\psi_\rho}{\tr\lr{\psi_\rho\rho}},\psi_\rho\in\mathcal{S}(\I-\Pi_\sigma),0<c_1\leq\left\|\frac{\psi_\rho}{\tr\lr{\psi_\rho\rho}}\right\|_{\infty}^{-1},\Ls\leq\I-\Lr,\Ls\in\mathcal{P}(\I-\Pi_{\rho+\sigma}).
	\end{gather*}
	This gives the desired result.
\end{proof}

\end{lemma}

\begin{lemma} \label{lem: para2}  The set $\mathcal{E}^2_{s}(\rho,\sigma)$ is completely characterized as $\mathcal{E}^2_{s}(\rho,\sigma)=$
\begin{align*}
	\Bigg\{\Lambda&:\Ls=\frac{c_2\psi_\sigma}{\tr\lr{\psi_\sigma\sigma}},\psi_\sigma\in\mathcal{S}(\I-\Pi_\rho),\Lr\leq\I-\Ls,\Lr\in\mathcal{P}(\I-\Pi_{\rho+\sigma}),0<c_2\leq\left\|\frac{\psi_\sigma}{\tr\lr{\psi_\sigma\sigma}}\right\|_{\infty}^{-1}\Bigg\}.
\end{align*}
	\begin{proof} The proof follows similar to the proof of Lemma \ref{lem: para1}.
%		From the definition of $\mathcal{E}_s^2(\rho,\sigma)$ as given in Theorem \ref{thm: beyondSymmetricSplit}, we get
%		\begin{align*}
%			&\hspace{-5mm}\Lr\in\mathcal{P}(\I-\Pi_{\rho+\sigma}),\Ls\in\mathcal{P}(\I-\Pi_\rho),\tr\lr{\Ls\sigma}\neq0,\Lr+\Ls\leq\I\\
%			&\Leftrightarrow\ds \Ls=\frac{c_2\psi_\sigma}{\tr\lr{\psi_\sigma\sigma}},c_2>0,\psi_\sigma\in\mathcal{S}(\I-\Pi_\rho),\Ls\leq\I,\Lr\leq\I-\Ls,\Lr\in\mathcal{P}(\I-\Pi_{\rho+\sigma})\\
%			&\Leftrightarrow\ds \Ls=\frac{c_2\psi_\sigma}{\tr\lr{\psi_\sigma\sigma}},\psi_\sigma\in\mathcal{S}(\I-\Pi_\rho),0<c_2\leq\left\|\frac{\psi_\sigma}{\tr\lr{\psi_\sigma\sigma}}\right\|_{\infty}^{-1},\Lr\leq\I-\Ls,\Lr\in\mathcal{P}(\I-\Pi_{\rho+\sigma}).
%		\end{align*}
%	This gives the desired result.
	\end{proof}	
\end{lemma}

\begin{lemma} \label{lem: para3} The set $\mathcal{E}^3_{s}(\rho,\sigma)$ is completely characterized as
	\begin{align*}
		\mathcal{E}^3_{s}(\rho,\sigma)=\Bigg\{\Lambda:\Lr=\frac{c_3c_r\psi_\rho}{\tr\lr{\psi_\rho\rho}}, \psi_\rho\in\mathcal{S}(\I-\Pi_\sigma),& \ \Ls=\frac{c_3(1-c_r)\psi_\sigma}{\tr\lr{\psi_\sigma\sigma}},\psi_\sigma\in\mathcal{S}(\I-\Pi_\rho),\\&
		c_r\in[0,1], 0<c_3\leq\left\|\frac{c_r\psi_\rho}{\tr\lr{\psi_\rho\rho}}+\frac{(1-c_r)\psi_\sigma}{\tr\lr{\psi_\sigma\sigma}}\right\|_{\infty}^{-1}\Bigg\}.
	\end{align*}
\begin{proof} From the definition of $\mathcal{E}_s^3(\rho,\sigma)$ as given in Theorem \ref{thm: beyondSymmetricSplit}, we get
	\begin{gather*}
		\Lr\in\mathcal{P}(\I-\Pi_\sigma),\Ls\in\mathcal{P}(\I-\Pi_\rho),\tr\lr{\Lr\rho}\neq0,\tr\lr{\Ls\sigma}\neq0,\Lr+\Ls\leq\I.
		\intertext{Taking $\psi_\rho\in\mathcal{S}(\I-\Pi_\sigma),\psi_\sigma\in\mathcal{S}(\I-\Pi_\rho),c_r\in[0,1],c_3>0$, the above mentioned conditions can be parameterized as}
		\Lr=\frac{c_3c_r\psi_\rho}{\tr\lr{\psi_\rho\rho}},\Ls=\frac{c_3(1-c_r)\psi_\sigma}{\tr\lr{\psi_\sigma\sigma}},\Lr+\Ls\leq\I.
		\intertext{Now, substituting $\Lr$ and $\Ls$ in the condition $\Lr+\Ls\leq\I$, we get this condition as a bound on $c_3$ as}
		\Lr=\frac{c_3c_r\psi_\rho}{\tr\lr{\psi_\rho\rho}},\Ls=\frac{c_3(1-c_r)\psi_\sigma}{\tr\lr{\psi_\sigma\sigma}}, c_3\leq\left\|\frac{c_r\psi_\rho}{\tr\lr{\psi_\rho\rho}}+\frac{(1-c_r)\psi_\sigma}{\tr\lr{\psi_\sigma\sigma}}\right\|_{\infty}^{-1}.
	\end{gather*}
This gives the desired result.
\end{proof}
\end{lemma}

{\Blue 
	\begin{remark}
		We can see that parameterized error-minimizing measurements in Lemma \ref{lem: para1}, \ref{lem: para2},  and \ref{lem: para3} also satisfy conditions for USD obtained in \cite{chefles1998unambiguous,feng2004unambiguous,stojnic2007unambiguous,bandyopadhyay2014unambiguous,jafarizadeh2008optimal}. Hence, we can see that the set of error minimizing measurements and their parameterized form presented in this paper as a generalization of USD measurements. As an intermediate stage of the derivation, unambiguity condition in \cite{jafarizadeh2008optimal} is mentioned as $\mathrm{supp}(\Lambda_\rho)\subseteq \mathrm{ker}(\sigma)\text{ and }\mathrm{supp}(\Lambda_\sigma) \subseteq \mathrm{ker}(\rho)$, which is also satisfied by all the measurements in Lemma \ref{lem: para1},\ref{lem: para2}, and \ref{lem: para3}.
	\end{remark}
}

\begin{remark}
	Note that all in the three sets, i.e., $\mathcal{E}^1_{s}(\rho,\sigma), \mathcal{E}^2_{s}(\rho,\sigma)$ and $\mathcal{E}^3_{s}(\rho,\sigma)$, while parameterize measurement in terms of $\psi_\rho$ and $\psi_\sigma$,  taking $\psi_\rho$ and $\psi_\sigma$ such that $\tr\lr{\psi_\rho\rho}=0$ or $\tr\lr{\psi_\sigma\sigma}=0$ is not valid, as clear from the expression of $\Lr$ and $\Ls$ in Lemma \ref{lem: para1}, Lemma \ref{lem: para2}, and Lemma \ref{lem: para3}. This ensures that $\tr\lr{\Lambda_\rho\rho}\neq0$ or $\tr\lr{\Ls\sigma}\neq0$. This is needed because $\tr\lr{\Lambda_\rho\rho}\neq0$ or $\tr\lr{\Ls\sigma}\neq0$ is among the conditions on the sets $\mathcal{E}_s^1(\rho,\sigma)$, $\mathcal{E}_s^2(\rho,\sigma)$ and $\mathcal{E}_s^3(\rho,\sigma)$ in Theorem \ref{thm: beyondSymmetricSplit}. This further ensures that denominator of $e(\Lambda)$ is non-zero and $e(\Lambda)$ remains defined.
\end{remark}

\noindent \textbf{Constructing an arbitrary error minimizing measurement: }In Theorem \ref{thm: beyondSymmetricSplit}, we saw that an arbitrary error-minimizing measurement $\Lambda$ must belong to one of $\mathcal{E}_s^1(\rho,\sigma)$, $\mathcal{E}_s^2(\rho,\sigma)$ or $\mathcal{E}_s^3(\rho,\sigma)$. Further, the sets were written in a parameterized form in the Lemma \ref{lem: para1}, Lemma \ref{lem: para2} and Lemma \ref{lem: para3}. Utilizing these parameters, a method to construct an arbitrary error-minimizing measurement is given in Table \ref{table: neqConstruction}. The first column in the table states the free variables/parameters to take, and the other three give the constraints on these. The conditions for $\mathcal{E}_s^1(\rho,\sigma)$, $\mathcal{E}_s^2(\rho,\sigma)$ and $\mathcal{E}_s^3(\rho,\sigma)$ are written in second, third and fourth columns respectively. To construct an arbitrary error-minimizing measurement from $\mathcal{E}_s^1(\rho,\sigma)$, we can start from the first row, pick an arbitrary free variable satisfying the conditions stated in the second column, and proceed towards the last row in this way. After getting $\Lr$ and $\Ls$ in the last row, the error-minimizing measurement is given by $\{\Lr,\Ls,\I-\Lr-\Ls\}$. Constructing measurements from the sets $\mathcal{E}_s^2(\rho,\sigma)$ or $\mathcal{E}_s^3(\rho,\sigma)$ utilizes similar process over corresponding columns.

\begin{table}[H]
	\centering
	\begin{tabular}{||c||c|c|c||}
		 \hline
		 \hline
		 \multirow{2}{*}{Set}&\multirow{2}{*}{$\mathcal{E}_s^1(\rho,\sigma)$}&\multirow{2}{*}{$\mathcal{E}_s^2(\rho,\sigma)$}&\multirow{2}{*}{$\mathcal{E}_s^3(\rho,\sigma)$}\\ &  &  & \\\hline\hline
		\multirow{2}{*}{$\psi_\rho$,}& \multirow{2}{*}{$\psi_\rho\in\mathcal{S}(\I-\Pi_\sigma)$,}  &\multirow{2}{*}{Not needed,} &  \multirow{2}{*}{$\psi_\rho\in\mathcal{S}(\I-\Pi_\sigma)$,}  \\&&&\\
		\multirow{2}{*}{$\psi_\sigma$}& \multirow{2}{*}{Not needed} &  \multirow{2}{*}{$\psi_\sigma\in\mathcal{S}(\I-\Pi_\rho)$}  & \multirow{2}{*}{$\psi_\sigma\in\mathcal{S}(\I-\Pi_\rho)$}   \\&&&\\\hline
		\multirow{2}{*}{$c_r$,}&  \multirow{2}{*}{Not needed,} &  \multirow{2}{*}{Not needed,} & \multirow{2}{*}{$ c_r\in[0,1]$,} \\&&&\\
		\multirow{2}{*}{$c_{\{\cdot\}}$}& \multirow{2}{*}{$\ds c_1\leq\left\|\frac{\psi_\rho}{\tr\lr{\psi_\rho\rho}}\right\|_{\infty}^{-1}$}  &  \multirow{2}{*}{$\ds c_2\leq\left\|\frac{\psi_\sigma}{\tr\lr{\psi_\sigma\sigma}}\right\|_{\infty}^{-1}$} & \multirow{2}{*}{$\ds c_3\leq\left\|\frac{c_r\psi_\rho}{\tr\lr{\psi_\rho\rho}}+\frac{(1-c_r)\psi_\sigma}{\tr\lr{\psi_\sigma\sigma}}\right\|_{\infty}^{-1}$} \\
		&&&\\&&&\\\hline
		\multirow{4}{*}{$\Lr,\Ls$} & \multirow{2}{*}{$\displaystyle
		\Lr=c_1\frac{\psi_{\rho}}{\tr\lr{\psi_{\rho}\rho}},$}& \multirow{2}{*}{ $\displaystyle
		\Ls=c_2\frac{\psi_{\sigma}}{\tr\lr{\psi_{\sigma}\sigma}},$}& \multirow{2}{*}{$\displaystyle
		\Lr=c_3c_r\frac{\psi_{\rho}}{\tr\lr{\psi_{\rho}\rho}},$}\\&&&\\
		& \multirow{2}{*}{$\ds\Ls\in\mathcal{P}(\I-\Pi_{\rho+\sigma}),\Ls\leq \I-\Lr$}& \multirow{2}{*}{$\displaystyle\Lr\in\mathcal{P}(\I-\Pi_{\rho+\sigma}),\Lr\leq \I-\Ls$}&\multirow{2}{*}{$\displaystyle
		\Ls=c_3(1-c_r)\frac{\psi_{\sigma}}{\tr\lr{\psi_{\sigma}\sigma}}$}\\&&&\\
		\hline\hline
	\end{tabular}
\caption{A method to construct an arbitrary error-minimizing measurement from each of the three sets $\mathcal{E}_s^1(\rho,\sigma)$, $\mathcal{E}_s^2(\rho,\sigma)$ and $\mathcal{E}_s^3(\rho,\sigma)$, illustrating the selection of values for various parameters}
\label{table: neqConstruction}
\end{table}
Now, we derive the acceptance for an arbitrary error minimizing measurement from the three sets in terms of the parameters given in Table \ref{table: neqConstruction}.
\begin{theorem} \label{thm: notSymmetricAcceptance}
	For an arbitrary error minimizing measurement $\Lambda\in \mathcal{E}_s(\rho,\sigma)$, expressions of acceptance for the states $\rho$ and $\sigma$ are given as follows.
\begin{itemize}
	\item If $\Lambda\in \mathcal{E}_s^1(\rho,\sigma)$: $A_\rho(\Lambda)=c_1$ and $A_\sigma(\Lambda)=0$.
	\item If $\Lambda\in \mathcal{E}_s^2(\rho,\sigma)$: $A_\rho(\Lambda)=0$ and $A_\sigma(\Lambda)=c_2$.
	\item If $\Lambda\in \mathcal{E}_s^3(\rho,\sigma)$: $A_\rho(\Lambda)=c_3c_r$ and $A_\sigma(\Lambda)=c_3(1-c_r)$.
\end{itemize}
\noindent Here, $c_1,c_2,c_3$ and $c_r$ are parameters as given in Table \ref{table: neqConstruction}.
	\begin{proof}We know that 
		\begin{align*}
			\Gamma\in\mathcal{P}(\I-\Pi_\rho)&\Rightarrow\tr\lr{\Gamma\rho}=0,\\
			\Gamma\in\mathcal{P}(\I-\Pi_\sigma)&\Rightarrow\tr\lr{\Gamma\sigma}=0,\\ 
			\text{and } \Gamma\in\mathcal{P}(\I-\Pi_{\rho+\sigma})&\Rightarrow\tr\lr{\Gamma\rho}=0\text{ and }\tr\lr{\Gamma\sigma}=0.
		\end{align*}
		Using these, for measurements in the three sets $\mathcal{E}_s^1(\rho,\sigma)$, $\mathcal{E}_s^2(\rho,\sigma)$, and $\mathcal{E}_s^3(\rho,\sigma)$, acceptance is calculated as given below one by one.
		\begin{itemize}
			\item If $\ds \Lambda\in \mathcal{E}_s^1(\rho,\sigma)$, then $\ds \Lr=\frac{c_1\psi_{\rho}}{\tr\lr{\psi_{\rho}\rho}},\psi_\rho\in\mathcal{P}(\I-\Pi_\sigma),\Ls\in\mathcal{P}(\I-\Pi_{\rho+\sigma})$ and so
			\begin{gather*}
				\tr\lr{\Lr\rho}=c_1,\tr\lr{\Lr\sigma}=0,\tr\lr{\Ls\rho}=0,\tr\lr{\Ls\sigma}=0 \quad 
				\Rightarrow A_\rho(\Lambda)=c_1,A_\sigma(\Lambda)=0.
			\end{gather*}
		
			\item If $\Lambda\in \mathcal{E}_s^2(\rho,\sigma)$, then $\ds \Ls=c_2\frac{\psi_{\sigma}}{\tr\lr{\psi_{\sigma}\sigma}},\psi_\sigma\in\mathcal{P}(\I-\Pi_\rho),\Lr\in\mathcal{P}(\I-\Pi_{\rho+\sigma})$ and so
			\begin{gather*}
			\tr\lr{\Lr\rho}=0,\tr\lr{\Lr\sigma}=0,\tr\lr{\Ls\rho}=0,\tr\lr{\Ls\sigma}=c_2\quad
			\Rightarrow A_\rho(\Lambda)=0,A_\sigma(\Lambda)=c_2.
			\end{gather*}
			\item If $\Lambda\in \mathcal{E}_s^3(\rho,\sigma)$, then $\ds \Lr=\frac{c_3c_r\psi_{\rho}}{\tr\lr{\psi_{\rho}\rho}}, \Ls=\frac{c_3(1-c_r)\psi_{\sigma}}{\tr\lr{\psi_{\sigma}\sigma}},\psi_\rho\in\mathcal{P}(\I-\Pi_\sigma),\psi_\sigma\in\mathcal{P}(\I-\Pi_\rho)$ and so
			\begin{align*}
			\tr\lr{\Lr\rho}=c_3c_r,\tr\lr{\Lr\sigma}=0, \tr\lr{\Ls\rho}& =0, \tr\lr{\Ls\sigma}=c_3(1-c_r)\\
			& \Rightarrow A_\rho(\Lambda)=c_3c_r,A_\sigma(\Lambda)=c_3(1-c_r).
			\end{align*}
		\end{itemize}\vspace{-7mm}
	\end{proof}
\end{theorem}

Now, we maximize the acceptance over the set $\mathcal{E}\left(\rho,\sigma,p\right)$. To do so, we maximize the expression obtained in Theorem \ref{thm: notSymmetricAcceptance} with respect to the parameters. The maximum acceptance expression is stated in the following theorem.

\begin{theorem}[Maximum acceptance] \label{thm: maxAccNeq} Given a pair of states as $\rho$ and $\sigma$ such that $\Pi_\rho\neq\Pi_\sigma$, then the maximum value of acceptance for states $\rho$ and $\sigma$ are summarized below in the Table \ref{table: maxAccNeq}. 
\begin{table}[H]
\centering
\begin{tabular}{||c||c|c||}
	\hline\hline
	\multirow{2}{*}{$\rho,\sigma$ such that}& \multirow{2}{*}{$A_\rho^s$} & \multirow{2}{*}{$A_\sigma^s$}  \\ &&  \\ \hline\hline
	\multirow{2}{*}{$\Pi_\rho>\Pi_\sigma$}& \multirow{2}{*}{$1-\tr\lr{\Pi_\sigma\rho}$} & \multirow{2}{*}{$0$} \\ &&\\ \hline
	\multirow{2}{*}{$\Pi_\rho<\Pi_\sigma$}& \multirow{2}{*}{$0$} & \multirow{2}{*}{$1-\tr\lr{\Pi_\rho\sigma}$}\\ && \\ \hline
	\multirow{2}{*}{$\Pi_\rho\not<\Pi_\sigma$ and $\Pi_\rho\not>\Pi_\sigma$}& \multirow{2}{*}{$1-\tr\lr{\Pi_\sigma\rho}$} & \multirow{2}{*}{$1-\tr\lr{\Pi_\rho\sigma}$}\\ && \\ \hline\hline
\end{tabular}
\caption{Table summarizing the maximum achievable acceptance for various cases considered in Section \ref{sec: rhoNeqsigma}}
\label{table: maxAccNeq}
\end{table}
\begin{proof} Recall from Corollary \ref{col: case1} that when $\Pi_\rho>\Pi_\sigma$, $\mathcal{E}_s(\rho,\sigma,p)=\mathcal{E}^1_{s}(\rho,\sigma)$. Similarly, from Corollary \ref{col: case2}, we know that when $\Pi_\rho<\Pi_\sigma$, $\mathcal{E}_s(\rho,\sigma,p)=\mathcal{E}^2_{s}(\rho,\sigma)$. We begin with the third case, i.e., $\Pi_\rho\not<\Pi_\sigma$ and $\Pi_\rho\not>\Pi_\sigma$. In this case $\mathcal{E}_s(\rho,\sigma,p)=\mathcal{E}^1_{s}(\rho,\sigma)\cup \mathcal{E}^2_{s}(\rho,\sigma)\cup \mathcal{E}^3_{s}(\rho,\sigma)$. Starting from maximizing acceptance for the state $\rho$ over the set $\mathcal{E}_s(\rho,\sigma,p)$, which can be obtained by maximizing over all the three sets followed by taking the maximum of the three, we get
\begin{align}
		A_\rho^s=\max_{\Lambda\in \mathcal{E}_s(\rho,\sigma,p)} A_\rho(\Lambda)&=\max\lr{\max_{\Lambda\in \mathcal{E}_s^1(\rho,\sigma)} A_\rho(\Lambda),\max_{\Lambda\in \mathcal{E}_s^2(\rho,\sigma)} A_\rho(\Lambda),\max_{\Lambda\in \mathcal{E}_s^3(\rho,\sigma)} A_\rho(\Lambda)}.
	\end{align}
Note that $\ds A_\rho(\Lambda)=0\ \forall\ \Lambda\in \mathcal{E}_s^2(\rho,\sigma)$. Hence,  we have
	\begin{align}		
		A_\rho^s&=\max\lr{\max_{\Lambda\in \mathcal{E}_s^1(\rho,\sigma)} A_\rho(\Lambda),\max_{\Lambda\in \mathcal{E}_s^3(\rho,\sigma)} A_\rho(\Lambda)}.
	\end{align}
Using the expression of acceptance from Theorem \ref{thm: notSymmetricAcceptance}, we get 
	\begin{align}
		A_\rho^s&=\max\lr{\max_{\Lambda\in \mathcal{E}_s^1(\rho,\sigma)} c_1,\max_{\Lambda\in \mathcal{E}_s^3(\rho,\sigma)} c_3c_r}.
	\end{align}
Applying bounds on $c_1$ and $c_3$ for measurements in $\mathcal{E}_s^1(\rho,\sigma)$ and $\mathcal{E}_s^3(\rho,\sigma)$ from Lemma \ref{lem: para1} and Lemma \ref{lem: para3} respectively (also given in Table \ref{table: neqConstruction}), we get
	\begin{align}
		A_\rho^s&=\max\lr{\max_{\psi_\rho\in\mathcal{S}(\I-\Pi_\sigma)} \left\|\frac{\psi_\rho}{\tr\lr{\psi_\rho\rho}}\right\|_{\infty}^{-1},\max_{\psi_\rho\in\mathcal{S}(\I-\Pi_\sigma),\psi_\sigma\in\mathcal{S}(\I-\Pi_\rho)} c_r\left\|\frac{c_r\psi_\rho}{\tr\lr{\psi_\rho\rho}}+\frac{(1-c_r)\psi_\sigma}{\tr\lr{\psi_\sigma\sigma}}\right\|_{\infty}^{-1}}.
	\end{align}
Taking $c_r$ inside the max-norm expression
	\begin{align}	
		A_\rho^s&=\max\lr{\max_{\psi_\rho\in\mathcal{S}(\I-\Pi_\sigma)} \left\|\frac{\psi_\rho}{\tr\lr{\psi_\rho\rho}}\right\|_{\infty}^{-1},\max_{\psi_\rho\in\mathcal{S}(\I-\Pi_\sigma),\psi_\sigma\in\mathcal{S}(\I-\Pi_\rho)} \left\|\frac{\psi_\rho}{\tr\lr{\psi_\rho\rho}}+\frac{(1-c_r)\psi_\sigma}{c_r\tr\lr{\psi_\sigma\sigma}}\right\|_{\infty}^{-1}}. \label{eq: whichOfTheTwo}
	\end{align}
Note that $\ds\frac{\psi_\rho}{\tr\lr{\psi_\rho\rho}}\leq \frac{\psi_\rho}{\tr\lr{\psi_\rho\rho}}+\frac{(1-c_r)\psi_\sigma}{c_r\tr\lr{\psi_\sigma\sigma}}\Rightarrow\left\|\frac{\psi_\rho}{\tr\lr{\psi_\rho\rho}}\right\|_{\infty}\leq \left\|\frac{\psi_\rho}{\tr\lr{\psi_\rho\rho}}+\frac{(1-c_r)\psi_\sigma}{c_r\tr\lr{\psi_\sigma\sigma}}\right\|_{\infty}$. So, we get 
\begin{align*}
	&\left\|\frac{\psi_\rho}{\tr\lr{\psi_\rho\rho}}+\frac{(1-c_r)\psi_\sigma}{c_r\tr\lr{\psi_\sigma\sigma}}\right\|_{\infty}^{-1}\leq \left\|\frac{\psi_\rho}{\tr\lr{\psi_\rho\rho}}\right\|_{\infty}^{-1}\\
	&\Rightarrow \max_{\psi_\rho\in\mathcal{S}(\I-\Pi_\sigma),\psi_\sigma\in\mathcal{S}(\I-\Pi_\rho)} \left\|\frac{\psi_\rho}{\tr\lr{\psi_\rho\rho}}+\frac{(1-c_r)\psi_\sigma}{c_r\tr\lr{\psi_\sigma\sigma}}\right\|_{\infty}^{-1}\leq \max_{\psi_\rho\in\mathcal{S}(\I-\Pi_\sigma)} \left\|\frac{\psi_\rho}{\tr\lr{\psi_\rho\rho}}\right\|_{\infty}^{-1}.
\end{align*}
Substituting above in \eqref{eq: whichOfTheTwo}, we get
	\begin{align}
		A_\rho^s&=\max_{\psi_\rho\in\mathcal{S}(\I-\Pi_\sigma)}\left\|\frac{\psi_\rho}{\tr\lr{\psi_\rho\rho}}\right\|_{\infty}^{-1}=\max_{\psi_\rho\in\mathcal{S}(\Pi_{\rho+\sigma}-\Pi_{\sigma})}\left\|\frac{\psi_\rho}{\tr\lr{\psi_\rho\rho}}\right\|_{\infty}^{-1}. \label{eq: direct}
	\end{align}
The equality can be proved similar to the proof of \eqref{eq: projEquiNorm} in proof of Lemma \ref{lem: singleMin} of Appendix \ref{sec: maxNormMinimization}. Note that among the all $\tilde\psi_\rho\in\mathcal{P}(\Pi_{\rho+\sigma}-\Pi_{\sigma})$, such that $\|\tilde\psi_\rho\|_{\infty}=1$, the one that maximizes $\tr\lr{\tilde\psi_\rho\rho}$ is $\tilde\psi_\rho=(\Pi_{\rho+\sigma}-\Pi_{\sigma})$. Hence, we obtain the maximum at $\ds\psi_\rho=\frac{\Pi_{\rho+\sigma}-\Pi_{\sigma}}{\tr\lr{\Pi_{\rho+\sigma}-\Pi_{\sigma}}}$ and is given by
\begin{align}	
	A_\rho^s&=\tr\lr{(\Pi_{\rho+\sigma}-\Pi_{\sigma}))\rho}=1-\tr\lr{\Pi_\sigma\rho}.
\end{align}
The derivation of $A_\sigma^s=1-\tr\lr{\Pi_\rho\sigma}$ is similar, and can be obtained following the same process.

For the case $\Pi_\rho<\Pi_\sigma$ or $\Pi_\rho>\Pi_\sigma$, the maximum value of acceptance can be obtained in a similar way by maximizing over the sets $\mathcal{E}_s^1(\rho,\sigma)$ or $\mathcal{E}_s^2(\rho,\sigma)$ respectively. %, which make the problem simpler and we directly get expression as in \eqref{eq: direct}, thus obtaining the stated values in Table \ref{table: maxAccNeq}.
\end{proof}

{\Blue 
	\begin{remark} \label{rem: desiredAccNeq}
		Note that for any desired value of acceptance, error-minimizing measurements can be designed by choosing the appropriate value of constant $c_1,c_2$, or $(c_3,c_r)$, provided that acceptance remains less than the maximum obtained in Theorem \ref{thm: maxAccNeq}.
	\end{remark}
}

\begin{remark} \label{rem: obtainsMaxAccNeq}
	If $\Pi_\rho>\Pi_\sigma$, a measurement that maximize acceptance is $\{\Pi_\rho>\Pi_\sigma,0,\I-\Pi_\rho+\Pi_\sigma\}$. If $\Pi_\rho<\Pi_\sigma$, a measurement that maximize acceptance is $\{0,\Pi_\sigma-\Pi_\rho,\I+\Pi_\rho-\Pi_\sigma\}$. In the case when $\Pi_\rho\not<\Pi_\sigma$, $\Pi_\rho\not>\Pi_\sigma$, and $\Pi_\rho\not=\Pi_\sigma$, a measurement that maximize acceptance is given by $\{\Pi_{\rho+\sigma}-\Pi_{\sigma},\Pi_{\rho+\sigma}-\Pi_{\rho},\I+\Pi_{\rho}+\Pi_{\sigma}-2\Pi_{\rho+\sigma}\}$.	
%	In all three cases, the stated maximum acceptance for the state $\rho$, i.e., $A_\rho^s$ can be achieved by taking the measurement $\{\Pi_{\rho+\sigma}-\Pi_{\sigma},0,\I+\Pi_{\sigma}-\Pi_{\rho+\sigma}\}$. In all three cases, the stated maximum acceptance for the state $\sigma$, i.e., $A_\sigma^s$ is achieved by taking the measurement $\{0,\Pi_{\rho+\sigma}-\Pi_{\rho},\I+\Pi_{\rho}-\Pi_{\rho+\sigma}\}$.
\end{remark}

\end{theorem}

{\Blue 
\begin{remark} \label{remark: compE2} For $n=2$ case, the condition for the existence of USD in Theorem 2 in \cite{feng2004unambiguous} is given as $\textrm{supp}(\rho+\sigma)\neq \textrm{supp}(\rho)$ and $\textrm{supp}(\rho+\sigma)\neq \textrm{supp}(\sigma)$. It simplifies to the condition $\Pi_\rho\not=\Pi_\sigma$, $\Pi_\rho\not<\Pi_\sigma$, and $\Pi_\rho\not>\Pi_\sigma$. From \eqref{eq: symBase}, we know that zero error (i.e., USD) can be obtained for any pair of states satisfying $\Pi_\rho\neq\Pi_\sigma$. So, for a pair of states satisfying $\Pi_\rho<\Pi_\sigma$, and $\Pi_\rho>\Pi_\sigma$, USD is possible, but the maximum acceptance for one of the states vanishes as obtained in Theorem \ref{thm: maxAccNeq}. So, the condition obtained in Theorem 2 in \cite{feng2004unambiguous} is sufficient but not necessary for USD. However, the condition obtained in Theorem 2 in \cite{feng2004unambiguous} is necessary and sufficient for USD with non-zero acceptance for both states. It is not mentioned in \cite{feng2004unambiguous}, but assumed in the proof of Theorem 2 in \cite{feng2004unambiguous}, and it becomes explicit from the maximum acceptance obtained in Theorem \ref{thm: maxAccNeq} in our paper.
	
\end{remark}
}

\section{Conclusions}
\label{sec:conclusion}
%\enlargethispage{-20cm} 

In this work, we have given the set of all measurements that minimize the postselected symmetric error. We make an important observation about the error-minimizing measurements that such measurements never declare one of the possible states unless the prior probability has a particular value given by $(p_\rho^*,p_\sigma^*)$ in \eqref{eq: star_p}. For a lower prior probability of either of the states, any error-minimizing measurement never declares that state. So, if $p_\rho<p_\rho^*$, any error-minimizing measurement never declares the unknown state as $\rho$. This observation raises a fundamental question about minimizing the postselected symmetric error $e(\Lambda)$ as a means to search for the best measurement. Then, we have stated a method to construct an arbitrary error-minimizing measurement in terms of variables, which can be varied to obtain other error-minimizing measurements. Furthermore, we show by an example that the value of acceptance varies for different measurements taken from the set of error-minimizing measurements, despite all being error-minimizing. The example illustrates the need to maximize acceptance over the set of error-minimizing measurements. So, we have obtained the expression of acceptance for an arbitrary error-minimizing measurement in terms of free variables used to construct the error-minimizing measurement. Then, we have derived the maximum acceptance. We obtain the maximum acceptance for both sates in closed-form expression if prior probabilities $p_\rho\neq p_\rho^*$ and $p_\sigma\neq p_\sigma^*$, and as an optimization problem otherwise. These results are generalizable in a simplified form for discriminating classical probability distributions by taking appropriate density matrices and measurements.

The maximum acceptance obtained in this work over the set of all error-minimizing measurements tells that any lower value of acceptance can be obtained while ensuring that the postselected symmetric error remains the same, i.e., achieves its minimum value. However, if a higher value of acceptance is desired, the postselected symmetric error will increase. So, this work opens up a problem of studying the trade-off between acceptance postselected symmetric error. {\Blue Such a study has been done in the past with different metrics for error and acceptance \cite{fiuravsek2003optimal}, and requires an independent study in the context of the problem studied in our work.} Such a study is likely to be useful in adaptive quantum hypothesis testing. This work also opens up other new directions for research, such as studying acceptance in the asymptotic case. It may lead to a new class of quantum divergence parameterized by acceptance. {\Blue In the past, one-shot USD has been generalized to USD of unknown states, and separation increasing transformations \cite{zhou2012unambiguous,chefles1998quantum,bergou2005universal}. PSD, studied in our work, can also be generalized in a similar fashion. \cite{hashimoto2010unitary} generalizes  maximizing the success probability given an error margin for discrimination of states \cite{hayashi2008state} to unitary operations. Work studied here can similarly be generalized to unitary operations and channels.} Postselected communication over quantum channels was proposed in\cite{ji2023postselected} and capacity is derived. A metric similar to acceptance in our work needs to be considered while obtaining capacity for postselected communication. The method of finding constraints on an arbitrary measurement for it to be an error-minimizing measurement, studied in our work, can be generalized to problems on probabilistic protocols, as in \cite{regula2022tight,regula2022probabilistic,regula2023overcoming} to obtain the set of all operations that maximize the success probability. Overall, this work opens up a new arena to explore other matrices for the performance of postselected symmetric hypothesis testing problems, going beyond error probability and studying acceptance in related problems.

%%%%%%
%% To balance the columns at the last page of the paper use this
%% command somewhere at the top of the first column of the last page:
%%
% \enlargethispage{-10cm} 
%%
%% where the exact amount of page reduction has to be adapted to the
%% actual situation.
%%
%% If the balancing should occur in the middle of the references, use
%% the following trigger:
%%
%\IEEEtriggeratref{3}
%%
%% which triggers a \newpage (i.e., new column) just before the given
%% reference number. Note that you need to adapt this if you modify
%% the paper. The "triggered" command can be changed if desired:
%%
% \IEEEtriggercmd{\enlargethispage{-20cm}}
%%
%%%%%%

%%%%%%

\input{apppendices}

%\nocite{*}
\bibliographystyle{IEEEtran}
%\IEEEtriggeratref{5}
\bibliography{refs}

%\newpage
%\input{../Appendix/other}

\end{document}

%% file: packages.tex
\usepackage{color}   %May be necessary if you want to color links
\usepackage[linktocpage=true,hidelinks]{hyperref}
\hypersetup{
	colorlinks=false,%true, %set true if you want colored links
	citebordercolor	={0 0 1},
}

\usepackage{cite}
\DeclareMathAlphabet\mathbfcal{OMS}{cmsy}{b}{n}

\usepackage{graphicx,epstopdf}

\usepackage{bbm}
\usepackage{amssymb}
\usepackage{amssymb,amsthm}
\usepackage{color}
\usepackage{amsfonts}
\usepackage{verbatim}
\usepackage{epstopdf}
\usepackage{cite}
\usepackage{tabulary,float}
\usepackage{mathtools}
\usepackage{enumitem}
\usepackage{dsfont}
\usepackage{graphicx}
\usepackage{subcaption}
\graphicspath{{Pics/}}
\usepackage{multirow}

%% file: defis.tex
\newcommand{\isdefinedas}{\stackrel{\Delta}{=}}

\def\home{\hbox{\kern3pt \vbox to13pt{}% 
   \pdfliteral{q 0 0 m 0 5 l 5 10 l 10 5 l 10 0 l 7 0 l 7 5 l 3 5 l 3 0 l f
               1 j 1 J -2 5 m 5 12 l 12 5 l S Q }%
   \kern 13pt}}

\usepackage{chngcntr}
\usepackage{apptools}

\newtheorem{theorem}{Theorem}%[section]
\newtheorem*{theorem*}{Theorem}
\newtheorem{lemma}{Lemma}%[section]

\newtheorem{corollary}{Corollary}
\newtheorem{remark}{Remark}

\usepackage[font=small]{caption}
\newcommand{\isdefined}{\stackrel{\Delta}{=}}

\newtheorem{lem}{Lemma}[section]
\newtheorem{thm}{Theorem}[section]

\newtheorem{defn}{Definition}%[section]

\newtheorem{exmp}{Example}%[section]

\newcommand{\newexample}[1]{%
\theoremstyle{exmp} % upright type for examples
\newtheorem{example#1}{Example}%
\expandafter\renewcommand\csname theexample#1\endcsname{#1(\alph{example#1})}%
}

\newexample{1}
\newexample{2}

\newcommand{\hilbert}{\mathcal{H}}

\newcommand{\one}{\mathds{1}}

\newcommand{\Red}{\color{red}}\newcommand{\Blue}{\color{black}}

\newcommand{\lr}[1]{\left(#1\right)}

\newcommand{\Ls}{\Lambda_{\sigma}}
\newcommand{\Lr}{\Lambda_{\rho}}
\newcommand{\tr}{\mathrm{Tr}}
\newcommand{\ds}{\displaystyle }
\newcommand{\proj}{\mathrm{P}}
\newcommand{\cp}{\mathcal{P}}
\newcommand{\cs}{\mathcal{S}}
\newcommand{\Tau}{\mathrm{T}}
\newcommand{\ch}{\mathcal{H}}
\newcommand{\I}{\mathrm{I}}

\AtAppendix{\counterwithin{theorem}{section}}
\AtAppendix{\counterwithin{lemma}{section}}

%% file: contriTable.tex
\begin{table}[ ]
	\centering
	\Blue
	\begin{tabular}{||>{\centering\arraybackslash}p{0.025\linewidth}||>{\centering\arraybackslash}p{0.065\linewidth}|>{\centering\arraybackslash}p{0.068\linewidth}|>{\raggedright\arraybackslash}p{0.7\linewidth}||} \hline \hline
		\textbf{Ref}& \textbf{Limited to USD?}&\textbf{Type of states}&\textbf{System model with key relevant contributions}
		\\ \hline \hline
		\cite{chefles1998unambiguous} & YES& $n$ Pure states&It is shown that USD is possible only if the states are linearly independent. A general form of USD measurements is obtained. A discussion on maximizing the success probability (expected acceptance) is included and illustrated by an example.
		\\\hline  
		\cite{feng2004unambiguous}& YES& $n$ Mixed states&A lower bound on the probability of rejection ($=1-$expected acceptance) is obtained in terms of the fidelity of each possible pair of states.
		\\ \hline
		\cite{stojnic2007unambiguous}& YES& $2$ Mixed states&Rank-2 states in Hilbert's space of dimension $4$ are considered for discrimination. A measurement that maximizes the success probability (expected acceptance) is numerically obtained. 
		\\ \hline 
		\cite{bandyopadhyay2014unambiguous}& YES& $n$  Pure states&Prior probability of states is known. An upper bound is obtained on the maximum success probability (expected acceptance) and shown to be tight and achievable in some cases.
		\\\hline
		\cite{gupta2024unambiguous}& YES& $n$  Pure states&Each pair of states has an equal and real inner product. For states with equal prior probability, for USD measurements, the success probability (expected acceptance) is maximized. 
		\\ \hline 
		\cite{zhang2023unambiguous} & UCD &  $n$  Pure states&The paper introduces unambiguous correct detection (UCD), a different generalization of USD. A new metric (=$1-$postselected symmetric error) is defined and it is shown that UCD improves the metric with help of examples.
		\\ \hline 
		\cite{zhou2012unambiguous} & YES& $2$  Pure states&Unknown quantum states with multiple available copies are taken for discrimination, and a method to design a programmable circuit for USD is presented.
		\\ \hline 
		\cite{chefles1998quantum}& More General& $2$  Pure states&In this work, given a quantum object in one of the two possible quantum states, a new pair of quantum states  with more separation (in terms of the inner product) is obtained, with partial success. The paper states the maximum probability of success (expected acceptance) for such operations and also provides a method to design such an operation.
		\\ \hline 
		\cite{jafarizadeh2008optimal} &YES & $n$ Pure states & For pure states with known prior probabilities $\{(|\psi_i\rangle,\eta_i): i\in[1,n]\}$, USD measurements are obtained as $\Pi_i = p_i|\psi_i\rangle\langle\psi_i|$, where reciprocal states $|\psi^{\perp}_k\rangle: \langle\psi^\perp_k|\psi_i\rangle = \delta_{ik}$. The probability of inconclusive outcomes $\left(=1-\sum_k \eta_k p_k \right)$ is minimized. An exact solution is obtained for $n=2,3$, and a numerical method is given for more states.
		\\ \hline 
		\cite{kleinmann2010unambiguous} &YES &2 Mixed states & Measurement for which error vanishes (i.e., USD measurements), and success probability (i.e., expected acceptance) is maximum, is obtained numerically by solving a set of equations, taken from a prior work.
		\\ \hline 
		\cite{hayashi2008state} &More general & 2 Pure states $|\psi_1\rangle,|\psi_2\rangle$ & For a given upper bound (termed error margin) on conditional error probabilities conditioned that, on measurement the state is declared as $|\psi_1\rangle$ and $|\psi_2\rangle$, respectively, the probability of correct detection (i.e., the probability of declaring the state correctly) is maximized.
		\\ \hline 
		\cite{hashimoto2010unitary} &Different & 2 Unitary operation & This work is a generalization of \cite{hayashi2008state} for unitary operations with uniform prior probability. Closed-form expressions of maximum success probability for a given error margin are obtained for an irreducible group of unitary operations.
		\\ \hline 
		\cite{fiuravsek2003optimal} &More general & $n$ Mixed states &An extra inconclusive outcome is added while discriminating a set of $n-$mixed states. The goal is to maximize relative success probability ($=1-$postselected symmetric error) for the measurements for which the probability of inconclusive outcomes ($=1-$expected acceptance) is fixed.
		\\ \hline 
		\cite{regula2022postselected} & More general & 2 Mixed states & Minimum postselected error is obtained for states with known prior probability. An example measurement is given for which the minimum error is achieved.
		\\ \hline \hline
		Our work & More general& 2 Mixed stated & Our work studies PSD of two mixed states with equal support, resulting in a non zero minimum postselected error as given in \cite{regula2022postselected}. The set of all measurements is obtained for which the minimum postselected error is achieved, along with a parameterization of such measurements for complete characterization. If the support of the two states is not equal, it reduces to the USD case. Acceptance is obtained for each error-minimizing measurement along with the derivation of maximum achievable acceptance over the set of all error-minimizing measurements.  If the support of the two states is not equal, it reduces to the USD case which is also studied as special case of PSD.\\ \hline \hline\end{tabular}
	\caption{\Blue The system model and key contributions of relevant references \cite{chefles1998unambiguous,feng2004unambiguous,stojnic2007unambiguous,bandyopadhyay2014unambiguous,gupta2024unambiguous,zhang2023unambiguous,zhou2012unambiguous,chefles1998quantum,jafarizadeh2008optimal,kleinmann2010unambiguous,hayashi2008state,hashimoto2010unitary,fiuravsek2003optimal,regula2022postselected} in context of our work.}
	\label{table: sumContri}
\end{table}

%% file: flow.tex
\begin{figure}[b]	
	\begin{center}
		\resizebox{0.8\linewidth}{!}
{	\begin{tikzpicture}[
		node distance=0.4cm and 0.5cm,
		every node/.style={rectangle, draw, minimum width=6cm, minimum height=1.2cm, align=center,, line width=1},
		smallbox/.style={rectangle, draw, minimum width=3, minimum height=1, align=center,inner sep=5pt,line width=1.5},
		bigbox1/.style={rectangle, draw, minimum width=7cm, minimum height=10.5cm, inner sep=0pt},
		bigbox2/.style={rectangle, draw, minimum width=8cm, minimum height=9.6cm, inner sep=0pt, line width=1.5},
		bigbox3/.style={rectangle, draw, minimum width=9cm, minimum height=9.6cm, inner sep=0pt, line width=1.5}
		]
		
		% First big box
		\node[bigbox3] (frame1) at (0,0) {};
		
		% Flowchart nodes in first box
		\node[ draw,text width=8cm] (a1) at ([yshift=2.4cm]frame1.center) {\textbf{Theorem \ref{lem: conditionSym}:}  presents the condition for a measurement to be an error minimizing measurement
		};
		\node[below=of a1,draw,text width=8cm] (a2) {\textbf{Theorem \ref{lem: setSym}:} gives the set of all error minimizing measurements $\mathcal{E}_s$
		};
		
		\node[above=2mm  of a1,draw=none,text width=8cm] (a0) {\textbf{When the supports of $\rho$ and $\sigma$ are equal i.e. $\Pi_\rho=\Pi_\sigma$}};
		
		\node[below= of a2,draw,text width=8cm] (a3) {\textbf{Theorem \ref{lem: genAccSym}:} evaluates the acceptance for each element in the set $\mathcal{E}_s$
		};
		\node[below=of a3,draw,text width=8cm] (a4) {\textbf{Theorem \ref{thm: symmetric}:} gives the maximizes acceptance achievable by measurement in the set of all error minimizing measurements $\mathcal{E}_s$
		};
		% Arrows in first box
		\draw[-{Latex}, line width=1] (a1) -- (a2);
		\draw[-{Latex}, line width=1] (a2) -- (a3);
		\draw[-{Latex}, line width=1] (a3) -- (a4);
		
		% Small rectangle at the bottom of first box
		\node[smallbox, below=0.31cm of a4] (a5) {\textbf{Section \ref{sec: rhoEqsigma}}};
		
		% Second big box
		\node[bigbox2, right=1cm of frame1] (frame2) {};
		
		% Flowchart nodes in second box
		\node[draw,text width=7cm] (b1) at ([yshift=2.6cm]frame2.center) {\textbf{Theorem \ref{thm: beyondSymmetricSplit} :} gives the set of all error minimizing measurements $\mathcal{E}_s$};
		
		\node[above=2mm of b1,draw=none,text width=8cm] (b0) {\textbf{When the supports of $\rho$ and $\sigma$ are not equal i.e. $\Pi_\rho\neq\Pi_\sigma$}};
		
		\node[below=of b1,draw,text width=7cm] (b2) {\textbf{Lemma \ref{lem: para1},\ref{lem: para2},\ref{lem: para3}:} parameterize an arbitrary error minimizing measurement};
		
		\node[below=of b2,draw,text width=7cm] (b3) {\textbf{Theorem \ref{thm: notSymmetricAcceptance}:} evaluates the acceptance for each element in the set $\mathcal{E}_s$};
		
		\node[below=of b3,draw,text width=7cm] (b4) {\textbf{Theorem \ref{thm: maxAccNeq}:} gives the maximizes acceptance achievable by measurement in the set of all error minimizing measurements $\mathcal{E}_s$};
		
		% Arrows in second box
		\draw[-{Latex}, line width=1] (b1) -- (b2);
		\draw[-{Latex}, line width=1] (b2) -- (b3);
		\draw[-{Latex}, line width=1] (b3) -- (b4);
		
		% Small rectangle at the bottom of second box
		\node[smallbox, below=0.25cm of b4] (b5) {\textbf{Section \ref{sec: rhoNeqsigma}}};
		
\end{tikzpicture}}

\end{center}
\caption{\Blue Figure showing connection among theorems in Section \ref{sec: rhoEqsigma} and Section \ref{sec: rhoNeqsigma}}
\label{fig: thmFlow}
\end{figure}

%% file: apppendices.tex
\appendices
%\input{oldAppendixA}
%\newpage
%\input{newAppendixA}
%\newpage
\section{Proof of Theorem \ref{lem: conditionSym}}\label{proof: conditionSym}
In this appendix, we take $\nu\in\mathcal{P}(\hilbert)$ as a fixed operator and $\zeta\in\mathcal{S}(\hilbert)$ as a variable operator. We find the maximum and the minimum value of $\tr\lr{\zeta\nu}$ by varying $\zeta\in\mathcal{S}(\hilbert)$ in Lemma \ref{lem: max_lemma} and Lemma \ref{lem: min_lemma} respectively. We also obtain the necessary and sufficient condition on $\zeta$ to obtain the maximum and the minimum value of $\tr\lr{\zeta\nu}$ in Lemma \ref{lem: max_lemma} and Lemma \ref{lem: min_lemma} respectively. We utilize these bounds, and condition on $\zeta$ for equality of these bounds in the proof of Theorem \ref{lem: conditionSym}. 

Eigenvalue decomposition of $\nu$ and the projection operator of the subspace spanned by eigenvectors is given as $\nu=\sum_{i} k_{\nu,i}|e_i\rangle\langle e_i|$  and $\Pi_\nu=\sum_{i} |e_i\rangle\langle e_i|$, where $k_{\nu,i}$ represents the eigenvalue and $|e_i\rangle$ is corresponding eigenvector of $\nu$. The projection operators the corresponding to the maximum and minimum eigenvalues are given as 
	\begin{align}
		\Pi^{\max}_\nu=\sum_{i:k_{\nu,i}=\|\nu\|_{\infty}} |e_i\rangle\langle e_i| \text{ and }	\Pi^{\min}_\nu=\sum_{i:k_{\nu,i}=\|\nu\|_{\infty,0}} |e_i\rangle\langle e_i| \text{ respectively.} \label{eq: decomposeNu}
	\end{align}
	
	\begin{lem} \label{lem: max_lemma}
		For a given $\nu\in\mathcal{P}(\hilbert)$, we get $%\|\nu\|_{\infty,0}\leq
		\tr\lr{\zeta\nu}\leq\|\nu\|_{\infty} \ \forall \ \zeta\in\mathcal{D}(\hilbert)$  
		and $$\tr\lr{\zeta\nu}=\|\nu\|_{\infty}\Leftrightarrow\zeta\in \mathcal{S}(\Pi_{\nu}^{\max}).$$
		\begin{proof}
			The proof is done in 2 steps. First, we have shown that $\tr\lr{\zeta\nu}\leq\|\nu\|_{\infty}$ and $\tr\lr{\zeta\nu}=\|\nu\|_{\infty} \Rightarrow \zeta\in \mathcal{S}(\Pi_{\nu}^{\max})$, followed by showing $\zeta\in \mathcal{S}(\Pi_{\nu}^{\max})\Rightarrow\tr\lr{\zeta\nu}=\|\nu\|_{\infty}$. \begin{enumerate}
				
				\item %$\tr\lr{\zeta\nu}\leq \|\nu\|_{\infty} $ and $\tr\lr{\zeta\nu}=\|\nu\|_{\infty} \Rightarrow \zeta\in\mathcal{S}(\Pi_\nu^{\max})$.
				Let us take the eigenvalue decomposition of $\zeta=\sum_ik_{\zeta,i}|\phi_i\rangle\langle\phi_i|$ with $\sum_i k_{\zeta,i}=1$. Hence, $\tr\lr{\zeta \nu}=\sum_ik_{\zeta,i} \langle\phi_i|\nu|\phi_i\rangle$. Note that $\langle\phi_i|\nu|\phi_i\rangle \leq \|\nu\|_{\infty}$ from the definition of max norm. Further, $\langle\phi_i|\nu|\phi_i\rangle = \|\nu\|_{\infty}\Leftrightarrow\Pi_\nu^{\max}|\phi_i\rangle=|\phi_i\rangle$ as for maximization of $\langle\phi_i|\nu|\phi_i\rangle$, $|\phi_i\rangle$ must lie in the eigenspace corresponding to the maximum eigenvalue of $\nu$.
				Hence, 
				\begin{align}
					\tr\lr{\zeta \nu}=\sum_ik_{\zeta,i} \langle\phi_i|\nu|\phi_i\rangle \leq \|\nu\|_{\infty} \sum_ik_{\zeta,i} = \|\nu\|_{\infty},
				\end{align}
				with $\tr\lr{\zeta \nu} = \|\nu\|_{\infty} \Leftrightarrow \langle\phi_i|\nu|\phi_i\rangle = \|\nu\|_{\infty}\ \forall\ i  \Leftrightarrow\Pi_\nu^{\max}|\phi_i\rangle=|\phi_i\rangle\ \forall\ i$. Hence, 			
				\begin{align}
					\Pi_{\nu}^{\max}\zeta\Pi_{\nu}^{\max} = \sum_i k_{\zeta,i} \Pi_{\nu}^{\max}|\phi_i\rangle\langle\phi_i|\Pi_{\nu}^{\max} = \sum_i k_{\zeta,i} |\phi_i\rangle\langle\phi_i| = \zeta \Rightarrow \zeta \in \mathcal{S}(\Pi_{\nu}^{\max}).
				\end{align}
				So, $\tr\lr{\zeta\nu}\leq\|\nu\|_{\infty}$ and $\tr\lr{\zeta\nu}=\|\nu\|_{\infty} \Rightarrow \zeta\in \mathcal{S}(\Pi_{\nu}^{\max})$.
				\item Note that $\nu=\sum_i k_{\nu,i} |e_i\rangle\langle e_i|$, hence $\Pi_\nu=\sum_{i}|e_i\rangle\langle e_i|$ and $\ds\Pi_\nu^{\max}=\sum_{i: k_{\nu,i}=\|\nu\|_{\infty}}|e_i\rangle\langle e_i|$. Now, from the definition of set $\mathcal{S}(\Pi_\nu^{\max})$, it follows that $\Pi_\nu^{\max}\zeta\Pi_\nu^{\max}=\zeta$ and $\tr\lr{\zeta}=1$. So we get
				\begin{align}
				{\zeta\nu} &= {\Pi_\nu^{\max}\zeta\Pi_\nu^{\max}\nu}= {\Pi_\nu^{\max}\zeta\sum_{\{i:k_{\nu,i}=\|\nu\|_{\infty}\}}\sum_j |e_i\rangle\langle e_i| k_{\nu,j} |e_j\rangle\langle e_j|}\\
					&= {\Pi_\nu^{\max}\zeta\sum_{\{i:k_{\nu,i}=\|\nu\|_{\infty}\}}\sum_j \one(i=j) k_{\nu,j} |e_j\rangle\langle e_j|}\\
					&= {\Pi_\nu^{\max}\zeta\sum_{i: k_{\nu,i}=\|\nu\|_{\infty}} \|\nu\|_{\infty} |e_j\rangle\langle e_j|} = \|\nu\|_{\infty}{\Pi_\nu^{\max}\zeta\Pi_\nu^{\max}}
					=\|\nu\|_{\infty}{\zeta}.
				\end{align}
				So we get $\zeta\in\mathcal{S}(\Pi_\nu^{\max})\Rightarrow{\zeta\nu}=\|\nu\|_{\infty}{\zeta} \Rightarrow\tr\lr{\zeta\nu} = \|\nu\|_{\infty} \tr\lr{\zeta}= \|\nu\|_{\infty}$.
			\end{enumerate}		
			Combining the two, we get the desired result.
		\end{proof}
	\end{lem}
	
	\begin{lem} \label{lem: min_lemma}
		For a given $\nu\in\mathcal{P}(\hilbert)$, we get $%\|\nu\|_{\infty,0}\leq
		\tr\lr{\zeta\nu}\geq\|\nu\|_{\infty,0} \ \forall \ \zeta\in\mathcal{S}(\Pi_\nu)$  
		and $$\tr\lr{\zeta\nu}=\|\nu\|_{\infty,0}\Leftrightarrow\zeta\in \mathcal{S}(\Pi_{\nu}^{\min}).$$
		\begin{proof} The proof follows similar to the proof of Lemma \ref{lem: max_lemma}.
		\end{proof}
	\end{lem}
	
\noindent \textbf{Proof of Theorem \ref{lem: conditionSym}}: From the definition of $e(\Lambda)$ in \eqref{eq: defSymError},	$e(\Lambda)$ can be written as $e(\Lambda)=\lr{1+R_{\Lambda,p}(\rho,\sigma)}^{-1}$, where 
		\begin{align}
			R_{\Lambda,p}(\rho,\sigma) \isdefined \frac{p_\rho\tr(\Lambda_{\rho} \rho)+p_\sigma\tr(\Lambda_{\sigma} \sigma)}{p_\sigma\tr(\Lambda_{\rho} \sigma)+p_\rho{\tr(\Lambda_{\sigma} \rho)}}.
		\end{align}
		On taking $\Lambda_\rho^s=\sigma^{1/2}\Lambda_\rho\sigma^{1/2}$ and $\Lambda_\sigma^s=\sigma^{1/2}\Lambda_\sigma\sigma^{1/2}$, we get
		\begin{align*}
			\tr\lr{\sigma^{-1/2}\Lambda_\rho^{s}\sigma^{-1/2}\rho} &= \tr\lr{\Pi_{\sigma}\Lambda_\rho\Pi_{\sigma}\rho}\stackrel{(a)}{=} \tr\lr{\Pi_{\rho}\Lambda_\rho\Pi_{\rho}\rho}
			= \tr\lr{\Lambda_\rho\Pi_{\rho}\rho\Pi_{\rho}}= \tr\lr{\Lambda_\rho\rho},
		\end{align*}
		where $(a)$ is due to $\Pi_\rho=\Pi_\sigma$. Also,
		$\tr\lr{\Lambda_\rho^s}=\tr(\sigma^{1/2}\Lambda_\rho\sigma^{1/2}) = \tr\lr{\Lambda_\rho\sigma}.$ 
		Similarly, $\tr\lr{\sigma^{-1/2}\Lambda_\sigma^{s}\sigma^{-1/2}\rho}=\tr\lr{\Lambda_\sigma\rho}$ and $\tr\lr{\Lambda_\sigma^s}=\tr\lr{\Lambda_\sigma\sigma}$. On substituting these in $R_{\Lambda,p}(\rho,\sigma)$, we get
		\begin{align*}
			R_{\Lambda,p}&(\rho,\sigma)=	\frac{p_\rho\tr\lr{\Lambda_{\rho}^s\sigma^{-1/2}\rho\sigma^{-1/2}}+p_\sigma\tr\lr{\Lambda_{\sigma}^s}}{p_\sigma\tr\lr{\Lambda_{\rho}^s}+p_\rho\tr\lr{\Lambda_{\sigma}^s\sigma^{-1/2}\rho\sigma^{-1/2}}}.
		\end{align*}
		Choosing $\Lambda_{\rho}^s=c_{\rho}\Lambda_{\rho}^{s,c}$ and $\Lambda_{\sigma}^s=c_{\sigma}\Lambda_{\sigma}^{s,c}$ such that trace of $\Lambda_{\rho}^{s,c}$ and $\Lambda_{\sigma}^{s,c}$ is $1$. We get
		\begin{align}
			R_{\Lambda,p}(\rho,\sigma)&=\frac{p_\sigma c_{\sigma}+p_\rho c_{\rho}\tr\lr{\Lambda_{\rho}^{s,c}\sigma^{-1/2}\rho\sigma^{-1/2}}}{p_\sigma c_{\rho}+p_\rho c_{\sigma}\tr\lr{\Lambda_{\sigma}^{s,c}\sigma^{-1/2}\rho\sigma^{-1/2}}} =\frac{c_\sigma+c_\rho\frac{p_\rho}{p_\sigma}\tr\lr{\Lambda_{\rho}^{s,c}\sigma^{-1/2}\rho\sigma^{-1/2}}}{c_\rho+c_\sigma\frac{p_\rho}{p_\sigma}\tr\lr{\Lambda_{\sigma}^{s,c}\sigma^{-1/2}\rho\sigma^{-1/2}}}.\label{eq: rLambdaSimplified}
		\end{align}
		Using Lemma \ref{lem: max_lemma} and Lemma \ref{lem: min_lemma}, we get
		\begin{gather*}
			\tr\lr{\Lambda_{\rho}^{s,c}\sigma^{-1/2}\rho\sigma^{-1/2}}\leq\|\sigma^{-1/2}\rho\sigma^{-1/2}\|_{\infty} \text{ with equality iff }\Lambda_{\rho}^{s,c}\in \mathcal{S}(\Tau^{\max}) \text{ and}\\
			\tr\lr{\Lambda_{\sigma}^{s,c}\sigma^{-1/2}\rho\sigma^{-1/2}}\geq\|\sigma^{-1/2}\rho\sigma^{-1/2}\|_{\infty,0} \text{ with equality iff }\Lambda_{\sigma}^{s,c}\in \mathcal{S}(\Tau^{\min}),
		\end{gather*}
		where $\Tau^{\max} = \Pi_{\sigma^{-1/2}\rho\sigma^{-1/2}}^{\max}$ and $\Tau^{\min} = \Pi_{\sigma^{-1/2}\rho\sigma^{-1/2}}^{\min}$. Substituting these bounds in the expression of $R_{\Lambda,p}(\rho,\sigma)$ in \eqref{eq: rLambdaSimplified}, we get 
		\begin{align}
			R_{\Lambda,p}&(\rho,\sigma)\leq \frac{c_\sigma+c_\rho R_{\max}(p_\rho\rho,p_\sigma\sigma)}{c_\rho+c_\sigma R_{\min}(p_\rho\rho,p_\sigma\sigma)} %\\&
			= R_{\max}(p_\sigma\sigma,p_\rho\rho) \frac{c_\sigma+c_\rho R_{\max}(p_\rho\rho,p_\sigma\sigma)}{c_\sigma+c_\rho R_{\max}(p_\sigma\sigma,p_\rho\rho)}, \label{eq: firstInq}
		\end{align}
		with the equality being obtained iff $\Lambda_{\rho}^{s,c}\in \mathcal{S}(\Tau^{\max})$ and $\Lambda_{\sigma}^{s,c}\in \mathcal{S}(\Tau^{\min})$. Now, ignoring the constant multiplier $R_{\max}(p_\sigma\sigma,p_\rho\rho)$, RHS of \eqref{eq: firstInq} is maximized as
		\begin{align}
			\max_{c_\rho,c_\sigma}\frac{c_\sigma+c_\rho R_{\max}(p_\rho\rho,p_\sigma\sigma)}{c_\sigma+c_\rho R_{\max}(p_\sigma\sigma,p_\rho\rho)}=\begin{cases}
				\frac{R_{\max}(p_\rho\rho,p_\sigma\sigma)}{ R_{\max}(p_\sigma\sigma,p_\rho\rho)},& \text{if }{R_{\max}(p_\rho\rho,p_\sigma\sigma)}>{ R_{\max}(p_\sigma\sigma,p_\rho\rho)},	\\ 1,&\text{otherwise}.
			\end{cases} \label{eq: max_cRhoSigma}
		\end{align}
		The maximum in \eqref{eq: max_cRhoSigma} being obtained at $\ds\underset{c_\rho,c_\sigma}{\arg\max}\frac{c_\sigma+c_\rho R_{\max}(p_\rho\rho,p_\sigma\sigma)}{c_\sigma+c_\rho R_{\max}(p_\sigma\sigma,p_\rho\rho)}$
		\begin{align}
			=\begin{cases}
				c_\rho>0,c_\sigma=0,& \text{if }{R_{\max}(p_\rho\rho,p_\sigma\sigma)}>{ R_{\max}(p_\sigma\sigma,p_\rho\rho)},\\
				c_\rho=0,c_\sigma>0,& \text{if }{R_{\max}(p_\rho\rho,p_\sigma\sigma)}<{ R_{\max}(p_\sigma\sigma,p_\rho\rho)},\\
				c_\rho\ge0,c_\sigma\ge0,(c_\rho,c_\sigma)\neq(0,0),& \text{if }{R_{\max}(p_\rho\rho,p_\sigma\sigma)}={ R_{\max}(p_\sigma\sigma,p_\rho\rho)}.
			\end{cases} \label{eq: cnoditionOnC}
		\end{align}
		Combining \eqref{eq: firstInq} and \eqref{eq: max_cRhoSigma}, we obtain
		\begin{align}
			R_{\Lambda,p}(\rho,\sigma)&\leq\begin{cases}
			{R_{\max}(p_\rho\rho,p_\sigma\sigma)},& \text{if }{R_{\max}(p_\rho\rho,p_\sigma\sigma)}>{ R_{\max}(p_\sigma\sigma,p_\rho\rho)},	\\ 
			{ R_{\max}(p_\sigma\sigma,p_\rho\rho)},&\text{otherwise}
		\end{cases}  \label{eq: RLd} \\
		&=\max({R_{\max}(p_\rho\rho,p_\sigma\sigma)},{ R_{\max}(p_\sigma\sigma,p_\rho\rho)}). 
		\end{align}
		So, $\min e(\Lambda)=(1+\max({R_{\max}(p_\rho\rho,p_\sigma\sigma)},{ R_{\max}(p_\sigma\sigma,p_\rho\rho)}))^{-1}$, which is one of the desired results.
		 
		Further, the minimum of $e(\Lambda)$ being achieved at $c_\rho$ and $c_\sigma$ satisfying \eqref{eq: cnoditionOnC} and $\Lambda_{\rho}^{s,c}\in \mathcal{S}(\Tau^{\max})$ and $\Lambda_{\sigma}^{s,c}\in \mathcal{S}(\Tau^{\min})$. Combining $\Lambda_{\rho}^{s,c}\in \mathcal{S}(\Tau^{\max})$ with $c_\rho=0$ and $c_\rho>0$ gives $\Lambda_{\rho}^{s}=\sigma^{1/2}\Ls\sigma^{1/2}=0$ and $\Lambda_{\rho}^{s}=\sigma^{1/2}\Ls\sigma^{1/2}\in \mathcal{P}(\Tau^{\max})$ respectively. Similarly, we get condition on $\sigma^{1/2}\Ls\sigma^{1/2}$. Thus, the final condition for equality in \eqref{eq: RLd}, and so the condition on $\Lambda$ for the minimum of $e(\Lambda)$ to be achieved, is obtained as
		$$	\begin{cases}
			\sigma^{1/2}\Lr\sigma^{1/2}\in\mathcal{P}(\Tau^{\max}),\quad\sigma^{1/2}\Ls\sigma^{1/2}=0,& \text{if }{R_{\max}(p_\rho\rho,p_\sigma\sigma)}>{ R_{\max}(p_\sigma\sigma,p_\rho\rho)},	\\ 
			\sigma^{1/2}\Lr\sigma^{1/2}=0,\quad\sigma^{1/2}\Ls\sigma^{1/2}\in\mathcal{P}(\Tau^{\min}),& \text{if }{R_{\max}(p_\rho\rho,p_\sigma\sigma)}<{ R_{\max}(p_\sigma\sigma,p_\rho\rho)},	\\ 
			\sigma^{1/2}\Lr\sigma^{1/2}\in\mathcal{P}(\Tau^{\max}),\quad\sigma^{1/2}\Ls\sigma^{1/2}\in\mathcal{P}(\Tau^{\min}),&\text{otherwise},
		\end{cases}
		$$
		which completes the proof.

%\newpage

\section{Generalized form of operator $\zeta$ given a constraint on subspace where it lies} \label{sec: genProj}

\begin{lem} \label{lem: null} For a variable operator $\zeta\in\mathcal{P}(\hilbert)$ and a fixed operator $\nu\in\mathcal{P}(\hilbert)$, the following are equivalent: \\
	$(\mathrm{1}) \ \tr\lr{\zeta\nu}=0 \qquad (\mathrm{2}) \ \Pi_\nu\zeta=\zeta\Pi_\nu=0 \qquad (\mathrm{3}) \  \zeta\in\mathcal{P}(\I-\Pi_\nu)\qquad (\mathrm{4}) \ \Pi_\nu\zeta\Pi_\nu=0.$
	\begin{proof} We will prove $(\mathrm{1})\Rightarrow(\mathrm{2})\Rightarrow(\mathrm{3})\Rightarrow(\mathrm{4})\Rightarrow(\mathrm{1})$.
		\\\noindent\textbf{(1)$\to$(2): }Let $\zeta=\sum_i k_i |\psi_i\rangle\langle\psi_i|$ be eigenvalue decomposition of $\zeta$ with $k_i>0$, then
		\begin{gather*}
			\tr\lr{\zeta\nu}=0  \Rightarrow  \sum _i k_i \langle\psi_i|\nu|\psi_i\rangle=0.
			\intertext{Note that $k_i>0$ and $\langle\psi_i|\nu|\psi_i\rangle\geq0$, hence $\sum_i k_i\langle\psi_i|\nu|\psi_i\rangle=0$ iff $\langle\psi_i|\nu|\psi_i\rangle=0$. So we obtain}
			\langle\psi_i|\nu|\psi_i\rangle=0 \ \forall\ i \stackrel{(a)}{\Rightarrow} \nu^{1/2}|\psi_i\rangle=0 \ \forall\ i \ \Rightarrow  \Pi_\nu|\psi_i\rangle=0 \ \forall\ i \ \Rightarrow  \Pi_\nu\zeta=\zeta\Pi_\nu=0.
		\end{gather*}
		$(a)$ follows from the fact that, if $\|\nu^{1/2}|\psi_i\rangle\|^2=\langle\psi_i|\nu|\psi_i\rangle=0$, then the vector $\nu^{1/2}|\psi_i\rangle=0$.
		\\\noindent\textbf{(2)$\to$(3): }$\Pi_\nu\zeta=\zeta\Pi_\nu=0\Rightarrow (\I-\Pi_\nu)\zeta=\zeta=\zeta(\I-\Pi_\nu)\Rightarrow \zeta=(\I-\Pi_\nu)\zeta(\I-\Pi_\nu)\Rightarrow \zeta\in\mathcal{P}(\I-\Pi_\nu)$. 
		\\\noindent\textbf{(3)$\to$(4): }$\zeta\in\mathcal{P}(\I-\Pi_\nu)\Rightarrow \zeta=(\I-\Pi_\nu)\zeta(\I-\Pi_\nu)\Rightarrow\Pi_\nu\zeta\Pi_\nu=\Pi_\nu(\I-\Pi_\nu)\zeta(\I-\Pi_\nu)\Pi_\nu = (\Pi_\nu-\Pi_\nu)\zeta(\Pi_\nu-\Pi_\nu)=0$.
		\\\noindent\textbf{(4)$\to$(1):} $\Pi_\nu\zeta\Pi_\nu=0\Rightarrow\nu^{1/2}\zeta\nu^{1/2}=0\Rightarrow \tr\lr{\zeta\nu}=0$.
	\end{proof}
\end{lem}

\begin{lem} \label{lem: baseEquiSubspace} For any $\nu\in\cp(\ch)$ and projector $\Pi$ such that $\Pi\nu=\nu$, then $$\zeta\in\cp(\Pi_\nu)\Rightarrow\zeta\in\cp(\Pi).$$ 
	\begin{proof}
		By taking Hermitian on both sides and multiplying both sides with $\nu^{-1}$ and $\nu$, we obtain  $$\Pi\nu=\nu\Leftrightarrow\nu\Pi=\nu\Leftrightarrow \Pi\Pi_\nu=\Pi_\nu\Leftrightarrow\Pi_\nu\Pi=\Pi_\nu.$$%Now multiplying by $\nu^{-1}$, we get $$\Pi\nu=\nu\Rightarrow\Pi\Pi_\nu=\Pi_\nu\text{ and }\nu\Pi=\nu\Rightarrow\Pi_\nu\Pi=\Pi_\nu.$$ On multiplying by $\nu$, we get the inverse condition as		$$\Pi\Pi_\nu=\Pi_\nu\Rightarrow\Pi\nu=\nu\text{ and }\Pi_\nu\Pi=\Pi_\nu\Rightarrow\nu\Pi=\nu.$$ So, $\Pi\nu=\nu\Leftrightarrow\nu\Pi=\nu\Leftrightarrow\Pi\Pi_\nu=\Pi_\nu\Leftrightarrow\Pi_\nu\Pi=\Pi_\nu$. 
		Now, for any $\zeta\in\cp(\Pi_\nu)$, we get
		$$\zeta=\Pi_\nu\zeta\Pi_\nu=\Pi\Pi_\nu\zeta\Pi_\nu\Pi= \Pi\zeta\Pi\Rightarrow\zeta\in\mathcal{P}(\Pi).$$
		Hence, $\zeta\in\cp(\Pi_\nu)\Rightarrow \zeta\in\cp(\Pi).$ This completes the proof. 
	\end{proof}	
\end{lem}

\begin{lem} \label{lem: rotatedSpace}
	For any $\nu\in\cp(\ch)$ and projector $\Pi$ such that $\Pi\Pi_\nu=\Pi_\nu\Pi=\Pi$, then $$\proj\Pi_\nu=\proj,\text{ where }\proj=\Pi_{\nu^{-1}\Pi\nu^{-1}}.$$
	\begin{proof} Using the definition of projector, we obtained that $\proj=\Pi_{\nu^{-1}\Pi\nu^{-1}}=\lr{\nu^{-1}\Pi\nu^{-1}}^{-1}\lr{\nu^{-1}\Pi\nu^{-1}}$ and $\nu^{-1}\Pi_\nu = \nu^{-1}$. Now, 
	\begin{align*}
		\proj\Pi_\nu&=\Pi_{\nu^{-1}\Pi\nu^{-1}}\Pi_{\nu}=(\nu^{-1}\Pi\nu^{-1})^{-1}(\nu^{-1}\Pi\nu^{-1})\Pi_\nu\\
		&=(\nu^{-1}\Pi\nu^{-1})^{-1}\nu^{-1}\Pi\nu^{-1}\Pi_\nu=(\nu^{-1}\Pi\nu^{-1})^{-1}\nu^{-1}\Pi\nu^{-1}=\proj,
	\end{align*}
	which is the desired result.	
	\end{proof}
\end{lem}

\begin{lem} \label{lem: inverseEqui}  For any $\nu\in\cp(\ch)$ and projector $\Pi$ such that $\Pi\Pi_\nu=\Pi_\nu\Pi=\Pi$, then
	\begin{gather*}
		\zeta\in\cp(\Pi_{\nu\proj\nu})\Rightarrow \zeta\in\cp(\Pi),
	\end{gather*}
	where $\proj=\Pi_{\nu^{-1}\Pi\nu^{-1}}=\lr{\nu^{-1}\Pi\nu^{-1}}\lr{\nu^{-1}\Pi\nu^{-1}}^{-1}$.
	\begin{proof} Note that $\nu\proj\nu = \nu\lr{\nu^{-1}\Pi\nu^{-1}}\lr{\nu^{-1}\Pi\nu^{-1}}^{-1}\nu 
		= {\Pi_\nu\Pi\nu^{-1}}\lr{\nu^{-1}\Pi\nu^{-1}}^{-1}\nu$, hence
		\begin{align*}
			\Pi\nu\proj\nu = {\Pi\Pi_\nu\Pi\nu^{-1}}\lr{\nu^{-1}\Pi\nu^{-1}}^{-1}\nu={\Pi_\nu\Pi\nu^{-1}}\lr{\nu^{-1}\Pi\nu^{-1}}^{-1}\nu=\nu\proj\nu.
		\end{align*}
		Therefore, we have $\Pi\nu\proj\nu=\nu\proj\nu$. Now, using the Lemma \ref{lem: baseEquiSubspace}, we obtain $\zeta\in\cp(\Pi_{\nu\proj\nu})\Rightarrow \zeta\in\cp(\Pi).$
	\end{proof}
\end{lem}

\begin{lem} \label{lem: modSupbspace} For projector $\proj$ such that $\proj\Pi_\nu=\proj$, then
	$$\zeta\in\mathcal{P}(\proj)\Rightarrow\nu^{-1}\zeta\nu^{-1}\in\mathcal{P}(\Pi_{\nu^{-1}\proj\nu^{-1}}).$$
	\begin{proof} Note that $\zeta\in\mathcal{P}(\proj)$ and $\proj\Pi_\nu=\proj$, hence, $$(\nu^{-1}\proj\nu^{-1})(\nu\zeta\nu)(\nu^{-1}\proj\nu^{-1})=\nu^{-1}\proj\Pi_\nu\zeta\Pi_\nu\proj\nu^{-1}=\nu^{-1}\proj\zeta\proj\nu^{-1}=(\nu^{-1}\zeta\nu^{-1}).$$
		Now, to prove $\nu^{-1}\zeta\nu^{-1}\in\mathcal{P}(\Pi_{\nu^{-1}\proj\nu^{-1}})$, we show that
		\begin{align*}
			\Pi_{\nu^{-1}\proj\nu^{-1}}(\nu^{-1}\zeta\nu^{-1})\Pi_{\nu^{-1}\proj\nu^{-1}}
			&=\Pi_{\nu^{-1}\proj\nu^{-1}}(\nu^{-1}\proj\nu^{-1})(\nu\zeta\nu)(\nu^{-1}\proj\nu^{-1})\Pi_{\nu^{-1}\proj\nu^{-1}}\\
			&\stackrel{(a)}{=}(\nu^{-1}\proj\nu^{-1})\nu\zeta\nu(\nu^{-1}\proj\nu^{-1})=\nu^{-1}\zeta\nu^{-1}.
		\end{align*}
		$(a)$ is obtained from the fact that $\eta\Pi_\eta=\Pi_\eta\eta=\eta$ on taking $\eta=\nu^{-1}\proj\nu^{-1}$. Thus, we get  $\nu^{-1}\zeta\nu^{-1}\in\mathcal{P}(\Pi_{\nu^{-1}\proj\nu^{-1}})$, which completes the proof.
	\end{proof}
\end{lem}

\begin{thm} \label{lem: genMainProj}
	Given $\Pi, \nu\in\mathcal{P}(\hilbert), \Pi\Pi_\nu=\Pi$, then the following statements regarding constraints on $\zeta\in\mathcal{P}(\hilbert)$ are equivalent
	\begin{enumerate}[label={$(\mathrm{\arabic*})$}]
		\item $\nu^{1/2}\zeta\nu^{1/2}\in\mathcal{P}(\Pi)$,
		\item $\Pi_\nu\zeta\Pi_\nu\in\mathcal{P}(\proj)$,
		\item $\zeta\in\mathcal{P}(\I-\Pi_\nu+\proj)$,
	\end{enumerate}
	where $\proj=\Pi_{\nu^{-1/2}\Pi\nu^{-1/2}}$.
	\begin{proof} We will first prove $(1)\Leftrightarrow(2)$, and then $(2)\Leftrightarrow(3)$.
		\\\noindent\textbf{(1)$\to$(2): }This statement can be obtained by appropriate substitution in Lemma \ref{lem: modSupbspace}.
		\\\noindent\textbf{(2)$\to$(1): }Using Lemma \ref{lem: modSupbspace}, we get
		$\Pi_\nu\zeta\Pi_\nu\in\mathcal{P}(\proj)\Rightarrow \nu^{1/2}\zeta\nu^{1/2}\in\mathcal{P}(\Pi_{\nu^{1/2}\proj\nu^{1/2}}).$ Then, using Lemma \ref{lem: inverseEqui}, we get
		$\nu^{1/2}\zeta\nu^{1/2}\in\mathcal{P}(\Pi_{\nu^{1/2}\proj\nu^{1/2}})\Rightarrow\nu^{1/2}\zeta\nu^{1/2}\in\mathcal{P}(\Pi).$
		\\\noindent\textbf{(2)$\to$(3): }Using Lemma \ref{lem: rotatedSpace}, we know that $\Pi\Pi_\nu=\Pi\Rightarrow \proj=\proj\Pi_\nu$. Hence
		\begin{align*}
			\Pi_\nu\zeta\Pi_\nu\in\mathcal{P}(\proj)&\Rightarrow \proj\Pi_\nu\zeta\Pi_\nu\proj=\proj\Pi_\nu\zeta\Pi_\nu=\Pi_\nu\zeta\Pi_\nu\proj=\Pi_\nu\zeta\Pi_\nu\\
			&\Rightarrow \proj\zeta\proj=\Pi_\nu\zeta\proj=\proj\zeta\Pi_\nu=\Pi_\nu\zeta\Pi_\nu\\
			&\Rightarrow (\Pi_\nu-\proj)\zeta(\Pi_\nu-\proj)=0.
			\intertext{Using Lemma \ref{lem: null}, $(\Pi_\nu-\proj)\zeta(\Pi_\nu-\proj)=0\Rightarrow(\Pi_\nu-\proj)\zeta=\zeta(\Pi_\nu-\proj)=0$ and so}
			\Pi_\nu\zeta\Pi_\nu\in\mathcal{P}(\proj)&\Rightarrow 
			(\I-\Pi_\nu+\proj)\zeta(\I-\Pi_\nu+\proj)=\zeta \Rightarrow \zeta\in \mathcal{P}(\I-\Pi_\nu+\proj).
		\end{align*}
		\noindent\textbf{(3)$\to$(2): }Using Lemma \ref{lem: rotatedSpace}, we know that $\Pi\Pi_\nu=\Pi\Rightarrow \proj=\proj\Pi_\nu$. Hence
		\begin{align*}		
			\zeta\in \mathcal{P}(\I-\Pi_\nu+\proj) &\Rightarrow \zeta=(\I-\Pi_\nu+\proj)\zeta\\
			& \Rightarrow (\Pi_\nu-\proj)\zeta=\zeta(\Pi_\nu-\proj)=0\\
			&\Rightarrow \Pi_\nu\zeta=\proj\zeta, \zeta\Pi_\nu=\zeta\proj\\
			& \Rightarrow \proj\Pi_\nu\zeta\Pi_\nu\proj=\proj\zeta\proj=\Pi_\nu\zeta\Pi_\nu\\
			&\Rightarrow \Pi_\nu\zeta\Pi_\nu\in\mathcal{P}(\proj).
		\end{align*}
		Thus proved.
	\end{proof}
\end{thm}

\begin{thm}\label{thm: projSubSet} For any $\sigma\in\mathcal{P}(\hilbert)$ and projector $\proj$ such that $\proj\Pi_\sigma=\proj$ then,
	\begin{enumerate}[label={$(\mathrm{\arabic*})$}]
		\item $\zeta\in\mathcal{P}(\proj)\Rightarrow\zeta\in\mathcal{P}(\I-\Pi_\sigma+\proj)$ and so $\mathcal{P}(\proj)\subseteq\mathcal{P}(\I-\Pi_\sigma+\proj)$, and
		\item $\zeta\in\mathcal{P}(\I-\Pi_\sigma+\proj)\Rightarrow\Pi_\sigma\zeta\Pi_\sigma\in\mathcal{P}(\proj)$ and so $\{\Pi_\sigma\zeta\Pi_\sigma:\zeta\in \mathcal{P}(\I-\Pi_\sigma+\proj)\}\subseteq\mathcal{P}(\proj)$.
	\end{enumerate}
\begin{proof} We will prove $\mathrm{(1)}$, and then $\mathrm{(2)}$.\\
	\noindent\textbf{(1): }Beginning with the definition of $\mathcal{P}(\proj)$, we get 
	$\zeta\in\mathcal{P}(\proj)\Rightarrow\zeta=\proj\zeta$. Now, multiplying both sides by $\Pi_\sigma$ and simplifying further, we obtain $\Pi_\sigma\zeta=\Pi_\sigma\proj\zeta=\proj\zeta$. So, we get $(\I-\Pi_\sigma+\proj)\zeta=\zeta$ and thus $\zeta=\mathcal{P}(\I-\Pi_\sigma+\proj)$. Now, note that $\zeta\in\mathcal{P}(\I-\Pi_\sigma+\proj)\ \forall \ \zeta\in\mathcal{P}(\proj)$, so $\mathcal{P}(\proj)\subseteq\mathcal{P}(\I-\Pi_\sigma+\proj)$.\\
	\noindent\textbf{(2): }Beginning with the definition of $\mathcal{P}(\I-\Pi_\sigma+\proj)$ and simplifying further, we get 
	$$\zeta\in\mathcal{P}(\I-\Pi_\sigma+\proj)\Rightarrow\zeta=(\I-\Pi_\sigma+\proj)\zeta\Rightarrow\Pi_\sigma\zeta=\proj\zeta=\proj\Pi_\sigma\zeta.$$
	Similarly, it can be shown that $\zeta\proj=\zeta\Pi_\sigma\proj$. So, we get 
	$\Pi_\sigma\zeta\Pi_\sigma=\proj\Pi_\sigma\zeta\Pi_\sigma\proj$. Now, by the definition of $\mathcal{P}(\proj)$, we get the stated result as $\Pi_\sigma\zeta\Pi_\sigma\in\mathcal{P}(\proj)$. Further, this is true for all $\zeta\in\mathcal{P}(\I-\Pi_\sigma+\proj)$, so we get $\{\Pi_\sigma\zeta\Pi_\sigma:\zeta\in \mathcal{P}(\I-\Pi_\sigma+\proj)\}\subseteq\mathcal{P}(\proj)$.
\end{proof}
\end{thm}
\section{Minimizing the max norm}
\label{sec: maxNormMinimization}

\begin{comment}
\begin{lem} \label{lem: projected norm}
	For any $\nu\in\mathcal{P}(\hilbert)$ and projector $\Pi$,
	$\|\Pi\nu\Pi\|_{\infty}\leq\|\nu\|_{\infty}$.
	\begin{proof} Note that $\ds \langle\phi|\Pi\nu\Pi|\phi\rangle=\langle\phi|\Pi|\phi\rangle\frac{\langle\phi|\Pi}{\sqrt{\langle\phi|\Pi|\phi\rangle}}\nu\frac{\Pi|\phi\rangle}{\sqrt{\langle\phi|\Pi|\phi\rangle}}\leq\frac{\langle\phi|\Pi}{\sqrt{\langle\phi|\Pi|\phi\rangle}}\nu\frac{\Pi|\phi\rangle}{\sqrt{\langle\phi|\Pi|\phi\rangle}}.$ Beginning with the definition of max-norm and using this inequality, we get
		\begin{align*}
			\|\Pi\nu\Pi\|_{\infty}&=\max_{|\phi\rangle}\langle\phi|\Pi\nu\Pi|\phi\rangle
			\leq \max_{|\phi\rangle}\frac{\langle\phi|\Pi}{\sqrt{\langle\phi|\Pi|\phi\rangle}}\nu\frac{\Pi|\phi\rangle}{\sqrt{\langle\phi|\Pi|\phi\rangle}}\\
			&\stackrel{(a)}{=} \max_{|\phi\rangle:|\phi\rangle=\Pi|\phi\rangle}\langle\phi|\nu|\phi\rangle\stackrel{(b)}{\leq} \max_{|\phi\rangle}\langle\phi|\nu|\phi\rangle\stackrel{(c)}{=}\|\nu\|_{\infty}.
		\end{align*}
	Here, $(a)$ follows from the fact that $\left\{\frac{\Pi|\phi\rangle}{\sqrt{\langle\phi|\Pi|\phi\rangle}}: |\phi\rangle\right\} = \left\{\phi:\Pi|\phi\rangle=|\phi\rangle\right\}$.  $(b)$ is obtained because the restriction $\{|\phi\rangle:|\phi\rangle=\Pi|\phi\rangle\}$ is removed and $(c)$ is obtained from the definition of max-norm.
	\end{proof}
\end{lem}
\end{comment}

\begin{lem}\label{lem: singleMin} For any $\sigma\in\mathcal{S}(\hilbert)$, and projector $\proj$ such that $\proj\Pi_\sigma=\proj$, then
	$$\min_{\psi\in\mathcal{S}(\I-\Pi_\sigma+\proj)}\left\|\frac{\psi}{\tr\lr{\psi\sigma}}\right\|_{\infty}=\frac{1}{\tr\lr{\proj \sigma}},$$	
	with minimum is achieved at $\psi={\proj}/{\tr\lr{\proj}}$.
	\begin{proof}
		We will prove the statement in lemma in 3 steps by showing the following as equal.
		\begin{enumerate}
			\item \label{it1} $\ds\min_{\psi\in\mathcal{S}(\I-\Pi_\sigma+\proj)}\left\|\frac{\psi}{\tr\lr{\psi\sigma}}\right\|_{\infty}$
			\item \label{it2}
			$\ds\min_{\psi\in\mathcal{P}(\proj)}\left\|\frac{\psi}{\tr\lr{\psi\sigma}}\right\|_{\infty}$
			%\item \label{it4}$\ds \min_{\tilde\psi\in\mathcal{S}(\Pi_{\sigma^{1/2}\proj\sigma^{1/2}})}\left\|\sigma^{-1/2}\tilde\psi\sigma^{-1/2}\right\|_{\infty}$
			\item \label{it3}$\ds\frac{1}{\tr\lr{\proj \sigma}}$
		\end{enumerate}
		\textbf{\ref{it1}$\to$\ref{it2}: }Observe that taking any $\tilde c>0$, we get
		$$\ds\min_{\psi\in\mathcal{S}(\I-\Pi_\sigma+\proj)}\left\|\frac{\psi}{\tr\lr{\psi\sigma}}\right\|_{\infty}=\min_{\psi\in\mathcal{S}(\I-\Pi_\sigma+\proj)}\left\|\frac{\tilde c\psi}{\tr\lr{\tilde c \psi\sigma}}\right\|_{\infty}=\min_{\psi\in\mathcal{P}(\I-\Pi_\sigma+\proj)}\left\|\frac{\psi}{\tr\lr{\psi\sigma}}\right\|_{\infty}.$$
		Now, using %note that $\tr\lr{\psi\sigma}$ is constant and so using Lemma \ref{lem: projected norm}, we get
		$\ds \left\|{\Pi_\sigma\psi\Pi_\sigma}\right\|_{\infty}\leq\left\|{\psi}\right\|_{\infty}$, we get %so
		\begin{gather}
			\min_{\psi\in\mathcal{P}(\I-\Pi_\sigma+\proj)}\left\|\frac{\Pi_\sigma\psi\Pi_\sigma}{\tr\lr{\psi\sigma}}\right\|_{\infty}\leq\min_{\psi\in\mathcal{P}(\I-\Pi_\sigma+\proj)}\left\|\frac{\psi}{\tr\lr{\psi\sigma}}\right\|_{\infty}.
		\end{gather}
		Substituting $\psi'=\Pi_\sigma\psi\Pi_\sigma$, and from Theorem \ref{thm: projSubSet} in Appendix \ref{sec: genProj}, we know that $\psi'\in\{\Pi_\sigma\psi\Pi_\sigma:\psi\in \mathcal{P}(\I-\Pi_\sigma+\proj)\}\subseteq\mathcal{P}(\proj)$. Hence,
		\begin{gather}	
			\min_{\psi'\in\mathcal{P}(\proj)}\left\|\frac{\psi'}{\tr\lr{\psi'\sigma}}\right\|_{\infty}\leq\min_{\psi'\in\{\Pi_\sigma\psi\Pi_\sigma:\psi\in \mathcal{P}(\I-\Pi_\sigma+\proj)\}}\left\|\frac{\psi'}{\tr\lr{\psi'\sigma}}\right\|_{\infty}=\min_{\psi\in\mathcal{P}(\I-\Pi_\sigma+\proj)}\left\|\frac{\Pi_\sigma\psi\Pi_\sigma}{\tr\lr{\psi\sigma}}\right\|_{\infty}\\
			\Rightarrow \min_{\psi\in\mathcal{P}(\proj)}\left\|\frac{\psi}{\tr\lr{\psi\sigma}}\right\|_{\infty}\leq\min_{\psi\in\mathcal{P}(\I-\Pi_\sigma+\proj)}\left\|\frac{\psi}{\tr\lr{\psi\sigma}}\right\|_{\infty}. \label{eq: lowerminMaxNormProof1}
		\end{gather}
		But, from Theorem \ref{thm: projSubSet} in Appendix \ref{sec: genProj}, we know that $\mathcal{P}(\proj)\subseteq\mathcal{P}(\I-\Pi_\sigma+\proj)$, hence
		\begin{gather}
			\min_{\psi\in\mathcal{P}(\I-\Pi_\sigma+\proj)}\left\|\frac{\psi}{\tr\lr{\psi\sigma}}\right\|_{\infty}\leq\min_{\psi\in\mathcal{P}(\proj)}\left\|\frac{\psi}{\tr\lr{\psi\sigma}}\right\|_{\infty}. \label{eq: lowerminMaxNormProof2}
			\intertext{	Combining the inequalities in \eqref{eq: lowerminMaxNormProof1} and \eqref{eq: lowerminMaxNormProof2}, we get}
			\min_{\psi\in\mathcal{P}(\I-\Pi_\sigma+\proj)}\left\|\frac{\psi}{\tr\lr{\psi\sigma}}\right\|_{\infty}=\min_{\psi\in\mathcal{P}(\proj)}\left\|\frac{\psi}{\tr\lr{\psi\sigma}}\right\|_{\infty}. \label{eq: projEquiNorm}
		\end{gather}
		\textbf{\ref{it2}$\to$\ref{it3}: } Taking $\ds \phi=\frac{\psi}{\|\psi\|_{\infty}}$, we get the constraint as $\phi\in\mathcal{P}(\proj)$ and $\|\phi\|_{\infty}=1$ and minimum as
		\begin{align*}
			\min_{\psi\in\mathcal{P}(\proj)}\left\|\frac{\psi}{\tr\lr{\psi\sigma}}\right\|_{\infty}=\min_{\phi\in\mathcal{P}(\proj),\|\phi\|_{\infty}=1} \frac{1}{\tr\lr{\phi\sigma}}=\lr{\max_{\phi\in\mathcal{P}(\proj),\|\phi\|_{\infty}=1}\tr\lr{\phi\sigma}}^{-1}.
		\end{align*}
		Focusing on $\phi$, $$\|\phi\|_{\infty}=1\Rightarrow \phi\leq \I\Rightarrow \proj\phi\proj\leq\proj\Rightarrow \phi\leq \proj\Rightarrow\sigma^{1/2}\phi\sigma^{1/2}\leq \sigma^{1/2}\proj\sigma^{1/2}\Rightarrow \tr\lr{\phi\sigma}\leq\tr\lr{\proj\sigma}.$$
		Also, equality is obtained as $\phi=\proj$ and so $\ds \max_{\phi\in\mathcal{P}(\proj),\|\phi\|_{\infty}=1}\tr\lr{\phi\sigma}=\tr\lr{\proj\sigma}$. Thus we obtain,
		\begin{align*}
			\min_{\psi\in\mathcal{P}(\proj)}\left\|\frac{\psi}{\tr\lr{\psi\sigma}}\right\|_{\infty}=\frac{1}{\tr\lr{\proj\sigma}}.
		\end{align*}
		For $\psi\in\mathcal{S}(\I-\Pi_\sigma+\proj)$, we get equality at $\ds \psi=\frac{\proj}{\tr\lr{\proj}}$.
		\begin{comment}
		On substituting $\phi=\sigma^{1/2}\psi\sigma^{1/2}$, we obtain
		\begin{gather}
			\left\|\frac{\psi}{\tr\lr{\psi\sigma}}\right\|_{\infty}=\left\|\frac{\sigma^{-1/2}\phi\sigma^{-1/2}}{\tr\lr{\phi}}\right\|_{\infty}.
			\intertext{Note that the mapping is invertible and one to one. Hence,}
			\min_{\psi\in\mathcal{P}(\proj)}\left\|\frac{\psi}{\tr\lr{\psi\sigma}}\right\|_{\infty}=\min_{\phi\in\mathcal{P}(\Pi)}\left\|\frac{\sigma^{-1/2}\phi\sigma^{-1/2}}{\tr\lr{\phi}}\right\|_{\infty}.
			\intertext{Now, taking $\phi$ such that $\tr\lr{\phi}=1$, we get}
			\min_{\phi\in\mathcal{P}(\Pi)}\left\|\frac{\sigma^{-1/2}\phi\sigma^{-1/2}}{\tr\lr{\phi}}\right\|_{\infty}=\min_{\phi\in\mathcal{S}(\Pi)}\left\|\sigma^{-1/2}\phi\sigma^{-1/2}\right\|_{\infty}.
		\end{gather}
		Now, 
		\begin{align}
			\min_{\phi\in\mathcal{S}(\Pi)}\left\|\sigma^{-1/2}\phi\sigma^{-1/2}\right\|_{\infty}&=\min\left\{C:\sigma^{-1/2}\phi\sigma^{-1/2}\leq C\proj \ \text{for some } {\phi\in\mathcal{S}(\Pi)}\right\}\\
			&=\min\left\{C:\phi\leq C\sigma^{1/2}\proj\sigma^{1/2} \ \text{for some } {\phi\in\mathcal{S}(\Pi)}\right\}\\
			&=\frac{1}{\tr\lr{\proj\sigma}},
		\end{align}
		at $\ds \phi= \frac{\sigma^{1/2}\proj\sigma^{1/2}}{\tr\lr{\proj\sigma}}$.
		\end{comment}		
	\end{proof}
\end{lem}
\begin{lem}\label{lem: doubleMin} For any $\rho,\sigma$ such that $\Pi_\rho=\Pi_\sigma$, $0<c_r<1$,
	\begin{align*}
		\min_{\substack{\psi_{\max}\in\mathcal{S}(\I-\Pi_\sigma+\proj^{\max}),\\\psi_{\min}\in\mathcal{S}(\I-\Pi_\sigma+\proj^{\min})}} \left\|\frac{c_r\psi_{\max}}{\tr\lr{\psi_{\max}\sigma}}+\frac{(1-c_r)\psi_{\min}}{\tr\lr{\psi_{\min}\sigma}}\right\|_{\infty}=\lr{\Upsilon_{\Tau^{\max},\Tau^{\min}}\lr{\sigma^{1/2}\Pi_{\proj^{\max}+\proj^{\min}}\sigma^{1/2},c_r}}^{-1}.
	\end{align*}
		Here, $\Upsilon_{\Pi_1,\Pi_2}(\sigma,r)$ is defined as %for any $\tr\lr{\Pi_1\Pi_2}=0$, we have
	$$\Upsilon_{\Pi_1,\Pi_2}(\sigma,r)\isdefinedas\{\max c: c(r\psi_1+(1-r)\psi_2)\leq\sigma \text{ for some }\psi_1\in\mathcal{S}(\Pi_1),\psi_2\in\mathcal{S}(\Pi_2), \text{ with }\tr\lr{\Pi_1\Pi_2}=0\}.$$
	More detailed properties of $\Upsilon_{\Pi_1,\Pi_2}\lr{\sigma,r}$
	are given in Appendix \ref{sec: upsilon}.
	\begin{proof}
		Following arguments similar to the previous proof. We get
		\begin{align*}
			&\min_{\substack{\psi_{\max}\in\mathcal{S}(\I-\Pi_\sigma+\proj^{\max}),\\\psi_{\min}\in\mathcal{S}(\I-\Pi_\sigma+\proj^{\min})}} \left\|\frac{c_r\psi_{\max}}{\tr\lr{\psi_{\max}\sigma}}+\frac{(1-c_r)\psi_{\min}}{\tr\lr{\psi_{\min}\sigma}}\right\|_{\infty}\\&=\min_{\substack{\psi_{\max}\in\mathcal{S}(\proj^{\max}),\\\psi_{\min}\in\mathcal{S}(\proj^{\min})}} \left\|\frac{c_r\psi_{\max}}{\tr\lr{\psi_{\max}\sigma}}+\frac{(1-c_r)\psi_{\min}}{\tr\lr{\psi_{\min}\sigma}}\right\|_{\infty}\\
			&=\min_{{\phi_{\max}\in\mathcal{S}(\Tau^{\max}),\phi_{\min}\in\mathcal{S}(\Tau^{\min})}} \left\|{c_r\sigma^{-1/2}\phi_{\max}\sigma^{-1/2}}+{(1-c_r)\sigma^{-1/2}\phi_{\min}\sigma^{-1/2}}\right\|_{\infty}.
		\end{align*}	
		Now, following on similar lines, we get		
		\begin{align*}
			&=\min\left\{C:{c_r\sigma^{-1/2}\phi_{\max}\sigma^{-1/2}}+(1-c_r){\sigma^{-1/2}\phi_{\min}\sigma^{-1/2}}\leq C\Pi_{\proj^{\max}+\proj^{\min}} \ \text{for some } {\phi_{\max}\in\mathcal{S}(\Tau^{\max}),\phi_{\min}\in\mathcal{S}(\Tau^{\min})}\right\}\\
			&=\left\{\max C:C\left(c_r\phi_{\max}+(1-c_r)\phi_{\min}\right)\leq \sigma^{1/2}\Pi_{\proj^{\max}+\proj^{\min}}\sigma^{1/2}\ \text{for some } {\phi_{\max}\in\mathcal{S}(\Tau^{\max}),\phi_{\min}\in\mathcal{S}(\Tau^{\min})}\right\}^{-1}\\
			&=\lr{\Upsilon_{\Tau^{\max},\Tau^{\min}}\lr{\sigma^{1/2}\Pi_{\proj^{\max}+\proj^{\min}}\sigma^{1/2},c_r}}^{-1},
		\end{align*}
		where last step follows from definition of $\Upsilon_{\Tau^{\max},\Tau^{\min}}$ in Appendix \ref{sec: upsilon}. Thus, completing the proof.
	\end{proof}
	
\end{lem}

\section{Definition and properties of $\Upsilon_{\Pi_1,\Pi_2}\lr{\sigma,r}$}
\label{sec: upsilon}
\begin{defn} Given a pair of orthogonal projectors $\Pi_1$ and $\Pi_2$ such that $\tr\lr{\Pi_1\Pi_2}=0$, we define a function $\Upsilon_{\Pi_1,\Pi_2}:\mathcal{P}(\hilbert)\times\mathbb{R}\to\mathbb{R}$ as
	$$\Upsilon_{\Pi_1,\Pi_2}(\sigma,r)=\{\max c: c(r\psi_1+(1-r)\psi_2)\leq\sigma \text{ for some }\psi_1\in\mathcal{S}(\Pi_1),\psi_2\in\mathcal{S}(\Pi_2)\}.$$
\end{defn}

	It denotes the largest value of trace of matrix $\psi\in\mathcal{P}(\Pi_1+\Pi_2)$, having the properties $\psi\leq\sigma$,  $\tr\lr{\psi\Pi_1}=r\tr\lr{\psi}$ and $\Pi_1\psi\Pi_2=0$.

In the rest of discussion, we consider the case when $\Pi_1+\Pi_2\leq\Pi_\sigma$. If this not the cases, $\Pi_1$ (and $\Pi_2$) can just be substituted by the projector corresponding to the subspace spanned by the intersection subspace of $\Pi_1$ (and $\Pi_2$) and $\Pi_\sigma$. The following two properties can be obtained by substituting the definition of $\Upsilon_{\Pi_1,\Pi_2}(\sigma,r)$.

\noindent\textbf{Properties:}
\begin{itemize}
	\item $\Upsilon_{\Pi_1,\Pi_2}(k\sigma,r)=k\Upsilon_{\Pi_1,\Pi_2}(\sigma,r)$,
	\item $\{\max c: c(k_1\psi_1+k_2\psi_2)\leq\sigma \text{ for some }\psi_1\in\mathcal{S}(\Pi_1),\psi_2\in\mathcal{S}(\Pi_2)\}=\frac{1}{k_1+k_2}\Upsilon_{\Pi_1,\Pi_2}\lr{\sigma,\frac{k_1}{k_1+k_2}}$.
%	\item If $\Pi_\sigma>\Pi_1+\Pi_2$, the substitution	$\Upsilon_{\Pi_1,\Pi_2}\lr{\sigma,r}=\Upsilon_{\Pi_1,\Pi_2}\lr{\sigma^{1/2}\Pi_{\sigma^{-1/2}(\Pi_1+\Pi_2)\sigma^{-1/2}}\sigma^{1/2},r}$ simplifies it as $\sigma^{1/2}\Pi_{\sigma^{-1/2}(\Pi_1+\Pi_2)\sigma^{-1/2}}\sigma^{1/2}\in\mathcal{P}(\Pi_1+\Pi_2)$. From here on-words, we can proceed with the assumption that $\Pi_1+\Pi_2=\Pi_\sigma$.
%	\item {\Red Convexity and concavity and comment on achievable minima}
\end{itemize}

\noindent\textbf{Some closed-form expressions:}
These closed form expression obtained for specific cases help in obtaining closed form expression of $A_\rho^{s}$ and $A_{\sigma}^s$ in case $\mathcal{C}3$ in Theorem \ref{thm: symmetric}.
\begin{itemize}
	\item  $\Upsilon_{\Pi_1,\Pi_2}(\sigma,0)=\tr\lr{\Pi_{\sigma^{-1/2}\Pi_2\sigma^{-1/2}}\sigma}=\tr((\sigma^{-1/2}\Pi_2\sigma^{-1/2})^{-1})$
	\item $\Upsilon_{\Pi_1,\Pi_2}(\sigma,1)=\tr\lr{\Pi_{\sigma^{-1/2}\Pi_1\sigma^{-1/2}}\sigma}=\tr((\sigma^{-1/2}\Pi_1\sigma^{-1/2})^{-1})$
	\item If $\Pi_\sigma>\Pi_1+\Pi_2$, substituting	$\Upsilon_{\Pi_1,\Pi_2}\lr{\sigma,r}=\Upsilon_{\Pi_1,\Pi_2}\lr{\sigma^{1/2}\Pi_{\sigma^{-1/2}(\Pi_1+\Pi_2)\sigma^{-1/2}}\sigma^{1/2},r}$ simplifies it as $\sigma^{1/2}\Pi_{\sigma^{-1/2}(\Pi_1+\Pi_2)\sigma^{-1/2}}\sigma^{1/2}\in\mathcal{P}(\Pi_1+\Pi_2)$. 
	%	\item {\Red Convexity and concavity and comment on achievable minima}
\end{itemize}
From here on-words, we can proceed with the assumption that $\Pi_1+\Pi_2=\Pi_\sigma$.
\begin{lem}
	If $\tr\lr{\Pi_1}=\tr\lr{\Pi_2}=1$, then
	$$\Upsilon_{\Pi_1,\Pi_2}\lr{\sigma,r}=R_{\min}(\sigma,r\Pi_1+(1-r)\Pi_2).$$
	\begin{proof}
		If $\tr\lr{\Pi_1}=1, \psi_1\in\mathcal{S}(\Pi_1)\Rightarrow \psi_1=\Pi_1$ and similarly, $\psi_2=\Pi_2$ and so,
		$$\Upsilon_{\Pi_1,\Pi_2}(\sigma,r)=\{\max c: c(r\Pi_1+(1-r)\Pi_2)\leq\sigma\}=R_{\min}(\sigma,r\Pi_1+(1-r)\Pi_2).$$
		Last step is obtained from the definition of $R_{\min}(\cdot,\cdot)$.
	\end{proof}
\end{lem}

\begin{lem}
	If $\sigma=\Pi_1\sigma\Pi_1+\Pi_2\sigma\Pi_2$ or say $\Pi_1\sigma\Pi_2=0$, then
	$$\Upsilon_{\Pi_1,\Pi_2}\lr{\sigma,r}=\begin{cases}
		\ds\frac{\tr\lr{\Pi_2\sigma}}{1-r},&\ds r\leq\frac{\tr\lr{\Pi_1\sigma}}{\tr\lr{\sigma}},\vspace{1mm}\\
		\ds\frac{\tr\lr{\Pi_1\sigma}}{r},&\text{otherwise}.
	\end{cases}.$$
	\begin{proof} From the definition,
		$\Upsilon_{\Pi_1,\Pi_2}(\sigma,r)=\{\max c: c(r\Pi_1+(1-r)\Pi_2)\leq\sigma\}.$ Hence,
		\begin{align*}
			c(r\psi_1+(1-r)\psi_2)&\leq \sigma =\Pi_1\sigma\Pi_1+\Pi_2\sigma\Pi_2\\
			&\Rightarrow\begin{cases}
				cr\psi_1\leq\Pi_1\sigma\Pi_1 &\Rightarrow cr\leq\tr\lr{\Pi_1\sigma}\\\text{and}\\
				c(1-r)\psi_2\leq\Pi_2\sigma\Pi_2&\Rightarrow c(1-r)\leq\tr\lr{\Pi_2\sigma}
			\end{cases}\\
		&\Rightarrow c\leq \min\left\{\frac{\tr\lr{\Pi_1\sigma}}{r},\frac{\tr\lr{\Pi_2\sigma}}{1-r}\right\}.
		\end{align*}
	So, $\ds\Upsilon_{\Pi_1,\Pi_2}\lr{\sigma,r}\leq \min\left\{\frac{\tr\lr{\Pi_1\sigma}}{r},\frac{\tr\lr{\Pi_2\sigma}}{1-r}\right\}$. Now, note that taking $\ds\psi_1=\frac{\Pi_1\sigma\Pi_1}{\tr\lr{\Pi_1\sigma}},\psi_2=\frac{\Pi_2\sigma\Pi_2}{\tr\lr{\Pi_2\sigma}}$, and $\ds c=\min\left\{\frac{\tr\lr{\Pi_1\sigma}}{r},\frac{\tr\lr{\Pi_2\sigma}}{1-r}\right\}$, we get
	$$\ds c(r\psi_1+(1-r)\psi_2)\leq\min\left\{\frac{\tr\lr{\Pi_1\sigma}}{r},\frac{\tr\lr{\Pi_2\sigma}}{1-r}\right\}\lr{r\frac{\Pi_1\sigma\Pi_1}{\tr\lr{\Pi_1\sigma}}+(1-r)\frac{\Pi_2\sigma\Pi_2}{\tr\lr{\Pi_2\sigma}}}\leq \Pi_1\sigma\Pi_1+\Pi_2\sigma\Pi_2=\sigma.$$
	So, we obtain  $\ds\Upsilon_{\Pi_1,\Pi_2}\lr{\sigma,r}= \min\left\{\frac{\tr\lr{\Pi_1\sigma}}{r},\frac{\tr\lr{\Pi_2\sigma}}{1-r}\right\}$.
	\end{proof}
\end{lem}

\section{Proof of Theorem \ref{thm: beyondSymmetricSplit}}
\label{proof: beyondSymmetricSplit}
We will use frequently here that $\tr\lr{\Gamma\nu}=0\Leftrightarrow\Gamma\in\mathcal{P}(\I-\Pi_\nu)$ (See Lemma \ref{lem: null} in Appendix \ref{sec: genProj} for proof).
		Starting with the definition of set $\mathcal{M}$ and the condition $e(\Lambda)=0$, the goal to find set of all ${\Lambda}$ such that $\Lr\geq0,\Ls\geq0,\Lr+\Ls\leq\I$ and $e(\Lambda)=0$. So, 
		\begin{align}
			\mathcal{E}_s(\rho,\sigma,p)&=\{\Lambda:e(\Lambda)=0,\Lr\geq0,\Ls\geq0,\Lr+\Ls\leq\I\}\nonumber\\
			&= \left\{\Lambda: \frac{\tr(p_\sigma\Lambda_\rho\sigma+p_\rho\Lambda_\sigma\rho)}{\tr((\Lambda_\rho+\Lambda_\sigma)(p_\rho\rho+p_\sigma\sigma))}=0,\Lr\geq0, \Ls\geq0,\Lr+\Ls\leq\I\right\}.\label{eq: firstSimple}
		\end{align}
		The previous step follows from using the definition of $e(\Lambda)$. Now, for a fraction to be $0$, numerator must be zero and denominator must remain non-zero. Hence, we get
		\begin{align*}
			\frac{\tr(p_\sigma\Lambda_\rho\sigma+p_\rho\Lambda_\sigma\rho)}{\tr((\Lambda_\rho+\Lambda_\sigma)(p_\rho\rho+p_\sigma\sigma))}=0&\Leftrightarrow\tr(p_\sigma\Lambda_\rho\sigma+p_\rho\Lambda_\sigma\rho)=0, \tr((\Lambda_\rho+\Lambda_\sigma)(p_\rho\rho+p_\sigma\sigma))\neq 0\\
			&\Leftrightarrow \tr(p_\sigma\Lambda_\rho\sigma+p_\rho\Lambda_\sigma\rho)=0, p_\rho\tr\lr{\Lr\rho}+p_\sigma\tr\lr{\Ls\sigma}\neq 0.
		\end{align*}
		Previous step is obtained by substituting $\tr(p_\sigma\Lambda_\rho\sigma+p_\rho\Lambda_\sigma\rho)=0$, and thus $\tr((\Lambda_\rho+\Lambda_\sigma)(p_\rho\rho+p_\sigma\sigma))=p_\rho\tr\lr{\Lr\rho}+p_\sigma\tr\lr{\Ls\sigma}$. Now using $\tr(p_\sigma\Lambda_\rho\sigma+p_\rho\Lambda_\sigma\rho)=0\Leftrightarrow\tr\lr{\Ls\rho}=0,\tr\lr{\Lr\sigma}=0$, we get
		\begin{align*}
			\frac{\tr(p_\sigma\Lambda_\rho\sigma+p_\rho\Lambda_\sigma\rho)}{\tr((\Lambda_\rho+\Lambda_\sigma)(p_\rho\rho+p_\sigma\sigma))}=0&\Leftrightarrow\tr\lr{\Ls\rho}=0,\tr\lr{\Lr\sigma}=0, p_\rho\tr\lr{\Lr\rho}+p_\sigma\tr\lr{\Ls\sigma}\neq 0.
		\end{align*}
		Substituting it in \eqref{eq: firstSimple}, we get $\mathcal{E}_s(\rho,\sigma,p)=$
		\begin{align}
			\left\{\Lambda: \tr\lr{\Ls\rho}=0,\tr\lr{\Lr\sigma}=0, p_\rho\tr\lr{\Lr\rho}+p_\sigma\tr\lr{\Ls\sigma}\neq 0,\Lr\geq0, \Ls\geq0,\Lr+\Ls\leq\I\right\}.
		\end{align}
		Substituting $\tr\lr{\Lr\sigma}=0,\Lr\geq0\Leftrightarrow\Lr\in\mathcal{P}(\I-\Pi_\sigma)$ and similarly $\tr\lr{\Ls\rho}=0,\Ls\geq0\Leftrightarrow\Ls\in\mathcal{P}(\I-\Pi_\rho)$, we get
		\begin{align*}
			\mathcal{E}_s(\rho,\sigma,p)&=
			\{\Lambda:\Lr\in\mathcal{P}(\I-\Pi_\sigma),\Ls\in\mathcal{P}(\I-\Pi_\rho), p_\rho\tr\lr{\Lr\rho}+p_\sigma\tr\lr{\Ls\sigma}\neq 0,\Lr+\Ls\leq\I\}.
			\intertext{Now at max one of $p_\rho\tr\lr{\Lr\rho}$ and $p_\sigma\tr\lr{\Ls\sigma}$ can be 0 to ensure that $p_\rho\tr\lr{\Lr\rho}+p_\sigma\tr\lr{\Ls\sigma}\neq 0$, thus giving rise to one of the three cases, any POVM must satisfy as given below}
			\mathcal{E}_s(\rho,\sigma,p)&=\left\{\Lambda:\begin{cases}
				\Lr\in\mathcal{P}(\I-\Pi_\sigma),\Ls\in\mathcal{P}(\I-\Pi_\rho),\tr\lr{\Lr\rho}\neq0,\tr\lr{\Ls\sigma}=0,\Lr+\Ls\leq\I\\\text{ OR}\\
				\Lr\in\mathcal{P}(\I-\Pi_\sigma),\Ls\in\mathcal{P}(\I-\Pi_\rho),\tr\lr{\Lr\rho}=0,\tr\lr{\Ls\sigma}\neq0,\Lr+\Ls\leq\I\\\text{ OR}\\
				\Lr\in\mathcal{P}(\I-\Pi_\sigma),\Ls\in\mathcal{P}(\I-\Pi_\rho),\tr\lr{\Lr\rho}\neq0,\tr\lr{\Ls\sigma}\neq0,\Lr+\Ls\leq\I
			\end{cases}\right\}.
			\intertext{ Note that $\Ls\in\mathcal{P}(\I-\Pi_\rho),\tr\lr{\Ls\sigma}=0\Leftrightarrow\Ls\in\mathcal{P}(\I-\Pi_\rho),\Ls\in\mathcal{P}(\I-\Pi_\sigma)\Leftrightarrow\Ls\in\mathcal{P}(\I-\Pi_{\rho+\sigma})$ for the first case. Similarly $\Lr\in\mathcal{P}(\I-\Pi_\sigma),\tr\lr{\Lr\rho}=0\Leftrightarrow\Lr\in\mathcal{P}(\I-\Pi_{\rho+\sigma})$ for the second case. Substituting these, we obtain}
			\mathcal{E}_s(\rho,\sigma,p)&=\left\{\Lambda:\begin{cases}					\Lr\in\mathcal{P}(\I-\Pi_\sigma),\Ls\in\mathcal{P}(\I-\Pi_{\rho+\sigma}),\tr\lr{\Lr\rho}\neq0,\Lr+\Ls\leq\I\qquad \text{ OR}\\
				\Lr\in\mathcal{P}(\I-\Pi_{\rho+\sigma}),\Ls\in\mathcal{P}(\I-\Pi_\rho),\tr\lr{\Ls\sigma}\neq0,\Lr+\Ls\leq\I\qquad\text{ OR}\\
				\Lr\in\mathcal{P}(\I-\Pi_\sigma),\Ls\in\mathcal{P}(\I-\Pi_\rho),\tr\lr{\Lr\rho}\neq0,\tr\lr{\Ls\sigma}\neq0,\Lr+\Ls\leq\I.
			\end{cases}\right\}.
		\end{align*}
		
		We write the set of measurements satisfying the first, second and third condition are $\mathcal{E}_s^1(\rho,\sigma)$, $\mathcal{E}_s^2(\rho,\sigma)$, $\mathcal{E}_s^3(\rho,\sigma)$ respectively. An arbitrary measurement should satisfy one of the three conditions, so the set of all such measurement is given by union of these three sets.